\documentclass[12pt]{article}
\usepackage[english]{babel}
\usepackage{tabulary}
\usepackage{longtable}
\usepackage{breqn}
\usepackage{graphicx}
\usepackage[left=1.00in,right=1.00in,top=1.5in,bottom=1.5in]{geometry}
\usepackage[symbol]{footmisc}
\usepackage{comment}

\newcommand{\sym}[1]{#1}
\usepackage{changepage}

\usepackage{etoolbox}

\usepackage{siunitx}
\usepackage{threeparttable}
\usepackage{subcaption}
\usepackage{booktabs}
\usepackage{multirow}
\usepackage{tabularx}
\usepackage{setspace}
\usepackage{pdflscape}
\usepackage{amsthm}
\usepackage{bm}
\usepackage[12pt]{moresize}
\usepackage{amsmath}
\mathchardef\mhyphen="2D 

\usepackage{changepage}
\usepackage{ragged2e}
\usepackage{bbm}
\usepackage{dsfont}
\usepackage{float}
\usepackage{nicefrac}
\usepackage{setspace}
\usepackage[titletoc,title]{appendix}
\usepackage{rotating}
\usepackage{microtype}
\usepackage[12pt]{moresize}
\usepackage{anyfontsize}
\usepackage{graphicx}
\usepackage{enumitem}
\usepackage{amssymb}

\makeatletter
\renewcommand*{\@fnsymbol}[1]{\ensuremath{\ifcase#1\or *\or \dagger\or \ddagger\or
    \mathsection\or \| \or \|\or **\or \dagger\dagger
    \or \ddagger\ddagger \else\@ctrerr\fi}}
\makeatother

\theoremstyle{plain}

\newtheorem{proposition}{Proposition}
\newtheorem{lemma}{Lemma}
\makeatother
\providecommand{\lemmaname}{Lemma}
\providecommand{\propositionname}{Proposition}

\usepackage{array}
\usepackage[labelfont = bf, labelsep = period, font = small]{caption}

\usepackage{lineno}

\usepackage[T1]{fontenc}
\usepackage[adobe-utopia]{mathdesign}

\usepackage[table,xcdraw,svgnames]{xcolor}
\usepackage{natbib}  
\usepackage[colorlinks]{hyperref}


\hypersetup{colorlinks = true}
\AtBeginDocument{%
\hypersetup{%
  	linkcolor =     blue!65!black,%
   	anchorcolor=    blue!65!black,%
   	citecolor =     blue!65!black,%
   	filecolor =     blue!65!black,%
   	urlcolor =      blue!5!black,%
   	menucolor =     blue!65!black,%
	}%
	}%

\usepackage{etoolbox}
\makeatletter
\patchcmd{\NAT@citex}
  {\@citea\NAT@hyper@{%
     \NAT@nmfmt{\NAT@nm}%
     \hyper@natlinkbreak{\NAT@aysep\NAT@spacechar}{\@citeb\@extra@b@citeb}%
     \NAT@date}}
  {\@citea\NAT@nmfmt{\NAT@nm}%
   \NAT@aysep\NAT@spacechar\NAT@hyper@{\NAT@date}}{}{}
\patchcmd{\NAT@citex}
  {\@citea\NAT@hyper@{%
     \NAT@nmfmt{\NAT@nm}%
     \hyper@natlinkbreak{\NAT@spacechar\NAT@@open\if*#1*\else#1\NAT@spacechar\fi}%
       {\@citeb\@extra@b@citeb}%
     \NAT@date}}
  {\@citea\NAT@nmfmt{\NAT@nm}%
   \NAT@spacechar\NAT@@open\if*#1*\else#1\NAT@spacechar\fi\NAT@hyper@{\NAT@date}}
  {}{}
\makeatother

\date{\normalsize\today}

\usepackage{verbatim}

\begin{document}

\title{\vspace{-0.5in}{\LARGE The Dynamic Trade-Off of Dual-Class Shares}\noindent\thanks{\fontsize{9.5}{10}\selectfont{We appreciate comments from Alejandro Drexler, Camilo Garcia-Jimeno, Todd Gormley, Christian Hellwig, Wei Jiang, Michelle Lowry, Thomas Philippon, Oliver Spalt, and seminar and conference participants from Bristol, CELS (Emory), Chicago Entrepreneurship Workshop, Chicago Fed (Finance and Micro), Drexel Corporate Governance Conference, EFA (Bratislava), Erasmus Corporate Governance Conference, KAIST-Korea University, Melbourne, Seoul National University, and Texas A\&M, and excellent research assistance from Roger Chen, Jee-Hun Choi, Sophie Croome, Daniel Gallego, Max Gillet, Hassan Ilyas, Kenny Jing, Julia Kharzeev, Dawoon Kim, Eric Kim, Seungsoo Kim, Yejun Kim, Boyao Li, Cindy Lu, Pradeep Muthukrishnan, Felicity Phelan, Ahmad Sabbagh, Ethan Sun, Brian Wickman, and Zakary Yudhishthu. This paper supersedes an earlier paper entitled ``Sticking around Too Long? Dynamics of the Benefits of Dual-Class Voting.'' Any views expressed are those of the authors and not those of the Federal Reserve System. All errors are our own.} \looseness=-1} \\ $\;$ \\} 

\author{%
\vspace{-1.0cm} \\
\begin{tabular}{ccccc}
\href{https://sites.google.com/view/hyunseobkim/}{\normalsize\textsc{Hyunseob Kim}}\thanks{\fontsize{9.5}{10}\selectfont{Federal Reserve Bank of Chicago, Economic Research Department. \href{mailto:hyunseob.kim@chi.frb.org}{hyunseob.kim@chi.frb.org}.}}  & \hspace{0.5cm} & \href{https://sites.google.com/view/doronlevit}{\normalsize\textsc{Doron Levit}}\thanks{\fontsize{9.5}{10}\selectfont{University of Washington, Foster School of Business. \href{mailto:dlevit@uw.edu}{dlevit@uw.edu}.}} & \hspace{0.5cm} & \href{https://www.hkubs.hku.hk/people/roni-michaely/}{\normalsize\textsc{Roni Michaely}}\thanks{\fontsize{9.5}{10}\selectfont{The University of Hong Kong, Faculty of Business and Economics. \href{mailto:ronim@hku.hk}{ronim@hku.hk}.}} \\
\normalsize \textit{Federal Reserve Bank of Chicago} & \hspace{0.5cm} & \normalsize \textit{University of Washington} & \hspace{0.5cm} & \normalsize \textit{University of Hong Kong}
\end{tabular}
\vspace{3mm}
}
 
\makeatletter
\patchcmd{\@maketitle}{\begin{center}}{\begin{adjustwidth}{-1in}{-1in}\begin{center}}{}{}
\patchcmd{\@maketitle}{\end{center}}{\end{center}\end{adjustwidth}}{}{}
\makeatother
\maketitle

\renewenvironment{abstract}
 {\small
  \begin{center}
  \bfseries \abstractname\vspace{-.5em}\vspace{0pt}
  \end{center}
  \list{}{%
    \setlength{\leftmargin}{7mm}
    \setlength{\rightmargin}{\leftmargin}%
  }%
  \item\relax}
 {\endlist}

\begin{abstract}

\noindent
Dual-class shares allocate control to founders whose firm-specific investments drive firm value but separate control from ownership, raising agency costs. We analyze this trade-off dynamically. Using new data on US dual-class firms spanning 52 years and difference-in-differences designs, we show that valuations rise following dual-class recapitalizations but decline over time, whereas innovative output increases persistently. These effects are concentrated in industries with greater firm-specific investments. We find corresponding results for stock unifications. Investment by mature dual-class firms is less sensitive to opportunities and voting premia increase with maturity. Our results support dynamic treatment effects and yield new policy implications.

\end{abstract}
\thispagestyle{empty}

\newpage
\pagenumbering{arabic}  
\fontsize{12}{14}\selectfont
\setstretch{1.3}

\section{Introduction}

A rising share of newly public firms has adopted dual-class shares in the US over the past several decades---Figure \ref{fig:IPOs} shows that over 40\% of the initial public offerings (IPOs) in 2025 had such structures. Accordingly, dual-class firms, including many mega firms such as Alphabet, represented over 9\% of US stock market value from 2013--2022---a more than fourfold increase from just over 2\% in the early 1980s.\footnote{This number may increase with SpaceX's dual-class IPO. See, e.g., \href{https://www.wsj.com/business/how-spacexs-ipo-cements-elon-musks-grip-on-the-company-90ef1ddd}{{\color{blue}{``How Supervoting Shares Tighten Musk's Iron Grip on SpaceX,''}}} \emph{The Wall Street Journal}, May 22, 2026.} Consequently, it has become crucial to understand how the allocation of voting rights---the choice between dual-class and single-class shares in particular---affects economic efficiency. When Meta's market value fell in 2022, for example, investors blamed its dual-class shares for the firm's poor performance.\footnote{See, e.g., \href{https://www.economist.com/leaders/2022/11/03/big-tech-big-trouble}{{\color{blue}{``Facebook and the Conglomerate Curse,''}}} \emph{The Economist}, November 3, 2022.} The SEC's then-Commissioner Robert \citet{jackson2018} had urged US stock exchanges to consider limiting use of dual-class shares. Consistent with increasing controversy surrounding the impacts of dual-class shares, the number of shareholder voting proposals related to dual-class structures between 2021 and 2023 nearly tripled relative to the annual average of the preceding decade.\footnote{Data on shareholder proposals are obtained from ISS Voting Analytics, covering Rule 14a-8 proposals targeting the elimination of dual-class structures.} If dual-class shares lead to poor performance through increased agency costs---as is often argued---why are they widely and increasingly used by young firms? What are the impacts of policies that limit their use? In this paper, we construct a new panel dataset on US public firms' voting and ownership structures that spans a half-century to study these increasingly important questions.\looseness=-1 

To organize our empirical analysis, we develop a theoretical framework in which dual-class shares encourage the founder's firm-specific investment (``effort''), but the value added from the investment is eroded by private benefits that rise with firm maturity---creating a dynamic trade-off. We show private benefits are an important source of incentives for the founder's effort. Because dual-class shares protect the founder's private benefits \emph{ex post}, he has stronger incentives to exert firm-specific effort \emph{ex ante}. In contrast, single-class shares limit his ability to extract private benefits down the road---because shareholders cannot commit not to replace him once his effort is sunk---which in turn weakens incentives. This holdup problem (e.g., \citet{grossman_hart1986}; \citet{hart_moore1990}) is mitigated by dual-class shares.\footnote{See, e.g., \citet{iyer2015ex, iyer2024contracting} for evidence on the role of holdup in contracting.} In essence, dual-class structures allocate control to the founder whose firm-specific investment is important for the firm's success.\looseness=-1 

The model predicts a transitory dual-class premium: young dual-class firms have higher value than comparable single-class firms, ceteris paribus, but this initial value gap shrinks and ultimately reverses as firms mature. We test these model predictions using a stacked difference-in-differences (DD) event-study design (e.g., \citet{gormley2016managers}; \citet{cengiz2019minwage}), separately examining two events: (i) dual-class recapitalizations, in which single-class firms adopt a dual-class structure, and (ii) dual-class stock unifications, in which dual-class firms transition to a single-class structure. To this end, we construct a new event-study dataset spanning 1971--2022, drawing from public disclosures, news sources, \emph{Moody's Manuals}, and CRSP--Compustat. 

We first show that these changes in corporate control structure are common: among about 1,150 firms that have ever had dual-class shares, about 230 and 350 firms have recapitalized with and unified dual-class shares, respectively. We use Tobin's \textit{q} as our measure of firm valuation (e.g., \citet{cronqvist2003agency}; \citet*{gompers2010extreme}) and patent outputs to proxy for the founder's firm-specific investment.\footnote{This measure assumes that the founder's investment, as reflected in patent outputs, is not easily redeployable outside the firm, leading to a holdup problem. See \cite{rajan2012} for an argument that entrepreneurship inherently involves producing differentiated products with little outside market value. } We find that in the years leading up to a recapitalization or unification, both Tobin's $q$ and patent output exhibit parallel pre-trends between switching and non-switching firms. In the year firms recapitalize into a dual-class structure, Tobin's $q$ rises by 0.21 relative to the control group of firms remaining single-class, and this valuation premium remains largely significant in the next five years. The premium then starts declining---by ten years post-recapitalization, the valuation gap essentially disappears. Also consistent with our framework, a recapitalization leads to a persistent increase in patent output---our proxy for founder effort. We find corresponding results for stock unifications: Tobin's $q$ \emph{declines} initially before reverting to its pre-unification level, while patent output experiences a persistent decline, all relative to firms staying dual-class. \looseness=-1

We bolster our main difference-in-differences analysis of the recapitalization by employing another DD design that exploits a regulatory difference between the New York Stock Exchange and other exchanges in the early-1980s on listing dual-class shares, combined with rising external governance forces then. We find a similar magnitude increase in valuation post-recapitalization, as well as no evidence of differential pre-trends. \looseness=-1

Complementing the analysis of valuations, we examine how the market's perceived value of dual-class shares varies with firm maturity. We find that when single-class firms younger than or equal to 12 years (``young firms'')---the sample median age for dual-class firms---announce a recapitalization, the average cumulative abnormal return (CAR) is 2.9\%.\footnote{To streamline the exposition, we refer to firms with ages less than or equal to the median as ``young'' and greater than the median as ``mature'' throughout the paper. Our empirical analysis employs alternative measures of firm maturity including a set of age indicators.} This positive market reaction early in the life cycle is consistent with the benefit of stronger founder incentives. However, when single-class firms aged greater than 12 years (``mature firms'') announce a recapitalization, the average CAR is 3.5 percentage points lower. This finding suggests that the market views dual-class recapitalizations as increasingly value-reducing as firms mature, plausibly driven by increasing private benefits extraction.\footnote{Table \ref{tab:voting_premium0} provides evidence that private benefits of control increase with firm maturity.} Analogously, we find that the average CAR for the announcement of a stock unification is 5.2 percentage points higher for mature firms compared to young firms. \looseness=-1

Crucially, our framework allows us to assess whether the dynamic effects of dual-class shares reflect treatment or selection. The model shows that if selection were to drive the dynamics, private benefit extraction should grow at a slower rate and founder effort would be \emph{smaller} in dual-class firms---the opposite of the prediction under treatment. Intuitively, when his private benefits are expected to grow slowly, the founder does not need to commit to good behavior through single-class shares, and so chooses dual-class shares.\footnote{In general, for selection to generate the transitory dual-class valuation premium, its effect must be time-varying. We address another type of selection related to departure of the founder in Section \ref{sec:q_robustness}.}\looseness=-1

We find several pieces of evidence consistent with treatment effects of dual-class shares rather than selection. First, the patent-output results described above indicate that founders of dual-class firms make greater firm-specific investments than their single-class counterparts. Second, to measure private benefits of control, we examine firms' investment decisions. The idea is that entrenched founders have stronger preferences for preserving what they have invested in than shareholders do (e.g., \citet{bertrand2003enjoying}). Consistent with this private benefit being more pronounced in mature dual-class firms than in single-class firms, we find that investment by mature dual-class firms is less sensitive to investment opportunities than that of their single-class counterparts, when firms are near a threshold to disinvest as proxied by low sales growth. Third, we find that the effects of dual-class recapitalization and unification on Tobin's $q$ and innovation are concentrated in industries where investments are more likely to be firm-specific---consistent with the core mechanism of treatment effects in the model. Fourth, the pre-trends in these outcomes are parallel between switchers and non-switchers. \looseness=-1

Given growing private benefits with maturity, founders of dual-class firms may have little incentive to unify multiple classes as their firms mature. Consequently, some academics and policymakers argue that dual-class shares should include a sunset provision---under which the passage of time or the transfer of the founder's ownership triggers unification (e.g., \citet{bebchuk2017untenable}; \citet{jackson2018}). We analyze the impact of such provisions by comparing unifying dual-class firms with and without sunset provisions. We find that unification has no discernible effect on valuation for firms with sunsets, whereas it has the negative effect described above for firms without a sunset. A plausible interpretation is that sunset provisions weaken founders' \emph{ex ante} incentives by reducing expected private benefits---which consequently dampens the negative valuation effect of unification. This highlights the tight dynamic link between private benefits, incentives, and firm value, as shown in our framework. Our analysis suggests that policymakers should weigh the potential costs of sunset provisions---weaker incentives for firm-specific investments---which have been largely overlooked in the current debate on dual-class structures. Importantly, recent research indicates that the nature of firms' investment and production has become more specific rather than redeployable across firms (e.g., \cite{ekerdt_wu2025}; \cite{barry2026mobility}). \looseness=-1

This paper contributes to the literature on the allocation of control rights. Empirical research has primarily examined the average cross-sectional relationship between voting rights, private benefits, and firm valuation (e.g., \citet{claessens2002disentangling}; \citet*{masulis2009agency}; \citet{gompers2010extreme}).\footnote{\citet{cronqvist2003agency} estimate the within-firm effect of controlling shareholders' voting rights on valuation for Swedish firms. See, e.g., \citet{adams2008one} for a review of this literature.} \citet{DEANGELO198533} pointed out that dual-class shares can incentivize managers to undertake firm-specific investment. However, to our knowledge, no prior work has characterized how the incentive effects of dual-class shares relate to the dynamics of private benefits, firm productivity, and value. \looseness=-1 

Our paper differs from this existing work in two key and related ways. First, we develop a theoretical model that integrates the incentive and private benefit effects of dual-class shares over the firm's life cycle, generating a rich set of \emph{dynamic} predictions regarding private benefits and firm productivity. Second, we provide long-run evidence consistent with these dynamic predictions by estimating the effects of dual-class recapitalization and stock unification using a difference-in-differences (DD) event-study design. This analysis is enabled by our new panel and event datasets spanning 1971--2022. These datasets on dual-class firms will be available on the authors' websites. \looseness=-1
A recent paper by \citet*{cremers2024life} documents evidence consistent with a transitory dual-class premium for IPO firms from 1980--2019 using pooled OLS (without firm fixed effects). The main empirical difference is that we estimate the effects of dual-class shares using \emph{changes} in governance structures in a DD event-study design with explicit control groups, whereas they examine the time-series evolution of single- and dual-class firm valuations since IPO. Our research design, including evidence on parallel trends, and our model jointly allow for a treatment effect interpretation of our estimates. We also offer new policy implications that are different from those of existing work, including \citet{cremers2024life}. \looseness=-1

More broadly, this paper contributes to the growing literature on the dynamic effects of corporate governance and control, including work on boards of directors (\citet*{field2013busy}); \citet{field2022}) and takeover defenses (\citet*{johnson2022lifecycle}). Relative to these papers focusing on the post-IPO evolution of firm and governance outcomes, this paper exploits governance switches, thanks to the new long-run dataset. In addition, dual-class shares are arguably among the most effective mechanisms for consolidating control rights (e.g., \citet{ruback1988}; \citet{gompers2010extreme}). Thus, dual-class structures provide a particularly powerful context in which to study the dynamic consequences of governance arrangements for firms and their controllers.\footnote{Our analysis focuses on the effects of dual-class relative to single-class shares, setting aside alternative contractual mechanisms. Implicitly, we recognize that under incomplete contracting, compensation schemes cannot fully specify managerial behavior, leaving control arrangements, such as dual-class shares, a meaningful role to play.} \looseness=-1

\section{Theoretical Framework and Predictions}

\label{sec:theory}In this section, we present a theoretical framework that motivates our empirical analysis.

\subsection{Setup}

Consider a firm run by a risk-neutral founder with an infinite horizon. At
the outset, the founder holds a fraction $\alpha \in (0,1]$ of the shares,
with the remainder owned by outside investors. At $t=0$, the founder (i)
chooses the firm's governance structure, which is publicly observable; (ii)
exerts unobservable, non-contractible firm-specific effort $x\geq 0$ at a private cost $%
\frac{1}{2}\kappa x^{2}$, which increases firm productivity and determines
its expected cash flows $y_{t}=xe^{gt}$ in each period $t>0$, where $g>0$; and (iii)
sells $ \alpha \left( 1-\Lambda \right)$ shares at the prevailing market
price, where $\Lambda \in \left( 0,1\right) $.\footnote{%
Assuming $\Lambda < 1$ implies that the founder partially internalizes outside shareholder value for $t > 0$, either because he actually sells shares, expects future liquidity needs, or anticipates that the firm may issue shares at the prevailing market price to raise capital or compensate employees.} Throughout, and under either governance structure, the
stock price at time $t$ equals the present value of all future dividends.
The time $t=0$ may be interpreted as the time of a dual-class recapitalization or stock unification, or the IPO.

We assume an agency conflict between the founder and shareholders.
Specifically, before distributing dividends, the founder diverts a fraction $%
b_{t}=\phi \left( 1-e^{-\tau t}\right) $ of cash flows as private benefits,
where $\phi \in (0,1]$ and $\tau >0$, leaving shareholders a dividend $%
\left( 1-b_{t}\right) y_{t}$. These private benefits can be understood as
the avoidance of personally costly actions that would benefit shareholders
or the pursuit of negative-NPV projects that generate perks for the founder.
The founder's utility from diversion is $\gamma b_{t}y_{t}$, with $\gamma
\in \left( \alpha \Lambda ,1\right) $. The upper bound ensures diversion is
socially wasteful, while the lower bound guarantees that diversion remains
individually rational despite the diluted value of shares that the founder keeps. The market discounts the shares he sells, $\alpha (1-\Lambda)$, by the expected diversion. The founder thus pays the agency costs upfront through lower sale proceeds.

Crucially, $b_{t}$ increases over time: the longer the founder holds office, the more
entrenched he becomes and the more he diverts (e.g., \citet*{gkl2020}), with 
$\tau >0$ governing the speed of entrenchment. Section \ref%
{sec:voting-premium} provides supporting evidence for this premise from voting premium
dynamics.

Shareholders may dismiss the founder and replace him with a new manager who
neither diverts nor adds value, so dividends equal $y_{t}${} under new
management. The governance structure determines majority voting control:
under single-class shares (one share--one vote), shareholders may remove the
founder at any time; upon dismissal, he retains his equity stake but loses
the ability to divert.\footnote{%
Alternatively, single-class shares may enable shareholders to monitor and
directly constrain diversion. Implicitly, we assume that the founder does not hold majority control under single-class shares, that is, $\alpha \Lambda < 0.5$.} Under dual-class shares, the founder holds
permanent majority control and cannot be removed. Appendix \ref%
{appendix:theory} extends the analysis to richer governance structures,
including dual-class shares with time-based sunset provisions.

\subsection{Analysis}

\subsubsection{Founder's productivity and treatment effects}

We begin by analyzing the model under the assumption that the firm's
governance structure is exogenous. Agency
costs arise in our model because the founder can extract private benefits that grow over
time. Crucially, once the founder's effort is sunk, shareholders have an
incentive to use their voting power to remove him. This \emph{ex post}
expropriation risk weakens the founder's \emph{ex ante} incentives to exert
effort---a classic holdup problem.\footnote{%
The mechanism parallels that in \citet*{burkart1997} who study holdup in
the presence of large shareholders.}

The founder's incentives to exert effort derive from two sources. First,
higher effort increases firm value and, therefore, the value of the
founder's equity stake. Second, because the founder's expected private
benefits of control rise with firm value, they too increase with effort. As
we show formally in Appendix \ref{appendix:theory}, Lemma \ref{res:effort}, the latter channel implies that
equilibrium effort is higher under dual-class than under single-class
shares, that is, $x_{\text{dual}}^{\ast }>x_{\text{single}}^{\ast }$. We provide evidence consistent with dual-class shares encouraging the founder's effort in Section \ref{sec:innovation}. The
central cost of single-class shares is thus the shareholders' inability to
commit not to replace the founder after effort has been exerted.

Let $p_{\text{dual}}^{\ast }\left( t\right) $ and $p_{\text{single}}^{\ast
}\left( t\right) $ be the equilibrium share prices at period $t\geq 0$ of a
dual-class firm and a single-class firm, respectively. The following
proposition presents our prediction on the effect of dual-class shares on firm
value over maturity (all proofs are in Appendix \ref%
{appendix:theory}).

\begin{proposition}
\label{price_dynamic}%
\hspace{-2.2 mm}%
\textit{ Suppose }$\gamma -\alpha \Lambda >\alpha \Lambda $.\textit{\ There
exist }$\underline{\phi }<1$\textit{\ and }$\overline{\phi }>0$\textit{\
where }$\underline{\phi }<$\textit{\ }$\overline{\phi }$\textit{, such that
if }$\phi \in (\underline{\phi },\overline{\phi })$\textit{\ then }$p_{\text{%
dual}}^{\ast }(t)>p_{\text{single}}^{\ast }(t)$\textit{\ if and only if }$%
t<T^{\ast }$\textit{\ for some }$T^{\ast }\in (0,\infty )$\textit{. }
\end{proposition}

Proposition \ref{price_dynamic} shows that as long as the founder places a larger weight on his ``net'' private benefit than his stake in the firm (i.e., $\gamma -\alpha \Lambda >\alpha \Lambda $),\footnote{Absent this condition, the positive effect of dual-class shares on the founder's incentives is outweighed by the diversion of private benefits, so firm value is always lower under dual-class shares than single-class shares.} and the cost of private
benefits extraction on shareholders is not too large (i.e., $\phi <\overline{%
\phi }$), firm value is higher under dual-class than single-class shares
when the firm is relatively young. Over time, however, this premium
shrinks and eventually reverses. We refer to this dynamic pattern as the 
\textbf{treatment effect} of dual-class shares on firm value: the
causal effect that would arise if governance structure were randomly assigned. Online Appendix Figure \ref{fig:prop1} visualizes the treatment effect by plotting how $p_{\text{dual}}^{\ast }\left( t\right) -p_{\text{single}}^{\ast
}\left( t\right) $ changes over time.

Intuitively, dual-class shares strengthen the founder's incentives by
assuring him that he will retain control and continue to extract private
benefits. In the early years, this effect dominates, and the firm commands a
valuation premium. Over time, however, accumulating agency costs erode the
gains from heightened effort. Provided the agency cost is not too small
(i.e., $\underline{\phi }<\phi $), a dual-class firm eventually trades at a
discount relative to an otherwise identical single-class firm, i.e., $p_{%
\text{dual}}^{\ast }(t)<p_{\text{single}}^{\ast }(t)$ for $t>T^{\ast }$. In
this way, Proposition \ref{price_dynamic} rationalizes the key empirical
findings we document below.

\paragraph{Governance changes, sunset provisions, and voting premium.}

When a single-class firm unexpectedly recapitalizes to a dual-class
structure, two opposing forces shape the market reaction. On one hand, the
founder's strengthened incentives raise firm productivity, exerting upward
pressure on the share price. On the other hand, the share price falls by the
present value of the private benefits the founder expects to extract going
forward. For young firms, where the present value of future private benefits
is relatively small, the incentive effect dominates and the abnormal return
is positive; as the firm matures and $b_{t}${} rises, the private benefit
effect increasingly dominates, turning the abnormal return negative. The
mirror image holds for an unexpected unification: the loss of the founder's
incentives depresses firm productivity, while shareholders simultaneously
gain the present value of private benefits the founder would otherwise have
extracted. For young firms the productivity loss dominates, yielding a
negative abnormal return, whereas for more mature firms the recovery of
private benefits dominates, yielding a more positive one. As we show formally in Online
Appendix \ref{appendix:theory_online}, these predictions imply that the abnormal
return decreases (increases) as firms mature following a dual-class recapitalization
(unification). This implication is consistent with both the long-run dynamic
difference-in-differences estimates for governance switchers and the short-run announcement returns evidence in
Section \ref{sec:q}.

In Online Appendix \ref{appendix:theory_online}, we also analyze dual-class shares with
time-based sunset provisions, which allow shareholders to fire the founder
only $n$ periods after the IPO. We show that the abnormal return following an unexpected unification is higher (e.g., less negative) for dual-class firms with sunset provisions than for those without,  which is consistent with the
evidence in Figure \ref{fig:appendix_event_study_sunsets}. Intuitively, founders of dual-class firms with sunset provisions exert less effort than those without such provisions. As a result, the reduction in effort---and the corresponding loss in firm value---following unification is smaller, leading to a less negative abnormal return.\footnote{This pattern is pronounced for young firms, in which the effort channel dominates and private benefit extraction remains relatively small.}

Lastly, our model enables us to study how the difference between the value
of voting and non-voting shares (referred to as the voting premium) changes
over maturity. In our framework, the difference between the value of a
voting share (which is held by the founder) and a non-voting share (held by
shareholders) is the present value of the founder's future private benefits
of control. In Online Appendix \ref%
{app:vp}, we show that the voting premium is positive and increases over
maturity, consistent with the evidence in Section \ref{sec:voting-premium}.

\subsubsection{Founder's governance choices and selection effects\label%
{sec:theory_selection}}

Our second result characterizes the founder's choice of governance structure.

\begin{proposition}
\label{optimal_gov}%
\hspace{-2.2 mm}%
\textit{ Suppose }$\gamma -\alpha \Lambda >\alpha \Lambda $.\textit{\ There
exists }$\phi ^{\ast }>\overline{\phi }$\textit{\ such that the founder
prefers dual-class shares over single-class shares at $t$ = 0 if and only if $\gamma\geq\alpha$ or}%
\begin{equation}
\phi <\phi ^{\ast }.  \label{eq:eq2}
\end{equation}%
\textit{Moreover, }$\phi ^{\ast }$\textit{\ increases in }$\gamma $\textit{\
and decreases in }$\tau $\textit{.}
\end{proposition}

To understand Proposition \ref{optimal_gov}, recall that dual-class shares
enable the founder to extract private benefits of control. Were this the
only consideration, the founder would always prefer a dual-class structure.
However, the founder also sells shares to outside investors who do not value private benefits, and hence partly
internalizes how the governance structure affects firm value for shareholders. By Proposition \ref{price_dynamic}, when $\phi <
$ $\overline{\phi }$, a dual-class firm is valued with a premium in $t=0$, so
the founder faces no trade-off: dual-class shares allow him to both consume
private benefits and maximize the market value of the shares he sells. 

However, when $\phi >\overline{%
\phi }$, the agency friction is severe enough that a dual-class structure
instead carries a valuation discount, confronting the founder with a
trade-off between extracting private benefits of control and the resulting
reduction in proceeds from selling shares. Proposition \ref{optimal_gov} characterizes how this trade-off is resolved. If the agency friction is not too severe, i.e., $\overline{\phi }<\phi <\phi ^{\ast }$, the utility from private benefits dominates and the founder still prefers dual-class shares. A low deadweight loss from diversion (high $\gamma $) and slow entrenchment of the founder (low $\tau $) reinforce this choice by increasing $\phi ^{\ast }$. Conversely, if agency costs are sufficiently large, i.e., $\phi >\phi ^{\ast}$, the negative impact on sale proceeds dominates, and the founder chooses single-class shares. Proposition \ref{optimal_gov} also shows that the founder chooses single-class shares only if $\gamma <\alpha $. This condition implies that private benefits he enjoys under dual-class shares do not compensate him for the losses on his entire initial stake. 

In Appendix \ref{appendix:theory} we show that if shareholders were to make the governance
decision, Proposition \ref{optimal_gov} continues to hold, with the threshold
$\phi^{\ast}$ replaced by $\overline{\phi}$. Since
$\overline{\phi} < \phi^{\ast}$, shareholders are less likely to adopt dual-class shares: they do not internalize the consumption of
private benefits, but only the positive effect on the founder's incentives.

Proposition \ref{optimal_gov} highlights that the choice of the firm's
governance structure is likely non-random and driven by factors that might
also determine how firm value evolves over time. Crucially, the model allows us to analyze whether the
dynamic pattern in Proposition \ref%
{price_dynamic} can also be attributed to a \textbf{selection effect}. The next, final result shows that in general, a selection effect alone can generate the
dynamic pattern but only under specific
conditions.\footnote{We focus attention on one-dimensional firm-level heterogeneity.}$^,$\footnote{Proposition \ref{selection_experiment} holds regardless of whether the founder
or shareholders choose the governance structure.}

\begin{proposition}
\label{selection_experiment}%
\hspace{-2.2 mm}%
\textit{ A selection effect can generate the dynamic pattern in Proposition \ref{price_dynamic} only if a selected dual-class firm has lower }$\tau $\textit{, the rate at which its founder becomes entrenched over time. In this case, its founder (i) extracts smaller private benefits, especially when the firm is young, and (ii) exerts less effort than an otherwise identical founder of a selected single-class firm.}
\end{proposition}

Note these predictions for the founder's private benefits extraction and effort under the selection are the opposite of those under a treatment effect above. When the smaller private benefits effect dominates for young firms (implying a higher firm value for dual-class firms), while the less effort effect dominates for mature firms (a lower firm value for dual-class firms), selection alone could generate the transitory dual-class valuation premium. Appendix \ref{appendix:theory} shows that such cases exist. Intuitively, for a selected dual-class firm with lower $\tau $, private benefits are smaller early on but marginal increases over time are larger when it is mature, contributing to a relative decline in value. Once the founder is fully entrenched in the long run, i.e., $b_t\to \phi$, the negative effect of smaller effort on the dual-class firm dominates.

In contrast, when selection is due to other parameters than $\tau $, the selection
effect is either (i) constant over time---i.e., the value of a selected
single-class firm is higher or lower than that of a selected
dual-class firm throughout; or (ii) the opposite from Proposition \ref{price_dynamic}---
i.e., dual-class firms' value is higher than single-class firms' when they
are mature, as opposed to when they are young (see the proof of Proposition \ref{selection_experiment} in Appendix \ref%
{appendix:theory}).

\subsection{Mapping to reduced-form analysis.}
\label{sec:theory_mapping}
The model posits that dual-class shares allocate control to the founder whose firm-specific investment is crucial for the firm's value. The key mechanism is that dual-class structures protect the founder from losing private benefits \emph{ex post}, providing him with strong incentives to exert firm-specific effort \emph{ex ante}. This model has two major predictions. First, young dual-class firms have higher value relative to otherwise identical single-class firms but this premium shrinks as they mature (Proposition \ref{price_dynamic}). We test this implication in Section \ref{sec:q}. Second, if this transitory dual-class premium is due to treatment (selection), then both the incentives for effort and private benefits are greater for dual-class (single-class) firms. In addition, under treatment (selection), the difference in private benefits between dual-class and single-class firms grows (shrinks) as firms mature (Proposition \ref{selection_experiment}). We test these implications in Sections \ref{sec:innovation} (for effort) and \ref{sec:sesitivity} (private benefits). Section \ref{sec:redeployability} tests a core mechanism of the model by examining the heterogeneity by the specificity of investments. Finally, Section \ref{sec:voting-premium} provides evidence for the model premise that private benefits increase over firm maturity.
\vspace{-2mm}
\section{Data Construction and Descriptive Evidence} \label{sec:data}

\subsection{Constructing panel of dual-class firms including governance changes} \label{subsec:dcfirms}

We construct a panel dataset of dual-class firms in the United States from 1971 to 2022. We hand-collect information on dual-class firms drawing from several sources including \emph{Moody's Manuals}, the SEC's EDGAR database, and news articles, as well as merging to CRSP--Compustat data. For the nine-year period 1994--2002, we build upon data from \cite{gompers2010extreme}\footnote{We thank Andrew Metrick for making the dataset on dual-class firms available on his website.} but collect further information necessary for our key empirical analysis such as difference-in-differences event studies. The resulting dataset represents (i) the longest panel on US dual-class firms, covering more than a thousand unique dual-class firms over a half-century and to our knowledge, (ii) the most extensive dataset on both dual-class recapitalization and unification events. We present several highlights of the dataset crucial for the empirical analysis below. Online Appendix \ref{appendix:data} provides details of the data construction.

First, key variables in our dataset include voting rights, such as votes per share and rights to elect directors, and the number of shares outstanding by share class. Importantly, it provides comprehensive history of companies that we identify have ever been dual-class firms by collecting information on the firm's first dual-class year through either IPO or recapitalization, as well as the firm's last dual-class year through either delisting or unification.

Second, and related, for our difference-in-differences analyses we carefully verify whether and when a dual-class firm has unified its share classes, or a single-class firm has recapitalized with dual-class shares.\footnote{For both recapitalizations and unifications, firms go through standard processes to change their governance structure. Recapitalizations must be approved by shareholders, typically after the board of directors approves such a change. Similarly, unifications are typically approved by shareholders across both stock classes.} In the years that the SEC's EDGAR database is available (1994 to 2022), we collect this information from 10-K filings, and in earlier years we use annual \emph{Moody's Manuals}, which list all equity classes of a firm and usually discuss recapitalization or unification, if it happened. We also collect the exact dates of these switches when available, as reported in the 10-K or \emph{Moody's}, and verify ambiguous cases by searching for online company histories and news articles that document the governance changes. Finally, we cross-reference these switches using CRSP data on the firm's shares outstanding by class and delisting codes. 

Third, to study short-term stock market reactions to dual-class recapitalizations, we collect the dates of firms' first public announcements in our own news searches using Factiva and LexisNexis. When our own search does not yield news articles on recapitalizations, we supplement with information from \cite{partch1987creation} and \cite{jarrell1988dual} (covering 1971--84 and 1976--87, respectively).\footnote{For cases that overlap in the two papers, we prioritize event dates from \cite{jarrell1988dual}, who verify and update some dates from \cite{partch1987creation}.} Importantly, we exclude events that are confounded by announcements of other major corporate events, such as mergers and acquisitions, spinoffs, and emergence from bankruptcy. For unification announcements, on which we are not aware of existing US data,\footnote{\cite{ang1989dual} examine dual-class recapitalization and unification cases for UK firms in daily and monthly return event studies, and show positive effects of the former and mixed effects of the latter on stock price.} we collect our own data by examining whether each dual-class firm's terminal year is due to a share unification using the sources mentioned above. We then search for news articles that publicly announce these unifications for the first time in major news outlets. \looseness=-1

\subsection{Additional data sources} \label{sec:addl_data}
\textbf{Patents.} A key outcome variable of interest in our analyses is patent citations, which capture innovative output as a proxy for the effort that a firm's founder exerts. We combine two patent datasets: data compiled by \cite{kogan2017technological}, and data directly from the US Patent Office (USPTO) combined with a patent-firm linkage created by \cite{ma2025obsolescence}. We refer to these datasets as KPSS and SM, respectively.\footnote{We thank Noah Stoffman for making the KPSS patent dataset available, and Song Ma for sharing his patent-firm linkages.} We aggregate the patent-permno-level KPSS data at the firm (i.e., gvkey) level, combining patents across multiple permnos in a given year. However, KPSS is disproportionately missing dual-class firms, as the authors drop patents that match with multiple permnos. Thus, we complement with the SM data between 1975 and 2020 (the years in which it is most complete) as follows: for firms recorded as dual-class at any point in our data, if KPSS does not record patents for a given firm-year, we fill in patent information using the SM-USPTO data. This combined data improves our measurement of patenting and citations for dual-class firms. In our empirical analysis, we use a forward-looking measure that combines the number of citations one, two, and three years in the future as an outcome variable (e.g., \citet{bernstein2015innovation}).

\textbf{Ownership structure.}
To calculate the Shapley value of ownership structure---a measure of the control-contest probability (e.g., \cite{rydqvist1987pricing}) and an important variable for analysis of the voting premium in Section \ref{sec:voting-premium}---we construct ownership variables for corporate insiders and institutional investors. We calculate ownership by corporate insiders, defined as the firm's executives and directors (e.g., \cite{fabisik2021}), using the SEC's Form 3, 4, and 5 as provided by Refinitiv. (Form 3 is the initial filing of ownership, Form 4 reports ownership changes, and Form 5 is an annual summary of ownership.) We use the total shares held by an insider following the last transaction to capture his or her latest holdings in a given year. The resulting insider ownership data span from 1988--2022. We forward- and backward-fill an insider's holdings for a given firm if a firm-insider-year's ownership is missing.\footnote{Form 4 is only filed if an insider transaction takes place, and Form 5 is generally sparsely populated (i.e., an insider might not have a Form 5 at the end of every year). We stop the backward-filling in 1980 when the institutional ownership data begin.} We calculate institutional ownership for a given firm-institution using the SEC's Form 13F provided by Refinitiv and WRDS following \cite{lewellen2022institution}. Form 13F is a mandatory report filed by institutional investors with at least \$100 million worth of assets under management. The dataset spans from 1980--2022 and records the total number of shares held by all reporting institutional investors in a given firm. We focus on firm-institutions with at least 5\% of a firm's equity (``block'' institutional ownership), which likely have a meaningful influence on the firm.

\subsection{Sample selection} \label{sec:sample}
We merge the panel and event datasets on dual-class firms with CRSP--Compustat. We require that firm-year observations have the following variables constructed: Tobin's \textit{q}, book assets, market leverage, research and development (R\&D) expenses scaled by lagged book assets, asset tangibility, return on assets (ROA), payout ratio, sales growth rate, and SIC codes. We impute missing values of R\&D to zero (for a similar adjustment see \citet{brav2018how}). A key variable for our analysis is firm age. Similar to \citet{pastor2003}, we define the firm's birth (hence age is zero) as public firm as the first year in which the firm appears in Compustat  with a valid share price. Appendix \ref{appendix:variables} shows the definitions of the variables. 

We exclude firms in the financial (SIC 6000-6999), utilities (SIC 4900-4999), and unclassified (SIC 9900-9999) industries. To mitigate the influence of outliers, we exclude firm-years with book assets less than \$20 million in 2000 constant dollars (adjusted using CPI) and Winsorize potentially unbounded variables at the 2$^{nd}$ and 98$^{th}$ percentiles. We obtain quantitatively similar results by Winsorizing at an alternative level such as 1\%. These sample selection criteria lead to a final sample of 9,950 dual-class firm-years across 1,151 unique firms from 1971 to 2022. By adding 150,752 single-class firm-years, we have 160,702 firm-years in total.

Figure \ref{fig:DC_counts} shows the number of dual-class firm-years with the aforementioned variables from 1971 through 2022, in comparison with the Compustat universe. The fraction of firms with dual-class shares ranged between 2.8\% and 3.4\% before the early 1980s, increased to 7.2\% to 7.4\% in the early 1990s, and has stayed between 6.0\% and 8.0\% since then. The rapid increase in the number of dual-class firms during the 1980s reflects many firms adopting dual-class shares in the period of heightened hostile takeover activities (e.g., \cite{jarrell1988dual}). Section \ref{sec:event_study_add} exploits these dual-class recapitalizations in the 1980s, along with different rules for listing dual-class shares between exchanges, to bolster our main DD design.

\subsection{Descriptive evidence: Valuation of dual-class vs. single-class firms} \label{sec:Firm-maturity}

Table \ref{tab:descriptive_stats} presents descriptive statistics for the samples of dual-class and single-class firm-years in the full panel. Dual-class firms tend to be larger, older, and more highly levered, and to have higher ROA than single-class firms. They also have lower capital expenditures and R\&D expenses than single-class firms. 

The table shows that on average, Tobin's \textit{q} for dual-class firms is 0.28 higher than for single-class firms, a statistically significant difference. In Appendix \ref{appendix:average}, we further examine this relation in a regression framework. Appendix Table \ref{tab:appendix_c1} shows that valuations are significantly higher for dual-class firms than otherwise comparable single-class firms in the cross-section, controlling for firm-level covariates and industry-by-year fixed effects. We then examine how this cross-sectional valuation gap evolves over time since firms' IPOs. Appendix Figure \ref{fig:appendix_c1_q_dynamics} shows that the difference in Tobin's $q$ is significantly positive for most of the first decade post-IPO but declines through age 50, when the difference becomes significantly negative. This initial, descriptive evidence is consistent with a ``transitory dual-class premium,'' as well as evidence in \citet{cremers2024life}, albeit on an extended sample. Nonetheless, it is likely affected by both observable and unobservable heterogeneities between dual-class and single-class firms that may vary over time. To more credibly test the model implication, we turn to the difference-in-differences analysis of changing governance structures in Section \ref{sec:dynamics}. \looseness=-1
 B, respectively.

\section{The Dynamic Effects of Dual-Class Shares} \label{sec:dynamics}

\subsection{Effects on valuation} \label{sec:q}
\subsubsection{Research design: Difference-in-differences event study} \label{sec:event_study}
Our main empirical analysis uses a stacked dynamic difference-in-differences (DD) event-study design (e.g., \citet{gormley2016managers}; \citet{cengiz2019minwage}). For recapitalization, the treated group includes firms that recapitalize with dual-class shares, whereas the baseline control group includes firms that remain single-class. For unification, the treated group includes firms that unify share classes, whereas the baseline control group includes firms that remain dual-class. For each event type, we stack the treated and control firms by event year. We construct alternative, matched control groups using pre-event observable firm characteristics below. The key identifying assumption for a treatment effect interpretation of the DD estimates is parallel trends---the outcomes of treated and control firms would have followed parallel paths in the absence of a treatment---which we test using pre-event observations. 

We estimate the following difference-in-differences equation for dual-class recapitalization and stock unification: \looseness=-1
\begin{align} \label{eq:dd}
    \textit{q}_{ijth} = \alpha_{ih} + \alpha_{jth} + \sum_{\substack{\tau=-5 \\ \tau \neq -1}}^{11} \lambda_{\tau} d[t+\tau]_{ith} + \sum_{\substack{\tau=-5 \\ \tau \neq -1}}^{11} \delta_{\tau} \textit{Treat}_{ih} \times d[t+\tau]_{ith} + \gamma'\textit{X}_{it} + \epsilon_{ijth},
\end{align}
where $\textit{q}_{ijth}$ is Tobin's \textit{q} for firm $i$ in three-digit SIC industry $j$ in year $t$ of event-year cohort $h$; $\alpha_{ih}$ and $\alpha_{jth}$ represent firm $i$ by cohort $h$ and industry $j$ by year $t$ by cohort $h$ fixed effects; $\textit{Treat}_{ih}$ is an indicator variable equal to one if firm $i$ experiences a recapitalization (or unification), and zero if remains single-class (dual-class) in event-year cohort $h$; $d[t+\tau]_{ith}$ is an indicator variable equal to one if year $t$ is $|\tau|$ years (-5 $\leq \tau \leq$ 11) before or after a recapitalization (unification) for firm $i$ in cohort $h$, and zero otherwise. Year $t-1$ is the baseline year and thus $d[t-1]_{ith}$ is equal to zero by construction; $\textit{X}_{it}$ is a vector of control variables in Eq. (\ref{eq:q}); and $\epsilon_{ijth}$ represents random errors clustered at the firm level. We require that both the treated and control firms exist over the [$t-2$, $t+2$] window, and exclude firm observations once they make another governance switch after year $t$. The resulting DD sample contains 186 recapitalizations and 224 unifications, representing a substantial share of 1,151 unique dual-class firms.  \looseness=-1

\subsubsection{Dynamic effects on Tobin's $q$} \label{sec:q_dynamics}
Figure \ref{fig:event_study_q} presents the estimation results for recapitalization (Panel A) and unification (Panel B). First, both panels show no evidence of pre-trends. All estimates on $\textit{Treat} \times d[t+\tau]$ for the pre-event period (i.e., $\tau$ = -5 through -2) are statistically insignificant and economically small. This result supports the key identifying assumption for a treatment effect interpretation of the event study estimates, and is inconsistent with alternative interpretations including anticipated switches reflected in $q$ and changes in $q$ affecting the switching decision---an instance of reverse causality. The lack of pre-trends is also inconsistent with other confounding events (such as restructuring) occurring before the recapitalization or unification. \looseness=-1

Estimates in Panel A show that firms that recapitalize from single-class to dual-class shares experience an immediate increase in Tobin's $q$ by 0.21 in that year, relative to a control group of firms that remain single-class ($p$-value = 0.018). The treated firms' $q$ remains elevated until year $t+5$ but declines from year $t+6$. By the last two years of the event window, the declines from year $t+5$ are statistically significant: e.g., $d[t+11]$ -- $d[t+5]$ = -0.366 ($p$-value = 0.004). These post-recapitalization dynamics of valuation mimic those in Proposition \ref{price_dynamic}---in which the stronger incentive effect of dual-class shares dominates in the early years, while accumulating agency costs erode this premium over time. \looseness=-1

Turning to unifications in Panel B, the estimated dynamics are essentially a mirror image of the corresponding estimates for recapitalizations. Firms that unify share classes experience a \emph{decrease} in Tobin's $q$ relative to a control group of firms that remain dual-class. The unifying firms' $q$ remains lower than the control firms' but the estimated treatment effects become closer to zero after year $t+5$. Again, these dynamic patterns are consistent with Proposition \ref{price_dynamic}. The somewhat lower precision of the estimates for unifications, compared to recapitalizations, is in part due to the much smaller number of firms in the control group---firms that continue to be dual-class. \looseness=-1

\subsubsection{Robustness tests} \label{sec:q_robustness}
Table \ref{tab:q_robustness} presents the DD estimates, summarizing the event-study results in fewer coefficients. Panels A and B present the pre- and post-treatment coefficients for recapitalization and unification, respectively. Column (1) in each panel shows the baseline estimates. Each of columns (2) through (8) uses an identical specification to column (1) except for one change to the baseline. Column (2) employs a matching approach to address a potential concern with the baseline event study that switching firms might have different pre-treatment characteristics than the control groups, which may bias the estimated treatment effects (e.g., \cite{Abadie2005}). Online Appendix \ref{appendix:matching} details the matching approach. Online Appendix Tables \ref{tab:recap_balance_q}-\ref{tab:unif_balance_cites} show that none of the average firm characteristics is significantly different between matched treated and control groups. Column (2) shows similar results to the baseline, suggesting that differences in observed firm characteristics do not drive our baseline results. Column (3) excludes dual-class firms that are equity carveouts identified using data from \citet{field2022} and \citet{Allen1998carveouts}.\footnote{We thank Michelle Lowry and John McConnell for sharing and making the datasets on equity carveouts available.} These carveout firms differ from the dual-class firms our model posits in that carveouts are typically controlled by other firms, instead of founders or their teams. For both types of events, the results are virtually identical after excluding carveout dual-class firms. 

In column (4), we consider an alternative assumption regarding the relative prices of superior and inferior share classes. Our baseline measure of $q$ uses the observed prices for both classes when available and otherwise assumes a 5\% premium for superior shares relative to inferior shares.\footnote{A 5\% premium is consistent with the average superior-class premium observed in our data (see Section \ref{sec:voting-premium}).} The column instead assumes equal prices across classes, and shows that our main results are robust to this alternative assumption. 

In column (5), we perform the event-study analyses using the \citet{sunabraham2021eventstudy} estimator. This approach avoids issues with two-way fixed effect models, such as not aggregating the treatment effects with appropriate weights. In our setting, two facts already mitigate this issue: (i) there are never-treated firms in the control group and (ii) our stacked-DD design aligns event-year cohorts. Nonetheless, we find our DD analysis is robust to using this alternative estimator. We also perform the DD regression weighted by firm size, proxied by the log of total assets, and find our results are largely unchanged in column (6).

In column (7) in Panel B, we address a selection concern specific to our event study of unification---firms might unify share classes when the founder (or his family member as heir) leaves the firm. In this case, the changes in $q$ may reflect the effect of the founder's departure, in addition to the effect of stock unification. We find that only 15\% of the unification cases we study have a founder or controlling family member stepping down from top management between two years before and after the unification.\footnote{We identify firms controlled by founders or their family heirs, and when they step down using the following data sources: family firm data from \cite{anderson2009} and \cite{anderson2012}, news sources, SEC filings, and officer and board member data from \cite{gkl2020}.} We then re-estimate the DD equation on a sub-sample that excludes these potentially confounded events. The result shows that the effect of unification on Tobin's $q$ is consistent with what we find on the full event sample. Finally, we Winsorize potentially unbounded variables at the $1^{st}$ and $99^{th}$ percentiles and find similar results with the baseline (column (9)).

Overall, the DD event-study estimates show dynamic effects of recapitalization and unification on valuations consistent with the model predictions. The lack of pre-trends provides credence to a treatment effect interpretation of the changes in Tobin's $q$ after firms making changes to voting structure. 

\subsubsection{Additional analyses: NYSE's one share-one vote policy and IPO sample} \label{sec:event_study_add}

We provide two additional pieces of evidence to bolster our main difference-in-differences analysis. First,  we exploit regulatory differences across stock exchanges on adopting dual-class shares. The New York Stock Exchange (NYSE) maintained its one share--one vote listing policy until 1986 (since 1926), while the American Stock Exchange (AMEX) and NASDAQ did allow listing with dual-class shares and such recapitalizations with little to no restriction throughout (\cite{Shorter2021DualClass}). Importantly, rising external governance forces during the early-1980s, especially the increasing threat of takeovers, created strong incentives for some firms to recapitalize.\footnote{The strengthening of external governance and the resulting increase in recapitalization can be interpreted through the lens of our model as follows: suppose that before the early 1980s, when external governance forces were weak, single-class shares allowed founders to extract \emph{some} private benefits---making them, in effect, dual-class shares with a finite sunset horizon $n < \infty$. In this case, the strengthening of external governance amounts to $n\to0$ for single-class firms. Since founders prefer a higher $n$---that is, stronger control rights---this shift would push some single-class firms to adopt dual-class shares (see Online Appendix~\ref{app:effect_n} for the conditions).} Under the assumption that this external governance shift influences firms' decisions to recapitalize similarly across the exchanges, ceteris paribus, non-recapitalizing firms on the NYSE (driven by a regulation) provide a plausible counterfactual for recapitalizing firms on the AMEX and Nasdaq in this period.

Accordingly, we define the treated group in this analysis as firms listed on the AMEX and Nasdaq that recapitalized from 1980 to September 1986, when the NYSE began allowing dual-class recapitalization. We construct a matched control group drawing from firms listed on the NYSE that remain single-class using an approach similar to that used in the main matched-sample DD analysis. A key assumption is that the outcomes of the matched NYSE firms capture the counterfactual outcomes of comparable AMEX/Nasdaq recapitalizing firms. Thus, to the extent that observable firm characteristics influence firms' outcomes in response to the external governance shocks, it is important to form the counterfactual with matched characteristics. Online Appendix Table \ref{tab:nyse_balance} shows that the matched treated and control firms are well-balanced on observable characteristics. We stop the analysis sample in September 1986, when the NYSE began allowing dual-class recapitalization, and thus non-recapitalizing firms on that exchange do not provide a plausive counterfactual for recapitalizing firms on the other exchanges anymore.

Given this short time frame, we estimate an adapted version of Eq. (\ref{eq:dd}) that uses quarterly data instead of annual. In addition, the resulting smaller event sample (n = 22) forces us to use less granular fixed effects---namely firm, matched pair, and year-quarter fixed effects.\footnote{We cluster standard errors at the matched pair level (\cite{Abadie_Spiess2022}). We address potential small sample bias from having 22 clusters by computing $p$-values using the wild cluster bootstrap (\cite{cameron2008bootstrap}). Online Appendix Table \ref{tab:event_study_nyse_bootstrap} shows that $p$-values obtained from standard clustering adjustments are very similar.}

Table \ref{tab:event_study_nyse} reports the estimation results. Prior to recapitalization, treated and control firms exhibit parallel pre-trends. Post-recapitalization, firms that recapitalize on the AMEX or Nasdaq exhibit an increase in Tobin's $q$ relative to firms that do not recapitalize but on the NYSE. By five to 11 quarters post-recapitalization, the increases in Tobin's $q$ are significant with the magnitudes being around 0.28 ($p$-value = 0.047 in column (2) with matched pair fixed effects). These magnitudes are similar to the short-run magnitudes from the baseline DD analysis of recpitalization, averaging 0.18 between years $t+1$ and $t+3$. The similar economic magnitudes lend credibility to our estimates. 

The second piece of evidence is motivated by our model, Proposition \ref{price_dynamic} in particular, which predicts related within-firm dynamics post-IPO---the valuation of dual-class firms declines faster over time relative to that of single-class firms. In Table \ref{tab:ipo_dynamics}, column (1) we perform a within-firm analysis of Tobin's $q$ after IPOs. We find that post-IPO, Tobin's $q$ for dual-class firms declines significantly faster than that for single-class firms, controlling for firm and industry-by-year-by-IPO cohort fixed effects and firm-level controls. While the nature of this test does not allow for identifying a treatment effect of dual-class shares, it nonetheless provides evidence consistent with the model prediction and the difference-in-differences event study estimates of governance switches.

\subsubsection{Market reactions to dual-class recapitalizations and stock unifications} \label{sec:market-reactions}

As a complementary approach to testing the model prediction for dynamics of the dual-class premium, we, in this section, examine how the market value of equity changes in the short run in response to dual-class recapitalizations and stock unifications. By using short-run stock returns, this approach sidesteps potential measurement issues of Tobin's $q$ related to, e.g., the replacement cost of assets (see, e.g., \citet{dybvig2015tobin} and \citet{bartlett2020misuse}). See Section \ref{subsec:dcfirms} for the description of our collection of announcement dates for these events. \looseness=-1

For the sample of recapitalization and unification announcements, we compute excess daily stock returns using a market-adjusted-return model: 
${\epsilon}_{it} =  \;R_{it} - R_{mt},$
where ${\epsilon}_{it}$ is the rate of abnormal return and ${R}_{it}$ is the rate of return for firm-stock $i$, and ${R}_{mt}$ is the rate of return for the market portfolio, proxied by \emph{vwretd} from CRSP, on day $t$. We use cumulative abnormal returns (CARs) over the [--3, +3] window in the analysis, and find similar results using alternative event windows [--3, +2] and [--2, +2] (see Online Appendix Table \ref{tab:recaps_and_unifications_robust}). To increase the sample size, for this analysis only we include firms in the financial, utilities, and unclassified sectors. Requiring the CAR provides a sample of 116 dual-class recapitalizations and 61 unifications announced between 1976 and 2018. \looseness=-1

We estimate a regression in which the CAR is the dependent variable and the independent variable is a $\mathbb{1}[\text{age $>$ 12}]$ indicator. The key identifying assumption is that recapitalizing or unifying firms are otherwise comparable with each other. The assumption includes, for example, that the market had similar expectations about the probability of recapitalization or unification before the announcement. Our effort to collect the earliest public announcement of a given event in major news outlets supports this assumption.

Table \ref{tab:recaps_and_unifications} presents the estimation results. The positive regression constant in column (1), 0.029 (\textit{p}-value = 0.012), implies a positive value effect of dual-class recapitalizations for young firms. This positive effect is consistent with the incentive role of dual-class shares outweighing private benefits of control early in the firm's life. However, the negative coefficient on $\mathbb{1}[\text{age $>$ 12}]$, --0.035 (\textit{p}-value = 0.017), shows that for mature firms, the value effect of recapitalization is lower by 3.5 percentage points. This result is consistent with the model implication that a dual-class recapitalization has an increasingly negative effect on firm value over maturity---due to increasing private benefits extraction---as well as the dynamic DD estimates in the previous section. The heterogeneous results across firm maturity potentially explain the insignificant average effect of the recapitalization announcement shown in previous research (e.g., \cite{partch1987creation}). \looseness=-1 

Next, we examine the market's reaction to the announcement of dual-class stock unifications. The positive coefficient on $\mathbb{1}[\text{age $>$ 12}]$ in column (2), 0.052 (\textit{p}-value = 0.017), shows that the market perceives a 5.2 percentage-point higher value effect when mature firms unify dual-class shares compared to when young firms unify. The economic magnitudes of the effect of switching between dual-class and single-class structures on mature, relative to young, firms estimated using the two events are comparable---estimates on $\mathbb{1}[\text{age $>$ 12}]$ = --0.035 and 0.052 in columns (1) and (2). These comparable magnitudes suggest the estimates capture a common economic force---such as the present value of private benefits of control as in the model. \looseness=-1

\subsection{Effects on technological innovation: evidence on firm-specific effort}\label{sec:innovation}

Section \ref{sec:q} provides evidence consistent with a transitory dual-class premium in which dual-class firms have higher valuations than single-class firms when they are young, but the valuation gap shrinks over maturity. In this section, we test the model implications for the founder's effort by employing patent output as an outcome of the firm-specific investment.\footnote{We recognize that the founder's investments in innovation can be either firm-specific or general. The firm-specificity of innovation may stems from two sources: (i) the founder's attachment to the firm he has built and (ii) the economic costs of redeploying the product of the investment elsewhere. For example, the costs of moving a social networking or ride sharing platform, whose value depends on network effects among both providers and users would be very high. Thus, we posit that the founder's firm-specific investment is an important driver of the firm's innovative output.} Under the treatment effect, the founder of a dual-class firm exerts a greater effort than the single-class counterpart, resulting in greater innovative output for the firm.\footnote{Related, \citet*{acharya2014wdl} show evidence that wrongful discharge laws encourage employees' firm-specific effort and spur innovation by limiting the firm's ability to holdup.} But under the selection effect that replicates the transitory premium, dual-class firms would produce \emph{less} innovative output (Proposition \ref{selection_experiment}). 

We estimate a version of difference-in-differences Eq. (\ref{eq:dd}) in which the dependent variable is the number of citations the firm-year's patents generate in one, two, or three years ahead (e.g., \citet{bernstein2015innovation}). Consistent with the literature on firm-level determinants of innovation output (e.g., \cite*{fang2014stock}), we include the following controls: log book assets, market leverage, R\&D, asset tangibility, and ROA. To account for both the extensive- and intensive-margin effects, we use the inverse hyperbolic sine transformation. In addition, we examine the two margins separately by employing (i) indicator variables equal to one if the firm-year has any patents that generate citations, and zero otherwise; and (ii) the log of the number of the citations the firm-year's patents generate as the dependent variables. \looseness=-1

Figure \ref{fig:event_study_cites} presents the event-study estimates for the transformed patent citations. Both panels show no evidence of pre-trends for recapitalization and unification, supporting the identifying assumption. Panel A shows that post-recapitalization, the switching firms experience an increase in patent output, as measured by their citations, relative to firms that stay single-class. In contrast, Panel B shows that post-unification, the switching firms experience a decline in patent citations, relative to firms that stay dual-class. The magnitude of the effects increases over time for both types of events, consistent with a time lag between (firm-specific) investment and the patent outcome. This ramp-up (-down) of patenting activities post-event, combined with parallel pre-trends, appears inconsistent with recapitalizing (unifying) firms selecting on pre-event changes in innovation.\footnote{\cite*{liang2022Sunset} and \cite*{baran2023dc_innovation} find that in the cross-section, dual-class shares and the wedge between insiders' voting and cash-flow rights, respectively, are positively associated with firms' patent output.} \looseness=-1

These results are instead consistent with a treatment effect in which the founder exerts greater firm-specific effort under dual-class than single-class structures---which in turn increases the probability of ``transforming'' the firm, as measured by generating greater innovative outputs. We further explore this implication by separately examining the extensive and the intensive margins of innovation in Table \ref{tab:patent_robustness}. In columns (6)-(7) in Panel A, we find that the positive effect of dual-class recapitalization on innovative output is largely driven by the extensive margin---i.e., the recapitalizing firms become to generate patents that are cited. For unifications (in Panel B), innovative output declines at both the extensive and intensive margins, although the latter effect is statistically insignificant.\footnote{Following \cite{lee2009bounds} and \citet{chen2024logzeros}, we estimate lower and upper bounds on the intensive margin effect based on different assumptions on which firm select to produce patents.}

Table \ref{tab:patent_robustness} also reports these DD estimates on innovation output using the series of robustness tests that we run in Section \ref{sec:q_robustness} for $q$: (i) using matched samples (column (2)); (ii) excluding dual-class firms with equity carveouts (column (3)); (iii) using the \citet{sunabraham2021eventstudy} estimator (column (4)); (iv) weighting the regression by the log of total assets (column (5)); and (v) excluding unification cases where the controlling founder or family member leaves around the event (column (8)). Finally, column (9) reports the event study results that use supplementary firm-patent linkages from \citet{ma2025obsolescence} for all firms, instead of only ever-dual firms (see Section \ref{sec:addl_data} for more details). Again, our results on firm innovation across all these tests are similar to those in column (1), the baseline.

Overall, the evidence on the dynamics of innovative output from the difference-in-differences analysis is consistent with a treatment effect of dual-class shares rather than selection. In particular, the results are consistent with the heightened incentives dual-class shares provide having a larger impact on firms' transitions into innovating (the extensive margin) than on the level of innovative output once they become innovative (the intensive margin). Arguably, the former requires greater firm-specific investment and transformation, for which dual-class shares are well suited. \looseness=-1

Finally, similar to the analysis of $q$ above, we conduct a within-firm analysis of innovative output for firms that go public with dual-class versus single-class share structures. We note that this test does not identify a treatment effect of dual-class shares on firm-specific investment. Nonetheless, the results in Table \ref{tab:ipo_dynamics}, column (2) show that post-IPO, dual-class firms' innovative output rises faster, especially when their age is in the third and fourth quartiles, relative to single-class firms', consistent with the model prediction.

\subsection{Effects of dual-class shares by specificity of investments} 
\label{sec:redeployability}

Next, we test a core mechanism of our model: dual-class shares incentivize founders to undertake firm-specific investments (``effort'') by protecting them from the loss of private benefits that result from the investments. This mechanism implies that the effects of dual-class shares are, ceteris paribus, more pronounced when founders' investments are more firm-specific; if all investments are general, there is no holdup, and thus no effect of dual-class shares we posit.

Given the lack of a readily available measure of the importance of firm-specific investments by founders, we construct a proxy using a combination of the specificity of key production inputs: capital and labor. The idea is that the founder's investment is likely more specific when the capital and labor that he manages are more specific. For capital, we employ the measure of asset redeployability from \citet{kim2017redeployability}. Assets that are used by fewer firms in the economy are considered more specific. For labor, we use the index of labor mobility from \citet{barry2026mobility}. If workers in an industry move to a more concentrated set of other industries, they are considered more specific. As a final measure, we employ the weighted average of the standardized labor and capital specificity measures using the labor share and [1 -- labor share] as the respective weights. We construct a measure of the labor share at the industry level following \cite*{elsby2013laborshare}.\footnote{The asset redeployability measure is at the firm level, weighted-averaged across industry segments. We fill in missing firm-level measures using BEA industry-level measures, crosswalked to three-digit SIC codes. The labor mobility measure is at the 1990 Census industry code level, crosswalked to three-digit SIC codes following \citet{autor2019indxwalks}. The labor share data are from the Bureau of Labor Statistics KLEMS. We measure the labor share using BEA industry codes, which we crosswalk to two-digit SIC industries from 1971 to 1986, and to three-digit NAICS industries from 1987 to 2022.}

Table \ref{tab:dd_redeployability} shows results for the difference-in-differences analysis for Tobin's \textit{q} and patent citations, with the sample split at the median of investment specificity. In Panel A, we find that the effects of recapitalization and unification on $q$ are concentrated among firms in which the founder's specific investments are likely more important---when their capital and labor are more specific. In contrast, the effects are generally smaller in magnitude and insignificant for firms with less specialized production inputs. In Panel B, we find parallel concentration of the effect of unification on patent outputs among firms with higher specificity but not for recapitalization.

\subsection{Sensitivity of investment to opportunities: evidence for private benefits} \label{sec:sesitivity}

We now turn to testing the model implication for private benefits of control. The model shows that if the dynamics of dual-class firms' valuation are due to treatment (selection), then founders of dual-class firms consume greater (smaller) private benefits than those of single-class firms---and this gap increases (decreases) over maturity. Motivated by work that shows entrenched founder-managers have stronger preferences for preserving existing capital and hence reducing ``creative destruction'' than shareholders do (e.g., \cite{morck2000inherited}; \cite{bertrand2003enjoying}), we employ the extent to which firm investment decisions are \emph{insensitive} to opportunities as a measure of private benefits.\footnote{See, e.g., \citet{foster2001productivity} for a discussion of theoretical frameworks and evidence on the link between creative destruction and productivity.}\looseness=-1

Specifically, we estimate how the $q$-sensitivity of investment differs between dual-class and single-class firms over maturity. Following the extensive literature on investment, we use Tobin's \textit{q} as a proxy for marginal $q$ and also include cash flows in an investment equation.\footnote{This measurement of expected marginal returns to capital using Tobin’s $q$ follows the literature on neoclassical $q$-theory of investment, which typically assumes no agency frictions (e.g., \cite{hayashi1982}; \cite{Kaplan1997investment}). Thus, one caveat of interpreting this analysis is that the effects of dual-class shares on market valuations, as we document earlier, might affect the estimated $q$-sensitivity of investment over maturity.} The resulting equation is: \looseness=-1
\begin{align} \label{eq:inv}
\textit{Investment}_{it} = \alpha_{i \times \text{Dual}} + \alpha_{t \times \text{Dual}} + \beta_1\textit{q}_{it-1} + \beta_2\textit{q}_{it-1} \times \textit{Dual}_{it} + \beta_3\textit{CF}_{it} + \beta_4\textit{CF}_{it} \times \textit{Dual}_{it} + \epsilon_{it},
\end{align}
where $\textit{Investment}_{it}$ is capital expenditures scaled by lagged book assets; $\textit{Dual}_{it}$ is an indicator variable equal to one if firm $i$ has multiple classes of shares with different voting rights in year $t$; $\alpha_{i \times \text{Dual}}$ and $\alpha_{t \times \text{Dual}}$ represent firm and year fixed effects interacted with $\textit{Dual}_{it}$; $\textit{q}_{it-1}$ is beginning-year Tobin's \textit{q}; $\textit{CF}_{it}$ is cash flows scaled by lagged book assets for firm $i$ in year $t$; and $\epsilon_{it}$ represents random errors clustered at the firm level. 

To examine whether dual-class and single-class firms exhibit different sensitivities over maturity, we estimate Eq. (\ref{eq:inv}) on subsamples of young (mature) firms with ages less than or equal to (greater than) 12 years, the median age. Given that the relevant private benefits involve reluctance to cut capital, we estimate these sensitivities on samples of firms facing weak demand conditions---when they are likely near the disinvestment threshold---proxied by sales growth in the first quintile (e.g., \cite{achyuta2013firing}). The results are robust to other cutoffs (see below). \looseness=-1

Table \ref{tab:invest_q_sensitivity} reports the estimation results for Eq. (\ref{eq:inv}). All columns show a strong positive correlation between investment and opportunities for single-class firms, as captured by significantly positive baseline coefficients on $q$---whether firms are relatively young or mature. Column (3) shows that mature dual-class firms have significantly lower investment-$q$ sensitivities than mature single-class firms ($p$-value = 0.011). In contrast, column (2) shows there is no such difference for young firms. The bottom row shows that the difference in the sensitivity between dual-class and single-class firms across the two columns is significant ($p$-value = 0.014). In Online Appendix Table \ref{tab:invest_q_sensitivity_robust} we show that the results are robust to the cutoff for sales growth to define the analysis sample. 

Overall, these results are consistent with a treatment effect of dual-class shares in which founder-managers of dual-class firms extract larger private benefits of control than those of single-class firms especially when the firms are mature. 

\subsection{Voting premium: evidence that private benefits increase with maturity} \label{sec:voting-premium}
The previous section shows evidence that private benefits consumption is larger for dual-class firms relative to single-class firms, particularly when they are mature. In this section, we test a related general premise: private benefits of control increase over firm maturity. This is crucial for the model, in which growing private benefits lead the valuation premium for dual-class firms to shrink in the long run, and also provide \emph{ex-ante} incentives to the founder. \looseness=-1

We use the voting premium---the difference between market prices of superior-voting shares relative to inferior-voting shares for dual-class firms---as a measure of private benefits of control. In the model, the voting premium is the expected present value of private benefits. In practice, other factors such as the probability of a control contest, the difference in liquidity across share classes, and heterogeneous beliefs among shareholders also affect the premium.\footnote{See, e.g., \cite{lease1983market}; \cite{zingales1995determines}; \cite{nenova2003value}; and \cite{doidge2004us}. \cite{zingales1995determines} shows both theoretically and empirically that the voting premium is determined by a combination of the probability of a control contest and the private benefits of control. Similarly, \cite{nenova2003value} argues that the value of control-block votes is a lower bound for the private benefits to the controlling shareholder. See \citet*{levit2026voting} for a theoretical analysis showing that the voting premium may also reflect heterogeneous preferences or beliefs among shareholders, and capture the marginal rather than the average benefit from additional voting
rights.} Therefore, our analysis controls for measures of other determinants of the voting premium, to the extent possible. \looseness=-1

The analysis relies on a subsample of dual-class firms that trade both the superior- and inferior-voting shares, and have information on ownership structure which begins in 1980. To mitigate the influence of illiquidity and microstructure noise, we exclude firms with superior or inferior share prices below \$5. The resulting sample includes 1,182 dual-class firm-years (for 114 unique firms) from 1980 through 2022. To analyze the relationship between firm maturity and the voting premium, we estimate the following equation:
\begin{align} \label{eq:vp}
\textit{VP}_{it} = \alpha_i + \alpha_t + \sum_{n=2}^{4}\beta_n{\Large \mathbb{1}}[\text{age in quartile}\, n]_{it}  + \gamma' \textit{X}_{it} + \epsilon_{it},
\end{align}
where $\mbox{\textit{VP}}_{it}$ is the voting premium for dual-class firm $i$ in year $t$, computed as $(P_A-P_B)/(P_B-rP_A)$, where $P_A$ ($P_B$) is the price of superior- (inferior-) voting shares and $r$ is the ratio between the number of votes for the inferior- and superior-voting shares (e.g., 1/10);\footnote{Under the assumption that $r$ < 1, the voting premium is the ratio between the price of one voting right and the price of one cash-flow right.} $\alpha_i$ and $\alpha_t$ represent firm and year fixed effects;\footnote{The relatively small size of the analysis sample does now allow us to include industry by year fixed effects.} ${\Large \mathbb{1}}[\text{age in quartile} \,n]_{it}$ is an indicator variable for the $n^{th}$ firm age quartile, and firm age is defined using years since a firm's IPO; $X_{it}$ includes proxies for the probability of a control contest: log market equity and the Shapley value of ownership structure, and the log of relative trading volume (\cite{zingales1995determines}); and $\epsilon_{it}$ represents random errors clustered at the firm level. In the presence of firm fixed effects, year fixed effects and firm age are collinear. We address this issue by normalizing year fixed effects following \cite{deaton1997}.\footnote{\cite{deaton1997}'s approach normalizes year fixed effects to have a zero mean and orthorgonalizes them to a linear time trend. Firm age therefore captures the linear trend, while the orthogonalized year fixed effects capture deviations from the trend. Consistent with this normalization, the year fixed effects when estimated without normalization and age indicators do not exhibit a significant trend.} The coefficient of interest is $\beta_n$ which we interpret as the relation between firm maturity and the private benefits. \looseness=-1

Table \ref{tab:voting_premium0} presents the estimation results for Eq. (\ref{eq:vp}). Column (1) shows that more mature firms exhibit a higher voting premium in the cross-section (without firm fixed effects). The estimates suggest that, relative to the youngest quartile of firms, voting premia for firms in the third and fourth age quartiles increase by 3.4 and 5.4 percentage points (the latter's \textit{p}-value = 0.038), respectively. Similarly, when we control for both firm and year fixed effects in column (2), voting premia increase by 4.8 and 9.6 percentage points for firms in the third and fourth quartiles (\textit{p}-values = 0.054 and 0.025), respectively. The economic magnitudes are considerable, given the mean voting premium of 5.1\%.\footnote{The mean is similar to those reported in previous research on US firms (e.g., \cite{lease1983market}; \cite{zingales1995determines}). The coefficients on control variables have expected signs: the log market value of equity is negatively and the Shapley value is positively related with the voting premium.} In column (3), as a common alternative normalization approach (see, e.g., \citet{argente_2024jpe}), we normalize one year fixed effect to zero and find estimates that are quantitatively similar to those in column (2). These results are robust when we exclude log market equity which might be mechanically correlated with the voting premium (Online Appendix Table \ref{tab:voting_premium_robust}). Overall, the evidence in this section shows that private benefits of control reflected in the value of votes increase over firm maturity, providing support for a key premise of our model and empirical analysis. \looseness=-1

\section{Discussions and Policy Implications} \label{sec:policy}

Given the increasing private benefits associated with managing dual-class firms over maturity, founders may have little incentive to unify dual-class shares when their firms are mature. Does this support the case for so-called ``sunset provisions?''---clauses in corporate charters that eliminate dual-class shares once a specific date is reached or a threshold event occurs. Some academics and policy makers have argued that dual-class shares should be precluded altogether or, if they are used, should include sunset provisions (e.g., \citet{bebchuk2017untenable}; \citet{jackson2018}; \cite{cremers2024life}). This paper's theoretical and empirical analysis however suggests that these policy implications are not clear-cut. \looseness=-1

First, our analysis shows that a key benefit of dual-class shares is to provide incentives for making firm-specific investments by founders and their teams---which spur technological innovation, for example. An important context for this incentive effect is that the nature of firm investments and production has become more specialized (see, e.g., \cite{ekerdt_wu2025}; \cite{barry2026mobility}). This trend toward more specific investments, rather than general ones, suggests that benefits of dual-class shares may have become more important. The trend also appears consistent with recent policy changes that allow for or even encourage public listing with dual-class shares---decisions by the Hong Kong and Singapore Exchanges in 2018, the London Stock Exchange's Main Market premier tier in 2021, and the EU's Multiple Voting Rights Directive in 2024. Precluding dual-class structures will prevent firms, and their investors and entrepreneurs from reaping these, perhaps increasingly important, benefits.

Second, our model shows that sunset provisions will weaken the incentive that dual-class shares provide---by reducing the founder's private benefits. Therefore, the model implies that the effect of dual-class shares with sunset provisions is more similar to that of single-class shares. In Figure \ref{fig:appendix_event_study_sunsets}, we find evidence consistent with this implication: when dual-class firms with sunsets unify share classes, their valuations do not change (Panel A). If dual-class shares with sunsets provide relatively little incentive to begin with, then their removal would also have a smaller impact. \looseness=-1 

More broadly, our theoretical framework allows us to analyze the impact of firm-specific and market-wide factors on the choice of dual-class shares (including sunset provisions) or single-class shares, and the resulting effects on firms. If a firm has a small deadweight loss from private benefits consumption, as captured by high $\gamma$ in the model, all else equal, dual-class shares are preferred to single-class. Cases in which the founder derives private benefits from ``altruistic preferences'' would be such an example. Additionally, a larger fraction of firms will adopt dual-class shares when external governance forces such as activist investors in the economy become stronger. Finally, an increase in the importance of the founder's firm-specific investments for the firm's success---the core mechanism of our model---would lead to increased adoption of dual-class shares. Some combination of these forces might explain the recent rise in dual-class IPOs. Analyzing this heterogeneity in the cross-section and over time is an important topic for future research.\footnote{See \cite{philippon2006} for time-varying effects of corporate governance on firm and aggregate outcomes.} \looseness=-1

\section{Conclusion} \label{sec:conclusion}

In this paper, we argue that dual-class shares provide entrepreneurs incentives to make firm-specific investments using both theoretical and empirical analyses. The heightened incentives allow dual-class firms to be more innovative and increase their value, whereas growing agency costs over maturity erode the value added from these incentives. Consistent with the model, the private benefits of control, as measured by the voting premium and reluctance to cut investment when opportunities are poor, are most pronounced for mature dual-class firms. Our model provides an analytical framework to study the treatment and selection effects of dual-class shares, highlighting their starkly different implications for the dynamics of private benefits, founder investment, and firm value. Our empirical analysis exploits new data on US dual-class firms over 50 years in dynamic difference-in-differences designs for recapitalizations and unifications. The evidence across the outcomes is consistent with a treatment effect interpretation. \looseness=-1 

Our results have broad implications beyond dual-class shares, because other forms of deviation from one share--one vote such as pyramids and cross-ownership---which are widely used outside the US to consolidate control---can be replicated by dual-class structure (\cite{bebchuk2000stock}; \cite{claessens2002disentangling}; \cite{faccio2002ultimate}; \cite{becht2020loyalty}). Moreover, our theoretical framework can be applied in future research to analyze the treatment and selection effects of other governance mechanisms, such as classified boards and takeover defenses. \looseness=-1


\pagebreak
\setstretch{0.9}
\bibliographystyle{Bibliography/jf}
\bibliography{Bibliography/dualclass_JF}
\clearpage
\clearpage
\setstretch{1}
\newpage
\newgeometry{left=1cm,bottom=2cm, right=1cm, top=1cm}
\renewcommand{\thefigure}{\arabic{figure}}

\begin{figure}[t] 
    \centering
    \captionsetup{font = large, justification=centering}
    \caption{\\\textbf{Fraction of Dual-Class IPOs among the Universe of IPOs, 1980--2025}}
    \includegraphics[width = \textwidth]{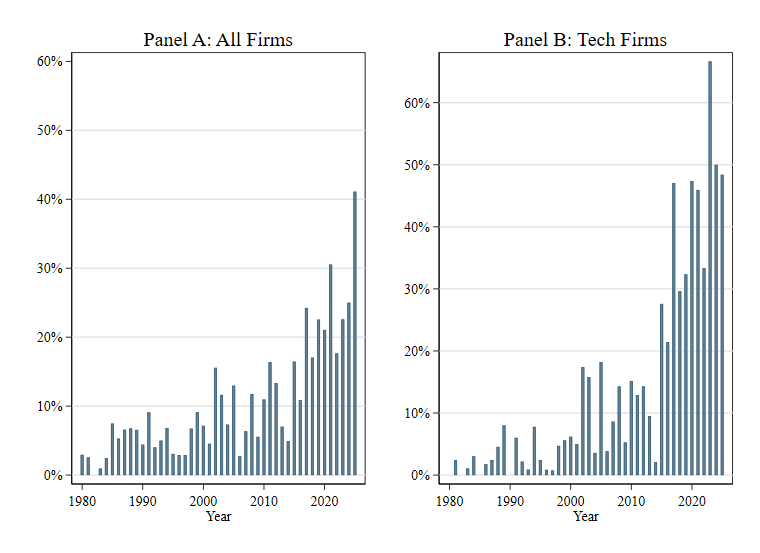}
    \captionsetup{font = small, justification=justified}
    \caption*{\textit{Notes}: This figure shows the fraction of dual-class IPOs among the universe of all IPOs (Panel A) and IPOs in the technology sectors (Panel B) as defined in \cite{loughran2004ipo}, plus SIC codes 3559, 3576, and 7389, from 1980 through 2025. The sample contains all offerings of common shares from Refinitiv (SDC Platinum) with an offer price of at least \$5.00, excluding REITs, ADRs, closed-end-funds, units (such as special purpose acquisition companies, SPACs), and companies not listed on CRSP within 7 days of the IPO. The data for 2025 are from Jay Ritter's website.}
    \label{fig:IPOs}
\end{figure}{}

\clearpage

\begin{figure}[t] 
    \centering
    \captionsetup{font = large, justification=centering}
    \caption{\\\textbf{Number and Fraction of Dual-Class Firms among Compustat Universe}}
    \includegraphics[width = \textwidth]{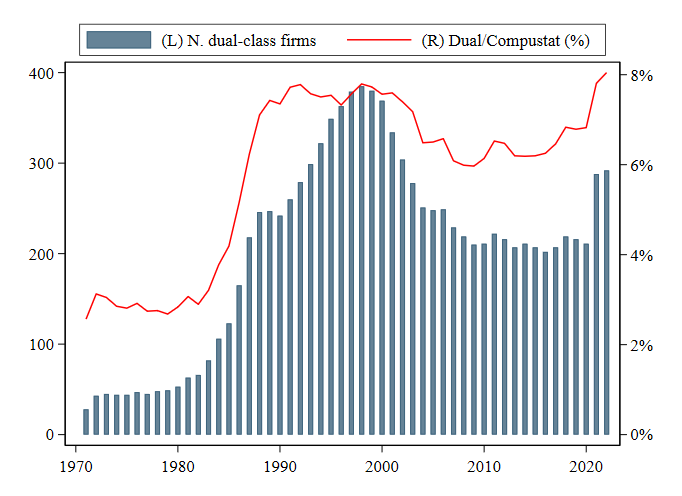}
    \captionsetup{font = small, justification=justified}
    \caption*{\textit{Notes}: This figure shows the number (blue bars, left y-axis) and percent (red line, right y-axis) of dual-class firms relative to Compustat firms that satisfy the sample selection criteria in Section \ref{sec:sample} from 1971 through 2022.}
    \label{fig:DC_counts}
\end{figure}{}

\clearpage

\begin{figure}[t] 
    \centering
    \captionsetup{font = large, justification=centering}
    \caption{\\\textbf{Event Study of Dual-Class Recapitalization and Stock Unification: Tobin's \textit{q}}}
    \includegraphics[width = 0.85\textwidth]{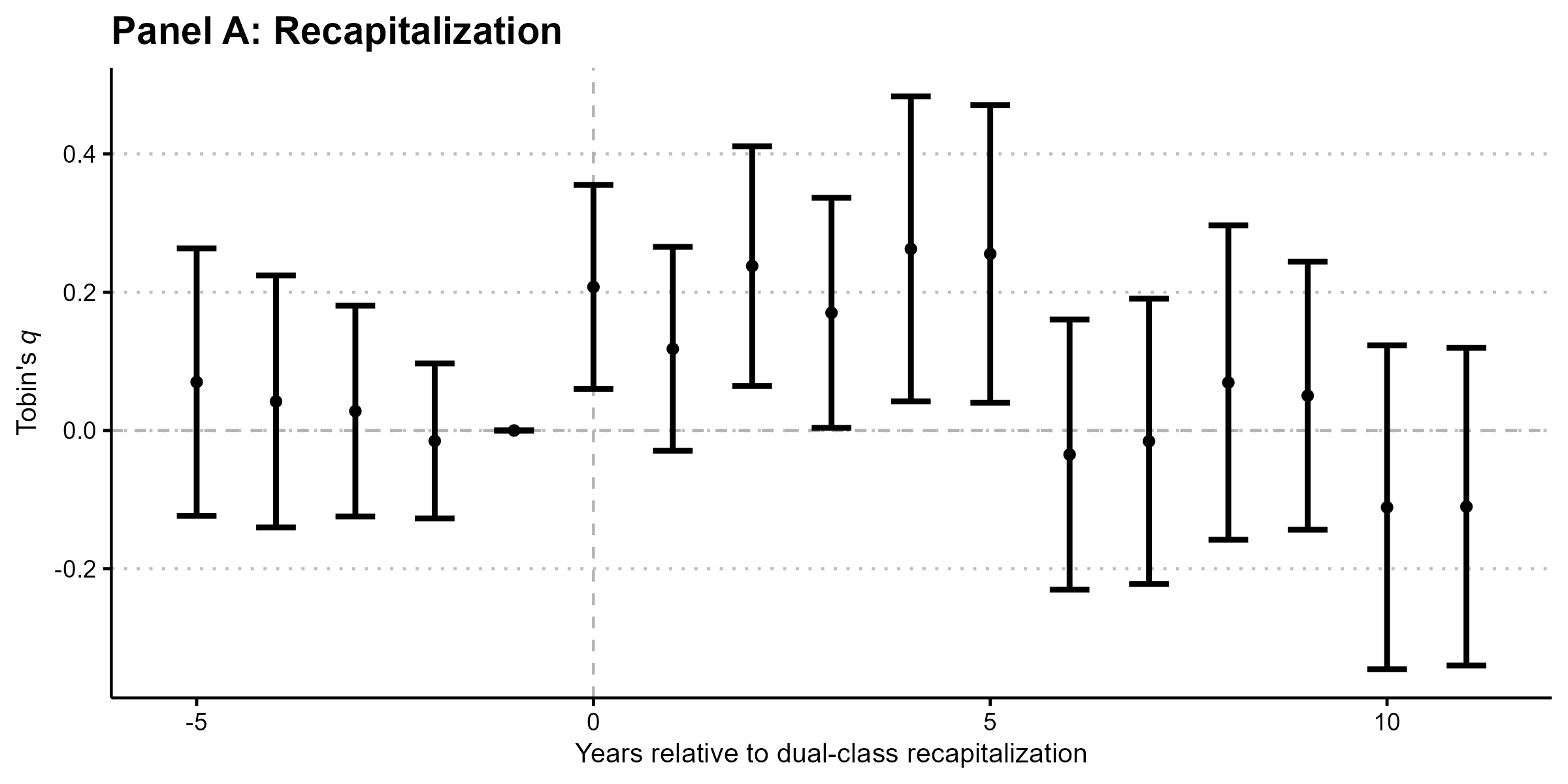}
    \includegraphics[width = 0.85\textwidth]{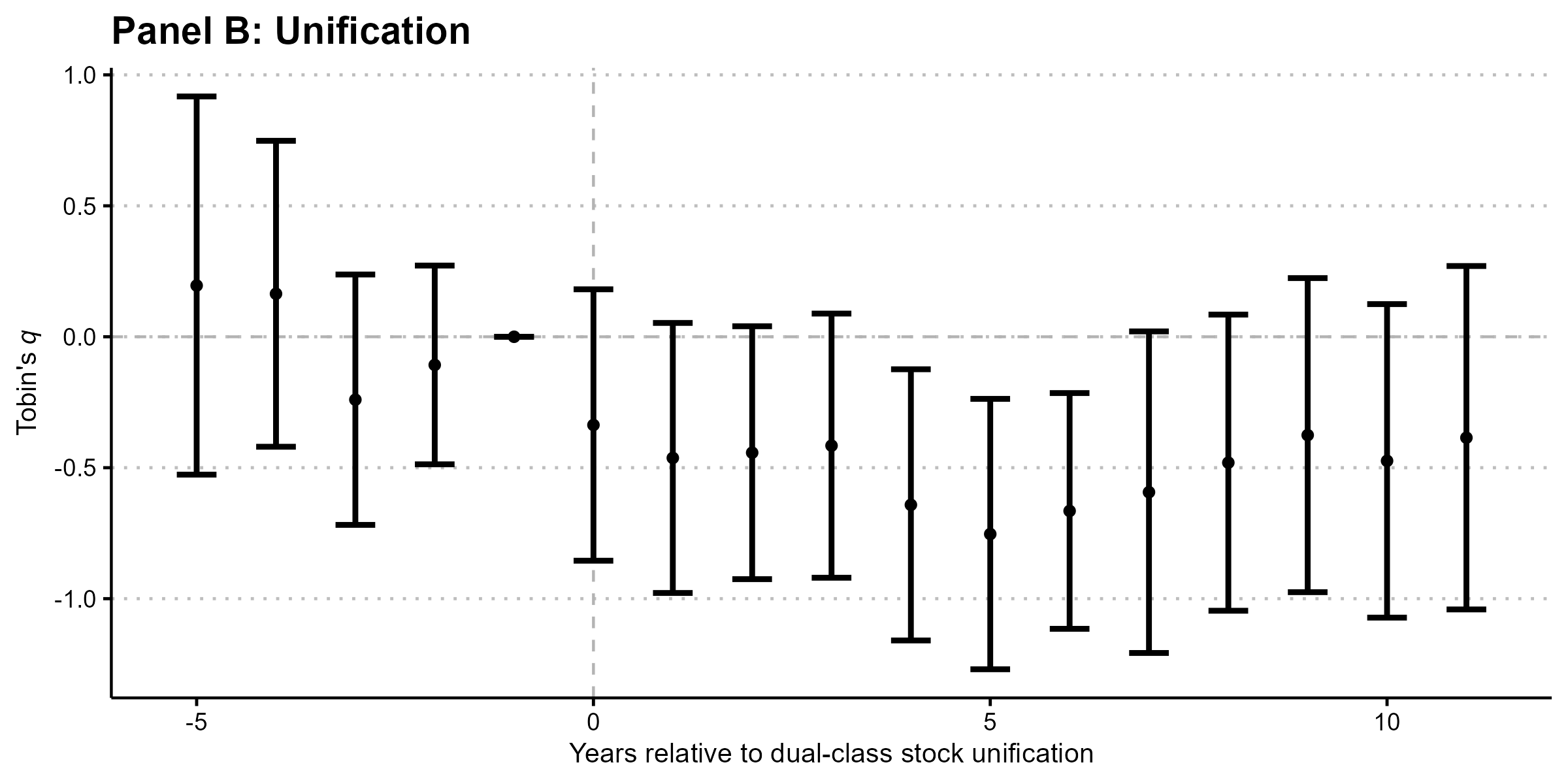}
    \captionsetup{font = small, justification=justified}
    \caption*{\textit{Notes}: This figure plots the stacked dynamic difference-in-differences event-study estimates of the effect of dual-class recapitalization (Panel A) and stock unification (Panel B) on Tobin's \textit{q} from Eq. (\ref{eq:dd}). The event study requires that both the treated and control groups of firms exist during the $[t-2, t+2]$ window. The markers on the lines show point estimates. The 90\% confidence intervals in vertical lines are based on standard errors clustered at the firm level. }
    \label{fig:event_study_q}
\end{figure}{}

\begin{figure}[t] 
    \centering
    \captionsetup{font = large, justification=centering}
    \caption{\\\textbf{Event Study of Dual-Class Recapitalization and Stock Unification: Innovative Output}}
    \includegraphics[width = 0.85\textwidth]{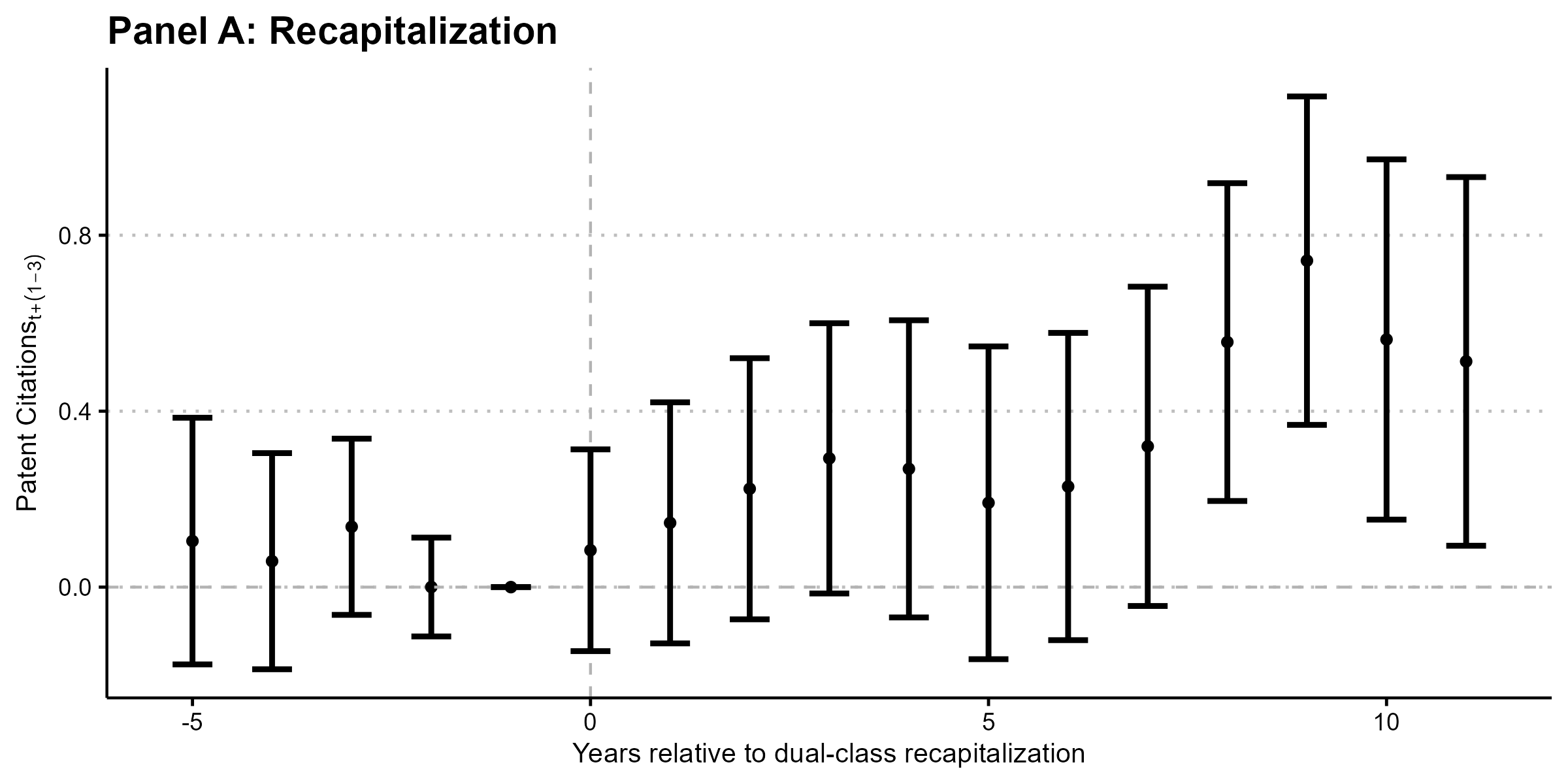}
    \includegraphics[width = 0.85\textwidth]{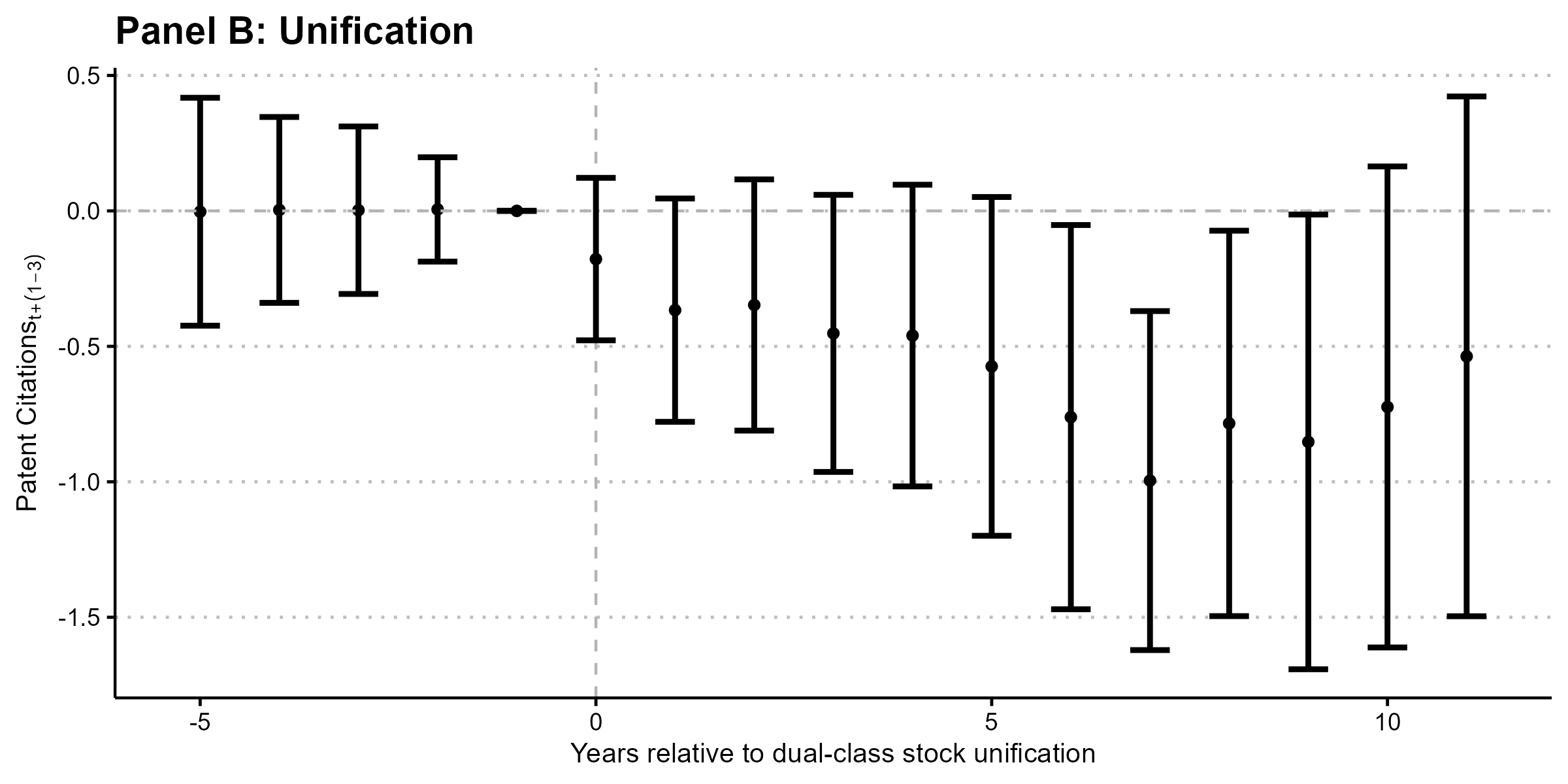}
    \captionsetup{font = small, justification=justified}
    \caption*{\textit{Notes}: This figure plots the stacked dynamic difference-in-differences event-study estimates of the effect of dual-class recapitalization (Panel A) and stock unification (Panel B) on innovative output from a version of Eq. (\ref{eq:dd}) in which the dependent variable is the inverse hyperbolic sine-transformed number of citations the firm's patents receive one, two, or three years ahead. The event study requires that both the treated and control groups of firms exist during the $[t-2, t+2]$ window. The markers on the lines show point estimates. The 90\% confidence intervals in vertical lines are based on standard errors clustered at the firm level.
    }
    \label{fig:event_study_cites}
\end{figure}{}

\clearpage
\setstretch{1}

\newpage
\newgeometry{left=1cm,bottom=2cm, right=1cm, top=1cm}
\renewcommand{\thefigure}{\Roman{figure}}

\begin{table}[]
    \centering
    \small
    \captionsetup{font = large, justification=centering}
    \caption{\textbf{Descriptive Statistics on Dual-Class and Single-Class Firms}}
\begin{tabular}{lcccccc}
    \toprule \toprule
     & \multicolumn{2}{c}{Dual-class} & & \multicolumn{2}{c}{Single-class} & Dual -- Single\\
     \cline{2-3} \cline{5-7}
        \multicolumn{1}{c}{Variable} & Mean & Std Dev &  & Mean  & Std Dev  & Diff.  \\ 
\midrule \addlinespace[\belowrulesep]
Book assets&3,966.8&18,872.3&&3,533.3&1,734.3&433.5 \tabularnewline
Log book assets&6.344&1.791&&5.756&2.016&0.588*** \tabularnewline
Age&15.053&12.985&&13.257&11.850&1.796*** \tabularnewline
Tobin's $q$&2.407&2.890&&2.126&2.198&0.281*** \tabularnewline
Sales growth&0.190&0.592&&0.220&0.705&-0.029*** \tabularnewline
ROA&0.128&0.165&&0.111&0.222&0.017*** \tabularnewline
Market leverage&0.293&0.261&&0.265&0.249&0.028*** \tabularnewline
Capex&0.072&0.092&&0.083&0.106&-0.011*** \tabularnewline
R\&D&0.024&0.065&&0.042&0.091&-0.018*** \tabularnewline
Tangibility&0.290&0.213&&0.313&0.242&-0.023*** \tabularnewline
Payout ratio&0.028&0.042&&0.026&0.042&0.002* \tabularnewline
Cash flows&0.073&0.175&&0.056&0.238&0.017*** 
\tabularnewline
IHS(citations)&1.502&2.550&&1.597&2.764&-0.095 
\tabularnewline
Observations&9,950&-&&150,752&-&- \tabularnewline
\toprule \toprule 
\end{tabular}
\begin{tablenotes}
\small
\item \textit{Notes:} This table presents descriptive statistics on financial variables for dual-class and single-class firm-years merged with CRSP and Compustat from 1971 through 2022. ``Book assets'' is book value of assets; ``Log book assets'' is the log of book assets; ``Age'' is the number of years since an IPO (proxied by the first appearance in Compustat with a valid stock price); ``Tobin’s \textit{q}'' is the ratio of the market value of capital to the book value of capital; ``Sales growth'' is the first difference in the log of sales; ``ROA'' is operating income before depreciation divided by lagged book assets; ``Market leverage'' is total debt divided by the sum of total debt and market equity; ``Capex'' is capital expenditures divided by lagged book assets; ``R\&D'' is research and development expenses divided by lagged book assets; ``Tangibility'' is net property, plant, and equipment divided by book assets; ``Payout ratio'' is total payout including dividends and repurchases divided by market equity; ``Cash flows'' is net income plus depreciation and amortization divided by lagged book assets; ``IHS(citations)'' is the inverse hyperbolic sine-transformed number of citations the firm's patents receive one, two, or three years ahead. The column ``Dual -- Single'' shows differences between the means for dual-class and single-class firms. *, **, and *** represent significance at the 10\%, 5\%, and 1\% levels, respectively, based on standard errors clustered at the firm level.
\end{tablenotes}
        \label{tab:descriptive_stats}
\end{table}

\clearpage

\begin{landscape}
\begin{table}[!htbp]
    \captionsetup{font = large, justification=centering}
    \caption{\textbf{The Effects of Dual-Class Recapitalization and Stock Unification on Tobin's $q$}}
\label{tab:q_robustness}
\centering
\begin{threeparttable}
\footnotesize
\setlength{\tabcolsep}{4pt}
\renewcommand{\arraystretch}{1.10}

\begin{tabular*}{\textwidth}{@{\extracolsep{\fill}}l*{8}{c}}
\toprule \toprule
& (1) & (2) & (3) & (4) & (5) & (6) & (7) & (8) \\
\midrule
\multicolumn{8}{l}{\textit{Panel A: Recapitalization}} \\
Pre-event 
& 0.042 & -0.059 & 0.038 & 0.046 & -0.011 & 0.056 & -& 0.075  \\
& (0.090) & (0.103) & (0.090) & (0.121)  & (0.099) & (0.099) & -& (0.104)  \\
$\tau = 0$ to $3$
& 0.189** & 0.141 & 0.182** & 0.265** & 0.153** & 0.187** & -& 0.249**  \\
& (0.083) & (0.089) & (0.083) & (0.113) & (0.075) & (0.090) & -& (0.101)\\
$\tau = 4$ to $7$
& 0.180* & 0.149 & 0.178* & 0.301** & 0.143* & 0.168 & -& 0.267**  \\
& (0.108) & (0.122) & (0.109) & (0.145)  & (0.085) & (0.120) & -& (0.131)  \\
$\tau = 8$ to $11$
& 0.057 & 0.064 & 0.052 & 0.133 & 0.017 & 0.024 & - & 0.105 \\
& (0.113) & (0.134) & (0.113) & (0.188) & (0.085) & (0.125) & -& (0.140) \\

\addlinespace[0.75em]
\multicolumn{8}{l}{\textit{Panel B: Unification}} \\
Pre-event
& -0.021 & -0.158 & -0.004 & 0.040  & -0.162 & -0.015 & -0.125 & -0.072 \\
& (0.242) & (0.222) & (0.270) & (0.256) & (0.146) & (0.234) & (0.268) & (0.302)  \\
$\tau = 0$ to $3$
& -0.415 & -0.610** & -0.445 & -0.276 & -0.524*** & -0.402* & -0.473 & -0.388  \\
& (0.268) & (0.305) & (0.286) & (0.296) & (0.152) & (0.241) & (0.292) & (0.332)  \\
$\tau = 4$ to $7$
& -0.670** & -0.814** & -0.601** & -0.533 & -0.845*** & -0.697*** & -0.694** & -0.687*  \\
& (0.284) & (0.321) & (0.305) & (0.327) & (0.181) & (0.270) & (0.323) & (0.364)  \\
$\tau = 8$ to $11$
& -0.436 & -0.306 & -0.337 & -0.236  & -0.822*** & -0.457 & -0.596 & -0.447  \\
& (0.329) & (0.419) & (0.332) & (0.371)  & (0.221) & (0.317) & (0.373) & (0.435)   \\

\addlinespace[0.75em]
\midrule
Matched sample 
&  & \checkmark &  &  &  &  &      \\
Carveout exclusion 
&  &  & \checkmark &  &  &  &     \\
Equal price
&  &  &  & \checkmark  &  &  &     \\
Sun \& Abraham
&  &  &  & & \checkmark &  &  &   \\
Weighted by log assets 
&  &  &  &  & & \checkmark & &    \\
Exclude family firms
&  &  &  &  &  & & \checkmark &    \\
Winsorize at 1\%
&  &  &  &  &  &  & & \checkmark   \\

\bottomrule \bottomrule
\end{tabular*}

\end{threeparttable}

\vspace{0.5em}

\footnotesize
\justifying
\noindent\textit{Notes:}
This table presents event-study estimates of the effects of dual-class recapitalization (Panel A) and stock unification (Panel B) on Tobin's $q$ from a version of Eq. (\ref{eq:dd}) which uses groups of event-time indicators. ``$\tau$'' refers to the number of years following recapitalization or unification while ``Pre-event'' refers to the period from $\tau =$ -5 to -2. Column (1) shows the baseline estimates. Each of columns (2) through (8) uses an identical specification to column (1) except for one change to the baseline. Column (2) employs a matched sample. Column (3) excludes dual-class firms that are equity carveouts. Column (4) assumes the same price between superior and inferior classes of shares. Column (5) performs the event-study analysis using the \citet{sunabraham2021eventstudy} estimator. Column (6) performs the DD regression weighted by the log of total assets. Column (7) in Panel B excludes unification events in which a founder or controlling family member steps down from management between two years before and after the unification. Column (8) Winsorizes variables at the 1\% tails. Standard errors clustered at the firm level are reported in parentheses, except for column (2) which clusters at the matched pair level. *, **, and *** represent significance at the 10\%, 5\%, and 1\% levels, respectively.
\end{table}
\end{landscape}

\clearpage

\begin{table}[]
    \small
    \captionsetup{font = large, justification=centering}
    \caption{\textbf{The Effects of Dual-Class Recapitalization on Tobin's $q$: \\ AMEX and Nasdaq vs. NYSE Firms Pre-1986}}
\begin{center}
    
\begin{tabular}{l*{2}{c}}
\toprule \toprule
                    
                    &\multicolumn{1}{c}{(1)}&\multicolumn{1}{c}{(2)}\\
                    \multicolumn{1}{c}{Dependent variable:} 
                    & \multicolumn{2}{c}{Tobin's $q$} \\
\midrule
Pre-event    &      0.031&     0.027   \\
                    &     (0.121)   &    (0.122)   \\
\addlinespace
$\tau =$ 0 to 4       &      0.073 &       0.077  \\
                    &     (0.084)   &     (0.085)   \\
\addlinespace
$\tau =$ 5 to 11        &      0.271*  &       0.285** \\
                    &     (0.136)   &     (0.137)   \\
\midrule
Firm fixed effects &    Y        &   Y \\
Matched pair fixed effects &    N        &   Y \\
Year-quarter fixed effects &    Y        &   Y \\
Firm-level controls &    Y        &   Y \\
$R^2$                & 0.908 & 0.909  \\
Observations        &      470   &      470   \\
\bottomrule \bottomrule
\end{tabular}

\begin{tablenotes}
\small
\item \textit{Notes:} This table presents event-study estimates of the effect of dual-class recapitalization on Tobin's \textit{q} from a modified, quarterly version of Eq. (\ref{eq:dd}), with details described in Section \ref{sec:event_study_add}. The treated group consists of firms listed on the AMEX and Nasdaq that recapitalized between 1980 and September 1986, while the control group consists of matched firms listed on the NYSE that remain single-class during the same period. The event study requires that both the treated and control groups of firms exist during the $[t-2, t+2]$ window. ``$\tau$'' refers to the number of quarters following recapitalization while ``Pre-event'' refers to the period from $\tau =$ -5 to -2. Standard errors clustered at the matched pair level are reported in parentheses. *, **, and *** represent significance at the 10\%, 5\%, and 1\% levels, respectively, based on wild cluster bootstrapped $p$-values (\cite{cameron2008bootstrap}). 

\end{tablenotes}
\label{tab:event_study_nyse}
\end{center}
\end{table}

\clearpage
\begin{table}[]
    \centering
    \small
    \captionsetup{font = large, justification=centering}
    \caption{\textbf{Market Reactions to Dual-Class Recapitalizations and Stock Unifications Conditional on Maturity}}
\begin{tabular}{lcc}
                                        & \multicolumn{1}{l}{}    & \multicolumn{1}{l}{}   \\ \toprule \toprule
\multicolumn{1}{c}{} & (1) & (2) \\ 
\multicolumn{1}{c}{Dependent variable:} & \multicolumn{2}{c}{Cumulative abnormal return $[-3, +3]$} \\
\multicolumn{1}{c}{Event:} & \multicolumn{1}{c}{Dual-class recapitalization} & \multicolumn{1}{c}{Dual-class stock unification} \\ 
\midrule
$\mathbb{1}[\text{age $>$ 12}]$ & -0.035** & 0.052** \\
& (0.015) & (0.021) \\
Constant & 0.029** & -0.003 \\
& (0.011) & (0.016) \\
\midrule
$R^2$ & 0.050 & 0.094 \\
Observations & 116 & 61 \\
\bottomrule \bottomrule
                                        & \multicolumn{1}{l}{}    & \multicolumn{1}{l}{}  
\end{tabular}
\begin{tablenotes}
\small
\item \textit{Notes:} This table examines the effects of dual-class recapitalization (column (1)) and stock unification (column (2)) announcements on stock returns conditional on firm maturity. The dependent variable is the cumulative abnormal return from three days before to three days after the announcement of a dual-class recapitalization and stock unification. ``$\mathbb{1}[\text{age $>$ 12}]$ '' is an indicator variable equal to one if firm age is greater than 12 years, the sample median for dual-class firms, and zero otherwise. White robust standard errors are reported in parentheses. ** represents significance at the 5\% level.
\end{tablenotes}
\label{tab:recaps_and_unifications}
\end{table}

\begin{landscape}
\begin{table}[!htbp]
\centering
    \captionsetup{font = large, justification=centering}
    \caption{\textbf{The Effects of Dual-Class Recapitalization and Stock Unification on Innovative Output}}
\label{tab:patent_robustness}
\begin{threeparttable}
\footnotesize
\setlength{\tabcolsep}{4pt}
\renewcommand{\arraystretch}{1.10}

\begin{tabular*}{\textwidth}{@{\extracolsep{\fill}}l*{9}{c}}
\toprule \toprule
& (1) & (2) & (3) & (4) & (5) & (6) & (7) & (8) & (9) \\
\midrule

\multicolumn{10}{l}{\textit{Panel A: Recapitalization}} \\
Pre-event 
& 0.138 & 0.163  & 0.129 & 0.131 & 0.165 & 0.024  & -0.051     &  - & 0.169 \\
& (0.119)  & (0.144) & (0.119) & (0.116) & (0.136) & (0.025)  & (0.098)     &  - & (0.120)  \\
$\tau = 0$ to $3$
& 0.267*  & 0.322 & 0.254 & 0.224* & 0.293* & 0.052* & [-0.201,-0.096]  &  -  & 0.264 \\
& (0.159) & (0.207) & (0.159) & (0.133) & (0.172) & (0.031) & [(0.156),(0.159)]  &  - & (0.161)  \\
$\tau = 4$ to $7$
& 0.355*  & 0.529** & 0.350* & 0.291* & 0.393* & 0.046 & [0.039,0.101]  &  -  & 0.359* \\
& (0.200) & (0.258) & (0.201) & (0.173) & (0.208) & (0.037) & [(0.232),(0.230)]    &-  & (0.202)   \\
$\tau = 8$ to $11$
& 0.673***  & 0.800*** & 0.665*** & 0.627*** & 0.731*** & 0.127*** & [-0.234,-0.035]  &  - & 0.711***  \\
& (0.212) & (0.289) & (0.212) & (0.197) & (0.220) & (0.040) & [(0.264),(0.258)]  &  - & (0.219)  \\

\addlinespace[0.75em]
\multicolumn{10}{l}{\textit{Panel B: Unification}} \\
Pre-event 
& 0.026  & -0.077  & 0.037 & 0.074  & 0.041  & -0.003  & 0.346     & -0.047  &  -   \\
& (0.149)  & (0.177)  & (0.157)  & (0.118)  & (0.149)  & (0.032)  & (0.269)  & (0.164)  &  -   \\
$\tau = 0$ to $3$ 
& -0.365 & -0.451  & -0.372  & -0.276* & -0.397* & -0.063  & [-0.148,-0.086]   &  -0.298  &  -  \\
& (0.223) & (0.289)  & (0.239)  & (0.161)  & (0.229)  & (0.047)  & [(0.277),(0.295)]  & (0.256)  &  -   \\
$\tau = 4$ to $7$ 
& -0.708** & -0.648  & -0.674*  & -0.566** & -0.759** & -0.118*  & [-0.369,-0.327]   &  -0.487 &  -   \\
& (0.338) & (0.416)  & (0.352)  & (0.220)  & (0.362)  & (0.068)  & [(0.516),(0.514)]  & (0.357)  &  -   \\
$\tau = 8$ to $11$ 
& -0.775* & -1.187*  & -0.881*  & -0.796*** & -0.719 & -0.062  & [-0.511,-0.451]   &  -0.474  &  -  \\
& (0.451) & (0.652)  & (0.452)  & (0.251)  & (0.487)  & (0.087)  & [(0.495),(0.481)]  & (0.423)  &  -   \\

\addlinespace[0.75em]
\midrule
Matched sample 
&  & \checkmark &  &  &  &  &  &  &   \\
Carveout exclusion 
&  &  & \checkmark &  &  &  &  & &     \\
Sun \& Abraham 
&  &  &  & \checkmark &  &  &  &  &    \\
Weighted by log assets 
&  &  &  &  &  \checkmark &  &  & &     \\
Extensive margin 
&  &  &  &  &  &  \checkmark &  &  &    \\
Intensive margin with Lee bounds
&  &  &  &  &  &  & \checkmark &  &    \\
Exclude family firms
&  &  &  &  &  &  &   & \checkmark &  \\
Supp. patent linkage, all firms
&  &  &  &  &  &  &   &  & \checkmark  \\

\bottomrule \bottomrule
\end{tabular*}
\end{threeparttable}

\vspace{0.5em}

\footnotesize
\justifying
\noindent\textit{Notes:}
This table presents event-study estimates of the effects of dual-class recapitalization (Panel A) and stock unification (Panel B) on innovative output from a version of Eq. (\ref{eq:dd}) which uses groups of event-time indicators, and the inverse hyperbolic sine-transformed number of citations the firm's patents receive one, two, or three years ahead as dependent variable. ``$\tau$'' refers to the number of years following recapitalization or unification while ``Pre-event'' refers to the period from $\tau =$ -5 to -2. Column (1) shows the baseline estimates. Each of columns (2) through (9) uses an identical specification to column (1) except for one change to the baseline. Column (2) employs a matched sample. Column (3) excludes dual-class firms that are equity carveouts. Column (4) performs the event-study analysis using the \citet{sunabraham2021eventstudy} estimator. Column (5) performs the DD regression weighted by the log of total assets. Columns (6)-(7) examine the extensive and the intensive margins of innovation separately. Column (7) estimates [lower, upper bounds] on the intensive margin effect following \cite{lee2009bounds}. Column (8) in Panel B excludes unification events in which a founder or controlling family member steps down from management between two years before and after the unification. Column (9) uses the supplementary patent linkage from \citet{ma2025obsolescence} for all firms, instead of only ever-dual class firms. This change does not affect the sample in Panel B (unification), which consists of ever-dual class firms. Standard errors clustered at the firm level are reported in parentheses, except for column (2) which clusters at the matched pair level. *, **, and *** represent significance at the 10\%, 5\%, and 1\% levels, respectively.

\end{table}
\end{landscape}

\clearpage
\begin{table}[htbp]
\def\sym#1{\ifmmode^{#1}\else\(^{#1}\)\fi}
    \captionsetup{font = large, justification=centering}
    \caption{\textbf{The Effects of Recapitalization and Unification by Specificity of Investments}}
    \centering
    
\begin{tabular}{l*{4}{c}}
\toprule \toprule
                    \multicolumn{5}{c}{\textbf{Panel A: Tobin's \textit{q}}} \\
\midrule
                    &\multicolumn{2}{c}{Recapitalization}        &\multicolumn{2}{c}{Unification}    \\
\midrule
                    &\multicolumn{1}{c}{(1)}&\multicolumn{1}{c}{(2)}&\multicolumn{1}{c}{(3)}&\multicolumn{1}{c}{(4)}\\
Specificity:            &\multicolumn{1}{c}{High}&\multicolumn{1}{c}{Low}&\multicolumn{1}{c}{High}&\multicolumn{1}{c}{Low}\\
\midrule
Pre-event        &      0.062   &      -0.087   &       -0.122   &       0.048   \\
                    &     (0.112)   &     (0.138)   &     (0.318)   &     (0.414)   \\
\addlinespace
$\tau =$ 0 to 3      &       0.343***  &       0.033   &      -0.851   &       -0.153   \\
                    &     (0.122)   &     (0.118)   &     (0.523)   &     (0.321)   \\
\addlinespace
$\tau =$ 4 to 7    &       0.258*   &     0.096   &       -1.043**   &       -0.409  \\
                    &     (0.151)   &     (0.164)   &     (0.520)   &     (0.385)   \\
\addlinespace
$\tau =$ 8 to 11    &       0.059   &      -0.003  &       -0.649 &       -0.306  \\
                    &     (0.173)   &     (0.155)   &     (0.491)   &     (0.564)   \\
\midrule

           \multicolumn{5}{c}{} \\

\toprule
                    \multicolumn{5}{c}{\textbf{Panel B: IHS(citations)}} \\
\midrule
                    &\multicolumn{2}{c}{Recapitalization}        &\multicolumn{2}{c}{Unification}    \\
 \midrule
                    &\multicolumn{1}{c}{(1)}&\multicolumn{1}{c}{(2)}&\multicolumn{1}{c}{(3)}&\multicolumn{1}{c}{(4)}\\
Specificity:               &\multicolumn{1}{c}{High}&\multicolumn{1}{c}{Low}&\multicolumn{1}{c}{High}&\multicolumn{1}{c}{Low}\\

\midrule

Pre-event        &      0.133   &      0.160   &       -0.004  &       -0.003   \\
                    &     (0.193)   &     (0.158)   &     (0.241)   &     (0.245)   \\
\addlinespace
$\tau =$ 0 to 3      &       0.224  &       0.222  &      -0.640**   &       -0.116   \\
                    &     (0.239)   &     (0.209)   &     (0.317)   &     (0.327)   \\
\addlinespace
$\tau =$ 4 to 7    &       0.373   &     0.227   &       -1.058**   &       -0.349  \\
                    &     (0.306)   &     (0.235)   &     (0.437)   &     (0.512)   \\
\addlinespace
$\tau =$ 8 to 11    &       0.443   &      0.676**  &       -1.433** &       -0.107  \\
                    &     (0.293)   &     (0.269)   &     (0.718)   &     (0.574)   \\
\bottomrule \bottomrule
\end{tabular}
\begin{tablenotes}
\small
\item \textit{Notes:} This table presents event-study estimates of the effects of dual-class recapitalization and stock unification on Tobin's $q$ and innovative output on split samples by the specificity of investments from a version of Eq. (\ref{eq:dd}) which uses groups of event-time indicators. ``$\tau$'' refers to the number of years following recapitalization or unification while ``Pre-event'' refers to the period from $\tau =$ -5 to -2. See Section \ref{sec:redeployability} for details of measuring investment specificity that combines capital and labor specificity. Panels A and B show estimates for Tobin's $q$ and the inverse hyperbolic sine (IHS)-transformed number of citations the firm's patents receive one, two, or three years ahead, respectively. Columns (1)-(2) show the results for recapitalization  and columns (3)-(4) show the results for unification. Columns (1) and (3) report each result for firms with above-median specificity, and columns (2) and (4) report each result for firms with below-median specificity. Standard errors clustered at the firm level are reported in parentheses. *, **, and *** represent significance at the 10\%, 5\%, and 1\% levels, respectively.
\end{tablenotes}
\label{tab:dd_redeployability}
\end{table}


\begin{table}[]
    \centering
    \small
    \captionsetup{font = large, justification=centering}
    \caption{\textbf{Investment-$q$ Sensitivity for Dual-Class vs. Single-Class Firms over Maturity}}
    {
\def\sym#1{\ifmmode^{#1}\else\(^{#1}\)\fi}
\begin{tabular}{l*{3}{c}}
\toprule \toprule
                    &\multicolumn{1}{c}{(1)}&\multicolumn{1}{c}{(2)}&\multicolumn{1}{c}{(3)}\\
                    \multicolumn{1}{c}{Dependent variable:} &\multicolumn{3}{c}{Capex}\\
                    \multicolumn{1}{c}{Sample:} &\multicolumn{1}{c}{Full}&\multicolumn{1}{c}{Young}&\multicolumn{1}{c}{Mature}\\
\midrule
$ q$             &       0.003\sym{***}&       0.003\sym{***}&       0.003\sym{***}\\
                    &     (0.000)         &     (0.001)         &     (0.001)         \\
\addlinespace
$q$ $\times$ Dual   &      -0.001         &       0.003         &      -0.003\sym{**} \\
                    &     (0.001)         &     (0.002)         &     (0.001)         \\
\midrule
Firm $\times$ Dual fixed effects&           Y         &           Y         &           Y         \\
Year $\times$ Dual fixed effects&           Y         &           Y         &           Y         \\
Cash flow $\times$ Dual controls  &           Y         &           Y         &           Y         \\
$R^2$               &       0.634         &       0.730         &       0.645         \\
Observations        &      32,038       &      17,545        &      14,493        \\
\midrule 
$q$ $\times$ Dual $\times$ (Mature -- Young)   &  \multicolumn{3}{c}{-0.006**} \\
$p$-value   &  \multicolumn{3}{c}{0.014} \\
\bottomrule \bottomrule
\end{tabular}
}
\begin{tablenotes}
\small
\item \textit{Notes:} This table examines the relationship between dual-class shares and sensitivities of investment to investment opportunities conditional on firm maturity. The dependent variable ``Capex'' is capital expenditures divided by lagged book assets. ``Dual'' is an indicator variable equal to one if a firm-year has multiple classes of shares with different voting rights, and zero otherwise. The full sample includes firm-years with sales growth in the first quintile of the distribution over 1971--2022. The ``Young'' firm sample includes firm-years with ages less than or equal to 12 years, the sample median for dual-class firms, while the ``Mature'' firm sample includes firm-years with ages greater than 12 years. ``$q$'' is Tobin's $q$ as defined in Table \ref{tab:descriptive_stats}. ``Cash flow'' is net income plus depreciation and amortization divided by lagged book assets. All regressions include ``Cash flow'' interacted with ``Dual,'' as well as the standalone variable, as controls. Standard errors clustered at the firm level are reported in parentheses. *, **, and *** represent significance at the 10\%, 5\%, and 1\% levels, respectively.
\end{tablenotes}
\label{tab:invest_q_sensitivity}
\end{table}

\clearpage

\begin{table}[]
    \centering
    \small
    \captionsetup{font = large, justification=centering}
    \caption{\textbf{Voting Premium over Firm Maturity}}
\label{tab:voting_premium0}

{
\def\sym#1{\ifmmode^{#1}\else\(^{#1}\)\fi}
\begin{tabular}{lccc}
\toprule \toprule
                    & (1) & (2) & (3) \\ 
\multicolumn{1}{c}{Dependent variable:} 
                    & \multicolumn{3}{c}{Voting premium} \\ 
\midrule
${\Large \mathbb{1}}[\text{age in quartile 2}]$
                    & 0.032        & 0.025        & 0.022        \\
                    & (0.019)      & (0.015)      & (0.015)      \\
${\Large \mathbb{1}}[\text{age in quartile 3}]$
                    & 0.034 & 0.048\sym{*} & 0.043 \\
                    & (0.023)      & (0.025)      & (0.027)      \\
${\Large \mathbb{1}}[\text{age in quartile 4}]$
                    & 0.054\sym{**}& 0.096\sym{**}& 0.085\sym{*}\\
                    & (0.026)      & (0.042)      & (0.048)      \\
Log market equity
                    & -0.010\sym{*}& -0.024\sym{***}& -0.025\sym{***}\\
                    & (0.006)      & (0.007)      & (0.007)      \\
Shapley value
                    & 0.057\sym{***}& -0.000        & 0.002        \\
                    & (0.022)      & (0.021)      & (0.022)      \\
Log volume (sup. to inf.)
                    & 0.006 & 0.005        & 0.006        \\
                    & (0.006)      & (0.005)      & (0.005)      \\
\midrule
Firm fixed effects   & N & Y & Y \\
Year fixed effects   & Y & Y & Y \\
$R^2$                & 0.073 & 0.535 & 0.535 \\
Observations         & 1,182 & 1,182 & 1,182 \\
\bottomrule \bottomrule
\end{tabular}
}
\begin{tablenotes}
\small
\item \textit{Notes:} This table examines how the voting premium changes with firm maturity using a sample of dual-class firms for which information on stock prices and trading volume is available from CRSP for both inferior- and superior-voting shares, and ownership structure information is available from 1980 through 2022. The dependent variable ``Voting premium'' is defined as $(P_A - P_B)/(P_B - rP_A)$, where $P_A$ ($P_B$) is the price of the superior- (inferior-) voting shares and $r$ < 1 is the relative number of votes between the inferior- and superior-voting shares. Firm age is measured as years since IPO and grouped into quartiles. Firms in the first age quartile serve as the omitted group. ``Log market equity'' is the log of market equity; ``Shapley value'' is the Shapley value of the ownership structure; ``Log volume (sup. to inf.)'' is the log of the ratio of trading volume between the superior- and inferior-voting shares. In column (2), we include year fixed effects using \cite{deaton1997}'s normalization. In column (3), we use an alternative restriction that normalizes one year fixed effect to zero. Standard errors clustered at the firm level are reported in parentheses. *, **, and *** denote statistical significance at the 10\%, 5\%, and 1\% levels, respectively.
\end{tablenotes}
\end{table}

\clearpage

\clearpage
\setstretch{1}
\newpage
\newgeometry{left=1cm,bottom=2cm, right=1cm, top=1cm}
\setcounter{figure}{0} 
\renewcommand{\thefigure}{A\arabic{figure}}

\clearpage

\begin{figure}[t]
    \centering
    \captionsetup{font = large, justification=centering}
    \caption{\\\textbf{Event Study of Dual-Class Stock Unification by Sunset Provision: Tobin's \textit{q}}}
    \includegraphics[width = 0.85\textwidth]{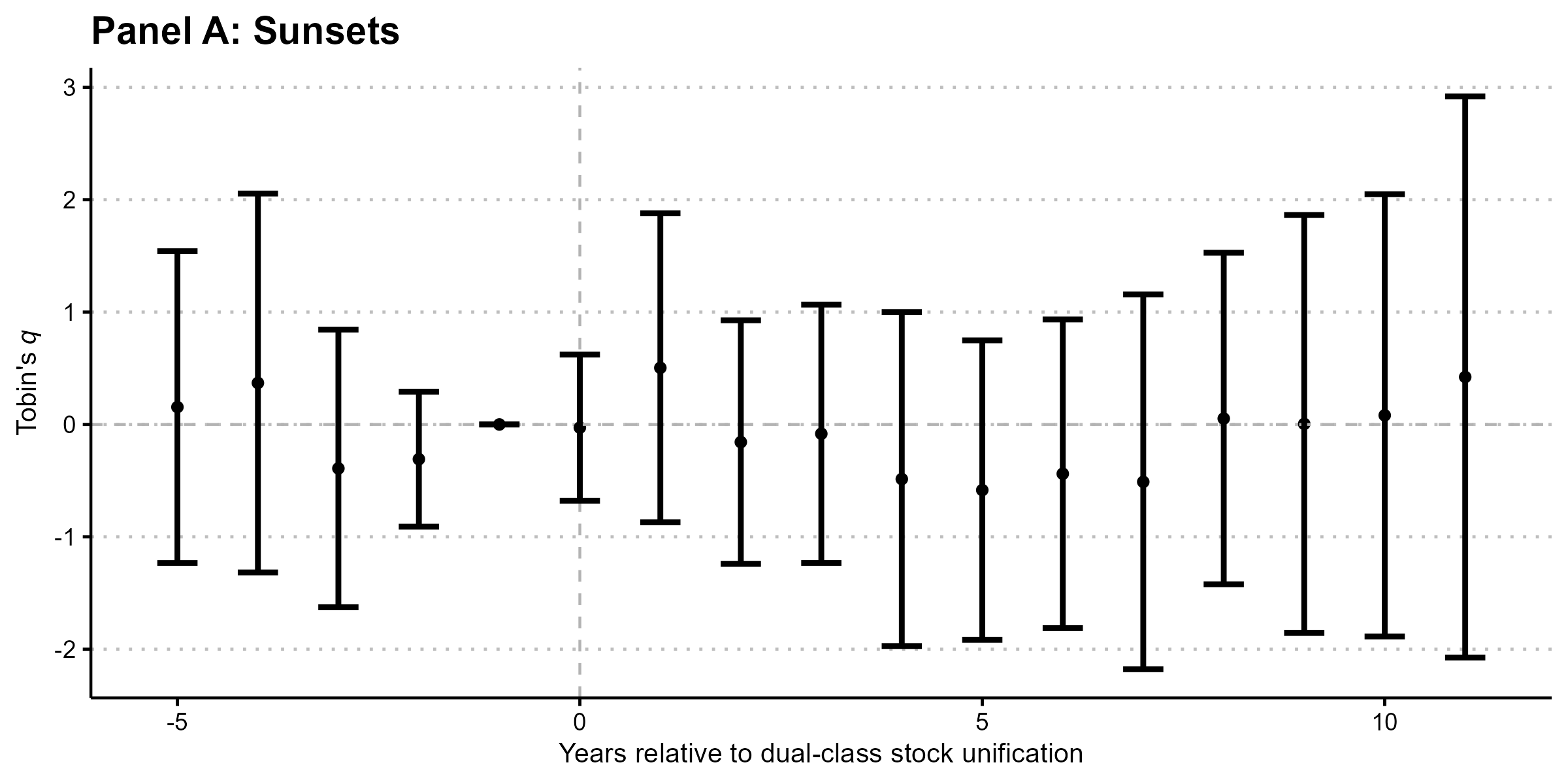}
    \includegraphics[width = 0.85\textwidth]{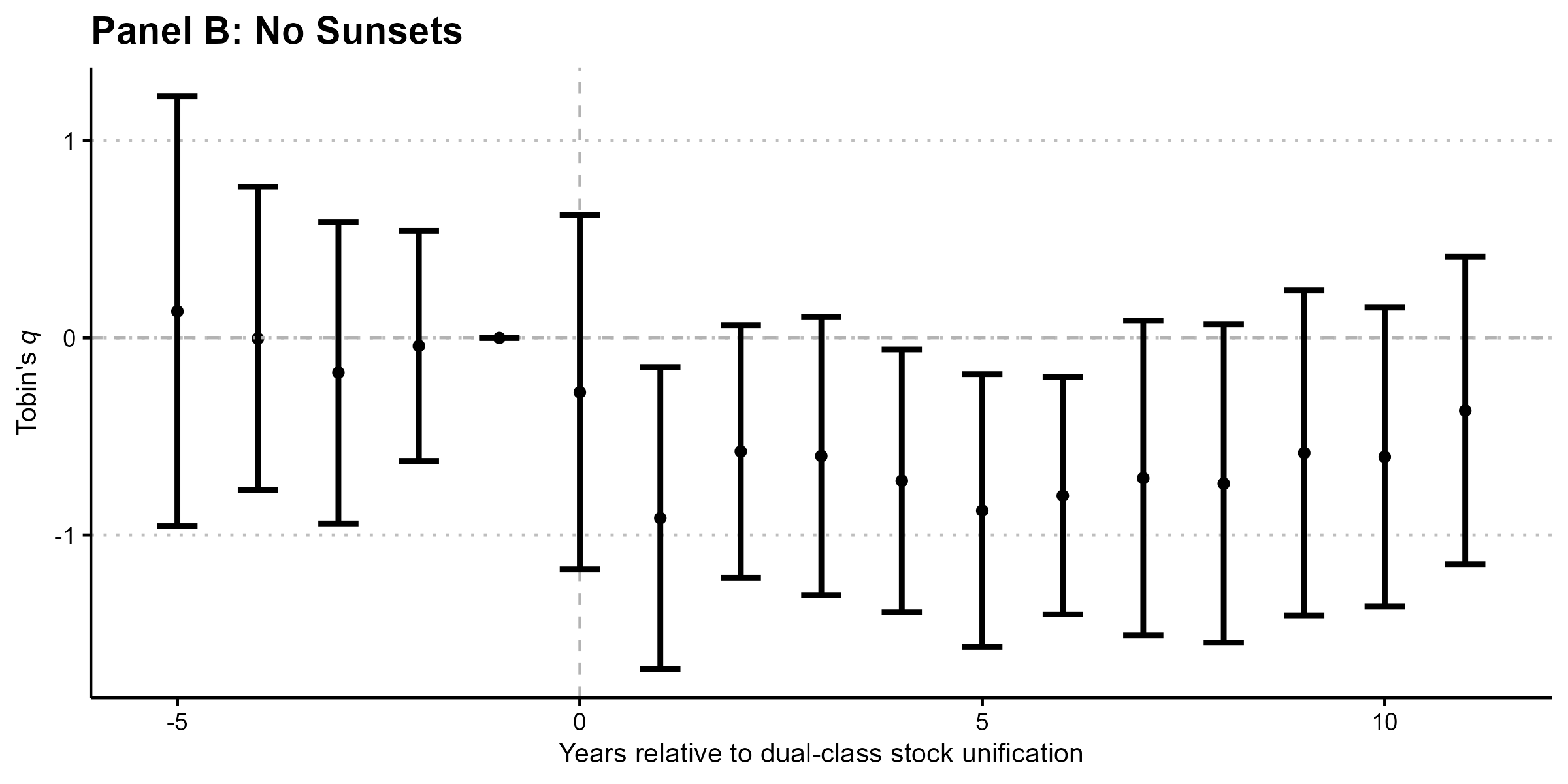}
    \captionsetup{font = small, justification=justified}
    \caption*{\textit{Notes}: This figure plots the stacked dynamic difference-in-differences event-study estimates of the effect of dual-class stock unification on Tobin's \textit{q} for firms with sunset provisions (Panel A) and without (Panel B) from Eq. (\ref{eq:dd}). The event study requires that both the treated and control groups of firms exist during the $[t-2, t+2]$ window. The markers on the lines show point estimates. The 90\% confidence intervals in vertical lines are based on standard errors clustered at the firm level.}
    \label{fig:appendix_event_study_sunsets}
\end{figure}{}

\clearpage
\setstretch{1}

\newpage
\newgeometry{left=1cm,bottom=2cm, right=1cm, top=1cm}
\setcounter{table}{0} 
\renewcommand{\thetable}{A\arabic{table}} 

\clearpage

\begin{table}[]
    \small
    \captionsetup{font = large, justification=centering}
    \caption{\textbf{Post-IPO Within-Firm Dynamics of Valuation and Innovative Output}}
\begin{center}
    
\begin{tabular}{l*{2}{c}}
\toprule \toprule
                    &\multicolumn{1}{c}{(1)}&\multicolumn{1}{c}{(2)}\\
        \multicolumn{1}{c}{Dependent variable:}            &\multicolumn{1}{c}{Tobin's \textit{q}}&\multicolumn{1}{c}{IHS(citations)}\\
\midrule
$\text{Dual} \times {\Large \mathbb{1}}[\text{age in quartile 2}]$    &      -0.518***&     0.040   \\
                    &     (0.194)   &    (0.110)   \\
\addlinespace
$\text{Dual} \times {\Large \mathbb{1}}[\text{age in quartile 3}]$      &      -0.691** &       0.546**  \\
                    &     (0.303)   &     (0.270)   \\
\addlinespace
$\text{Dual} \times {\Large \mathbb{1}}[\text{age in quartile 4}]$       &      -0.783*  &       1.343***\\
                    &     (0.461)   &     (0.335)   \\
\midrule
Firm fixed effects &    Y        &   Y \\
SIC3 $\times$ Year $\times$ IPO cohort fixed effects &    Y        &   Y \\
Firm-level controls &    Y        &   Y \\
$R^2$ & 0.706 & 0.836 \\
Observations        &      100,021   &      119,263   \\
\bottomrule \bottomrule
\end{tabular}

\begin{tablenotes}
\small
\item \textit{Notes:} This table presents the results of estimating within-firm dynamics of Tobin's \textit{q} and innovative output for firms that IPO with dual-class shares relative to firms that IPO with single-class shares. We estimate the following equation:

\begin{equation*}
    y_{ijth} = \alpha_{i} + \alpha_{jth} + \sum_{\substack{2}}^{4} \lambda_{n}{\Large \mathbb{1}}[\text{age in quartile}\,  n]_{it}  + \sum_{\substack{2}}^{4} \delta_{n}\textit{Dual}_{i} \times {\Large \mathbb{1}}[\text{age in quartile}\,  n]_{it}  + \gamma'\textit{X}_{it} + \epsilon_{ijth},
\end{equation*}

where ${\Large \mathbb{1}}[\text{age in quartile}\,  n]_{it}$ is an indicator variable equal to one if firm $i$'s age in year $t$ is in quartile $n$; $y_{ijth}$ is either Tobin's \textit{q} or the inverse hyperbolic sine (IHS)-transformed number of citations that patents receive one, two, or three years ahead for firm $i$ of IPO cohort $h$ in three-digit SIC industry $j$ in year $t$; $\alpha_{i}$ and $\alpha_{jth}$ represent firm $i$ and industry $j$ by year $t$ by IPO cohort $h$ fixed effects; $\textit{Dual}_{i}$ is an indicator variable equal to one if firm $i$ IPOs with dual-class shares and zero if with single-class shares; $\textit{X}_{it}$ is a vector of control variables in Eq. (\ref{eq:q}); and $\epsilon_{ijth}$ represents random errors clustered at the firm level. Firm-year observations are dropped from the sample upon switching from one type of governance structure to the other. Column (1) reports results for Tobin's \textit{q} and column (2) for IHS(citations). Standard errors clustered at the firm level are reported in parentheses. *, **, and *** represent significance at the 10\%, 5\%, and 1\% levels, respectively.
\end{tablenotes}
\label{tab:ipo_dynamics}
\end{center}
\end{table}


\numberwithin{equation}{section} 
\renewcommand\thefigure{A.\arabic{figure}} 
\renewcommand\thetable{A.\arabic{table}} 
\setcounter{figure}{0} 
\setcounter{table}{0} 
    \clearpage
    \begin{appendices}

    \section{Theory Appendix\label{appendix:theory}}

\justifying
This appendix provides a detailed analysis of the theoretical model
presented in Section \ref{sec:theory}. We consider a generalized governance
structures. Under governance structure $n>0$, the founder can be fired
only after period $n$. Thus, $n=0$ corresponds to single-class shares, and $%
n=\infty $ to dual-class shares.

Once the founder exerts effort, his impact on the shareholder value is only
through the consumption of private benefits. Therefore, shareholders fire
the founder as soon as they can, at period $t=n$ (the founder has no
incentives to resign voluntarily). Thus, $b_{t}=0$ for all $t\geq n$. The
share price at time $t\geq 0$ given effort $x$ is%
\begin{equation*}
p_{n}\left( x,t\right) =x\left[ D\left( t\right) -B_{n}\left( t\right) %
\right] 
\end{equation*}%
where%
\begin{eqnarray*}
D\left( t\right)  &\equiv &\int_{t}^{\infty }e^{-r\left( s-t\right) }e^{gs}ds
\\
\text{\ }B_{n}\left( t\right)  &\equiv &\int_{t}^{\infty }e^{-r\left(
s-t\right) }e^{gs}b_{s}\cdot 1_{s<n}ds,
\end{eqnarray*}%
and $r>g$ is the firm's cost of capital and $g>0$ is the firm's growth rate.
Notice 
\begin{eqnarray*}
D\left( t\right)  &=&\frac{e^{gt}}{r-g} \\
B_{n}\left( t\right)  &=&\phi e^{gt}\left[ \frac{1-e^{-\left( r-g\right)
\left( n-t\right) }}{r-g}-e^{-\tau t}\frac{1-e^{-\left( r-g+\tau \right)
\left( n-t\right) }}{r-g+\tau }\right] \cdot 1_{t<n}.
\end{eqnarray*}%
Moreover, $B_{n}\left( t\right) \in \lbrack 0,D_{t})$ and it increases in $n$
for $t<n$.

\subsection{Proofs of main results}

\begin{lemma}
\label{res:effort} Given governance structure $n$, the founder's optimal
effort at $t=0$ is 
\begin{equation}
x_{n}^{\ast }=\frac{\alpha \Lambda D\left( 0\right) +\left(
\gamma -\alpha \Lambda \right) B_{n}\left( 0\right) }{\kappa },
\label{effort}
\end{equation}%
and $x_{n}^{\ast }$ increases in $n$.
\end{lemma}

\begin{proof}
The founder's expected utility at $t=0$ is%
\begin{equation*}
u_{t=0}\left( x,n\right) =x\left[ \alpha \Lambda D\left( 0\right) +\left(
\gamma -\alpha \Lambda \right) B_{n}\left( 0\right) \right] +\alpha \left(
1-\Lambda \right) p_{n}\left( E\left[ x\right] ,0\right) -\frac{1}{2}\kappa
x^{2}.
\end{equation*}%
Since effort is not observed by outside shareholders, the proceeds from
selling shares do not depend on the actual choice of effort, but only on the
expectations of shareholders. The founder's optimal effort is given by $%
x_{n}^{\ast } $. Since $B_{n}\left( 0\right) $ increases in $n
$ and $\gamma >\alpha \Lambda $, $x_{n}^{\ast }$ increases
in $n$. That is, the founder exerts more effort when he has more job
security. Specifically, the founder exerts more effort under dual-class
shares.
\end{proof}

\ \ 

\begin{proof}[\textbf{Proof of Proposition \protect\ref{price_dynamic}}]
Let%
\begin{equation*}
p_{n}^{\ast }\left( t\right) =x_{n}^{\ast }\left[ D\left( t\right)
-B_{n}\left( t\right) \right] .
\end{equation*}%
Recall $B_{0}\left( t\right) =0$. Let $n>0$. Then, 
\begin{eqnarray*}
p_{0}^{\ast }\left( t\right)  &<&p_{n}^{\ast }\left( t\right)
\Leftrightarrow  \\
\frac{\alpha \Lambda D\left( 0\right) +\left( \gamma -\alpha \Lambda \right)
B_{0}\left( 0\right) }{\kappa }\left[ D\left( t\right) -B_{0}\left( t\right) %
\right]  &<&\frac{\alpha \Lambda D\left( 0\right) +\left( \gamma -\alpha
\Lambda \right) B_{n}\left( 0\right) }{\kappa }\left[ D\left( t\right)
-B_{n}\left( t\right) \right] \Leftrightarrow  \\
\frac{D\left( 0\right) }{B_{n}\left( 0\right) (\frac{D\left( t\right) }{%
B_{n}\left( t\right) }-1)} &<&\frac{\gamma }{\alpha \Lambda }-1.
\end{eqnarray*}%
Let $n=\infty $. Then%
\begin{equation*}
B_{\infty }\left( t\right) =\phi e^{gt}\left[ \frac{1}{r-g}-\frac{e^{-\tau t}%
}{r-g+\tau }\right] 
\end{equation*}%
and $\frac{D\left( t\right) }{B_{\infty }\left( t\right) }=\frac{1}{\phi }%
\frac{1}{1-\frac{r-g}{r-g+\tau }e^{-\tau t}}$ decreases in $t$, and hence
the left-hand side above increases in $t$. Thus, to prove the result it's
sufficient to require $p_{0}^{\ast }\left( 0\right) <p_{\infty }^{\ast
}\left( 0\right) $ and $p_{0}^{\ast }\left( \infty \right) >p_{\infty
}^{\ast }\left( \infty \right) $. The former holds if and only if%
\begin{equation*}
\frac{D\left( 0\right) }{D\left( 0\right) -B_{\infty }\left( 0\right) }<%
\frac{\gamma }{\alpha \Lambda }-1\Leftrightarrow \phi <\overline{\phi }%
\equiv \frac{\frac{\gamma }{\alpha \Lambda }-2}{\frac{\gamma }{\alpha
\Lambda }-1}\frac{r-g+\tau }{\tau }.
\end{equation*}%
Notice $\overline{\phi }>0$ only if $\frac{\gamma }{\alpha \Lambda }%
-2>0\Leftrightarrow \gamma -\alpha \Lambda >\alpha \Lambda $. The latter
holds if and only if%
\begin{equation*}
p_{0}^{\ast }\left( \infty \right) >p_{\infty }^{\ast }\left( \infty \right)
\Leftrightarrow \frac{D\left( 0\right) }{B_{\infty }\left( 0\right) \left(
\lim_{t\rightarrow \infty }\frac{D\left( t\right) }{B_{\infty }\left(
t\right) }-1\right) }>\frac{\gamma }{\alpha \Lambda }-1.
\end{equation*}%
Notice $\lim_{t\rightarrow \infty }\frac{D\left( t\right) }{B_{\infty
}\left( t\right) }=\frac{1}{\phi }$. Thus, we require%
\begin{equation*}
\phi >\underline{\phi }\equiv \frac{\frac{\gamma }{\alpha \Lambda }-2}{\frac{%
\gamma }{\alpha \Lambda }-1}-\frac{r-g}{\tau }\frac{\alpha \Lambda }{\gamma
-\alpha \Lambda}.
\end{equation*}%
Notice $\underline{\phi }<\min \left\{ 1,\overline{\phi }\right\} $ as
required.
\end{proof}

\ \ 

\begin{proof}[\textbf{Proof of Proposition \protect\ref{optimal_gov}}]
The founder's utility at $t=0$ is%
\begin{eqnarray*}
u_{n}^{\ast } &=&x_{n}^{\ast }\left[ \alpha \Lambda D\left( 0\right) +\left(
\gamma -\alpha \Lambda \right) B_{n}\left( 0\right) \right] +\alpha \left(
1-\Lambda \right) p_{n}^{\ast }\left( 0\right) -\frac{1}{2}\kappa
x_{n}^{\ast 2} \\
&=&\alpha \left( 1-\Lambda \right) p_{n}^{\ast }\left( 0\right) +\frac{1}{2}%
\kappa x_{n}^{\ast 2}.
\end{eqnarray*}%
The founder prefers dual class if and only if%
\begin{equation*}
u_{\infty }^{\ast }>u_{0}^{\ast }\Leftrightarrow p_{\infty }^{\ast }\left(
0\right) -p_{0}^{\ast }\left( 0\right) >\frac{x_{0}^{\ast 2}-x_{\infty
}^{\ast 2}}{2\alpha \left( 1-\Lambda \right) }\kappa .
\end{equation*}%
Since $x_{0}^{\ast }<x_{\infty }^{\ast }$, if $p_{\infty }^{\ast }\left(
0\right) >p_{0}^{\ast }\left( 0\right) $ then it must be $u_{\infty }^{\ast
}>u_{0}^{\ast }$. Substituting in the terms for $x_{0}^{\ast },$ $x_{\infty
}^{\ast },$ $p_{\infty }^{\ast }\left( 0\right) ,$ $p_{0}^{\ast }\left(
0\right) $, the inequality above holds if and only if%
\begin{eqnarray*}
x_{\infty }^{\ast }\left[ D\left( 0\right) -B_{\infty }\left( 0\right) %
\right] -x_{0}^{\ast }D\left( 0\right)  &>&\frac{x_{0}^{\ast 2}-x_{\infty
}^{\ast 2}}{2\alpha \left( 1-\Lambda \right) }\kappa \Leftrightarrow  \\
D\left( 0\right) -\frac{x_{\infty }^{\ast }}{x_{\infty }^{\ast }-x_{0}^{\ast
}}B_{\infty }\left( 0\right)  &>&-\frac{x_{0}^{\ast }+x_{\infty }^{\ast }}{%
2\alpha \left( 1-\Lambda \right) }\kappa \Leftrightarrow  \\
\frac{1}{1-\Lambda }-\frac{\Lambda }{\frac{\gamma }{\alpha }-\Lambda }
&>&\left( 1-\frac{1}{2}\frac{\frac{\gamma }{\alpha }-\Lambda }{1-\Lambda }%
\right) \frac{B_{\infty }\left( 0\right) }{D\left( 0\right) }\Leftrightarrow 
\end{eqnarray*}%
\begin{equation}
\frac{1}{1-\Lambda }-\frac{\Lambda }{\frac{\gamma }{\alpha }-\Lambda }%
>\left( 1-\frac{1}{2}\frac{\frac{\gamma }{\alpha }-\Lambda }{1-\Lambda }%
\right) \phi \frac{\tau }{r-g+\tau }.  \label{eq}
\end{equation}%
The left-hand side of (\ref{eq}) increases in $\frac{\gamma }{\alpha }$ and
the right-hand side decreases in $\frac{\gamma }{\alpha }$. At $\frac{\gamma 
}{\alpha }=1$ (\ref{eq}) holds. Thus, $\frac{\gamma }{\alpha }\geq
1\Rightarrow $ $u_{\infty }^{\ast }>u_{0}^{\ast }$. Notice that $\frac{%
\gamma }{\alpha }<1\Rightarrow 1-\frac{1}{2}\frac{\frac{\gamma }{\alpha }%
-\Lambda }{1-\Lambda }>0$ and $\frac{1}{1-\Lambda }-\frac{\Lambda }{\frac{%
\gamma }{\alpha }-\Lambda }>0\Leftrightarrow \frac{\gamma }{\alpha }%
>1-\left( 1-\Lambda \right) ^{2}$, where $1-\left( 1-\Lambda \right)
^{2}<2\Lambda $. Thus, if $2\Lambda <$ $\frac{\gamma }{\alpha }<1$ then $%
u_{\infty }^{\ast }>u_{0}^{\ast }$ if and only if%
\begin{equation*}
\phi <\phi ^{\ast }\equiv \frac{\tau +r-g}{\tau }\frac{\frac{1}{1-\Lambda }-%
\frac{\Lambda }{\frac{\gamma }{\alpha }-\Lambda }}{1-\frac{1}{2}\frac{\frac{%
\gamma }{\alpha }-\Lambda }{1-\Lambda }},
\end{equation*}%
where: $\phi ^{\ast }>0$, $\phi ^{\ast }$ increases in $\frac{\gamma }{%
\alpha }$, $\lim_{\frac{\gamma }{\alpha }\searrow 2\Lambda }\phi ^{\ast }=%
\frac{\Lambda }{\frac{2}{3}-\Lambda }>0$ for $2\Lambda <1$. Moreover, $\phi
^{\ast }$ decreases in $\tau $. Finally, note that $\phi ^{\ast }>\overline{%
\phi }$ in this range.
\end{proof}

\ \ 

\begin{proof}[\textbf{Proof of Proposition \protect\ref{selection_experiment}}]
According to the proof of Proposition \ref{optimal_gov}, when the founder is making the governance decision, firm $i$ adopts a
dual-class structure if and only if $\phi _{i}<\phi _{i}^{\ast }$. If instead shareholders  make the governance decision, then based on the proof of Proposition \protect\ref{price_dynamic}, they would choose a dual-class structure if and only if
\begin{equation*}
p_{0,i}^{\ast }\left( \infty \right) >p_{\infty ,i}^{\ast }\left( \infty \right)
\Leftrightarrow \phi_i <\overline{\phi_i }%
\equiv \frac{\frac{\gamma }{\alpha \Lambda }-2}{\frac{\gamma }{\alpha
\Lambda }-1}\frac{r-g+\tau }{\tau },
\end{equation*}%
and importantly, $\overline{\phi }$ decreases in $\tau$, which is the only property of the cutoff that we exploit for the proof. So all the arguments below hold for cases where shareholders make the governance decision.

Consider two
firms: firm $d$ chooses dual-class shares (i.e., $\phi _{d}<\phi
_{d}^{\ast }$), while firm $s$ chooses single-class shares (i.e., $\phi
_{s}>\phi _{s}^{\ast }$). We examine whether this selection pattern can
arise purely from one-dimensional heterogeneity, that is, whether firms can
differ along only a single dimension from the set $\left\{ g,r,\tau ,\phi
,\alpha ,\Lambda ,\gamma ,\kappa \right\} $ while still making different
governance choices. To test this, we conduct two counterfactual exercises.
In the first scenario, we assume dual-class firms are unexpectedly forced to
convert to single-class structure after their initial choice. In the second
scenario, we assume single-class firms are unexpectedly forced to adopt
dual-class structure. These exercises allow us to isolate whether observed
performance differences reflect causal effects of governance structure or
merely selection on firm characteristics.

\begin{enumerate}
\item \textbf{Dual-class firms are forced to be single-class firms. }In this
case, firm $d$ is forced to be single-class in spite of its choice to be
dual-class. Under this assumption, and according to Lemma \ref{res:effort}, the founder of firm $i\in \left\{ d,s\right\} $ chooses
the effort level%
\begin{equation*}
x_{\text{single,i}}^{\ast }\equiv x_{0,i}^{\ast }=\frac{\alpha _{i}\Lambda
_{i}D_{i}\left( 0\right) }{\kappa _{i}}=\frac{1}{\kappa _{i}}\frac{\alpha
_{i}\Lambda _{i}}{r_{i}-g_{i}},
\end{equation*}

and the share price at time $t$ is%
\begin{equation*}
p_{\text{single,i}}^{\ast }\left( t\right) \equiv p_{0,i}^{\ast }\left(
t\right) =x_{\text{single,i}}^{\ast }\frac{e^{g_{i}t}}{r_{i}-g_{i}}.
\end{equation*}%
To replicate the pattern in Proposition \ref{price_dynamic}, we ask whether
there exists $T^{\ast }\in \left( 0,\infty \right) $ such that $p_{\text{%
single,d}}^{\ast }\left( t\right) >p_{\text{single,s}}^{\ast }\left(
t\right) $ if $t<T^{\ast }$ and $p_{\text{single,d}}^{\ast }\left( t\right)
<p_{\text{single,s}}^{\ast }\left( t\right) $ otherwise.

An increase in any parameter from $\left\{ g,r,\tau ,\phi ,\alpha ,\Lambda
,\gamma ,\kappa \right\} $ affects $p_{\text{single,i}}^{\ast }\left(
t\right) $ in one of three ways: (i) it increases $p_{\text{single,i}}^{\ast
}\left( t\right) $ for all $t$, (as with $g$, $\alpha $, and $\Lambda $),
(ii) it decreases $p_{\text{single,i}}^{\ast }\left( t\right) $ for all $t$
(as with $r$ and $\kappa $), or (iii) it has no effect at all (as with $\tau
,$ $\phi ,$ and $\gamma $). Therefore, if firms differ along only one of
these dimensions, the share price ranking between dual-class and
single-class firms must be uniform across all time periods. Consequently, if
we observe performance reversals over time, such patterns cannot be
explained by one-dimensional heterogeneity alone.

\item \textbf{Single-class firms are forced to be dual-class firms. }In this
case, firm $s$ is forced to be dual-class in spite of its choice to be
single-class. Under this assumption, and according to Lemma \ref%
{res:effort}, the founder of firm $i\in \left\{ d,s\right\} $ chooses
the effort level%
\begin{eqnarray*}
x_{\text{dual,i}}^{\ast } &\equiv &x_{\infty ,i}^{\ast } =%
\frac{\alpha _{i}\Lambda _{i}D_{i}\left( 0\right) +\left( \gamma _{i}-\alpha
_{i}\Lambda _{i}\right) B_{n,i}\left( 0\right) }{\kappa _{i}} \\
&=&\frac{1}{\kappa _{i}}\frac{\alpha _{i}\Lambda _{i}}{r_{i}-g_{i}}+\frac{%
\left( \gamma _{i}-\alpha _{i}\Lambda _{i}\right) B_{n,i}\left( 0\right) }{%
\kappa _{i}} \\
&=&\frac{1}{\kappa _{i}}\frac{\alpha _{i}\Lambda _{i}}{r_{i}-g_{i}}+\left(
\gamma _{i}-\alpha _{i}\Lambda _{i}\right) \frac{\phi _{i}}{\kappa _{i}}%
\left[ \frac{1}{r_{i}-g_{i}}-\frac{1}{r_{i}-g_{i}+\tau _{i}}\right] \\
&=&\frac{1}{\kappa _{i}}\frac{\alpha _{i}\Lambda _{i}}{r_{i}-g_{i}}\left(
1+\phi _{i}\frac{\gamma _{i}-\alpha _{i}\Lambda _{i}}{\alpha _{i}\Lambda _{i}%
}\frac{\tau _{i}}{r_{i}-g_{i}+\tau _{i}}\right),
\end{eqnarray*}%
and the share price at time $t$ is%
\begin{eqnarray*}
p_{\text{dual,i}}^{\ast }\left( t\right) &\equiv &p_{\infty ,i}^{\ast
}\left( t\right) =x_{\text{dual,i}}^{\ast }\left[ D_{i}\left( t\right)
-B_{\infty ,i}\left( t\right) \right] \\
&=&x_{\text{dual,i}}^{\ast }\left[ \frac{e^{g_{i}t}}{r_{i}-g_{i}}-\phi
_{i}e^{g_{i}t}\left[ \frac{1}{r_{i}-g_{i}}-\frac{e^{-\tau _{i}t}}{%
r_{i}-g_{i}+\tau _{i}}\right] \right] \\
&=&\frac{1}{\kappa _{i}}\frac{\alpha _{i}\Lambda _{i}}{r_{i}-g_{i}}\left(
1+\phi _{i}\frac{\gamma _{i}-\alpha _{i}\Lambda _{i}}{\alpha _{i}\Lambda _{i}%
}\frac{\tau _{i}}{r_{i}-g_{i}+\tau _{i}}\right) \left( \frac{1-\phi _{i}}{%
r_{i}-g_{i}}+\frac{\phi _{i}e^{-\tau _{i}t}}{r_{i}-g_{i}+\tau _{i}}\right)
e^{g_{i}t}.
\end{eqnarray*}%
To replicate the pattern in Proposition \ref{price_dynamic}, we ask whether
there exists $T^{\ast }\in \left( 0,\infty \right) $ such that $p_{\text{%
dual,d}}^{\ast }\left( t\right) >p_{\text{dual,s}}^{\ast }\left( t\right) $
if $t<T^{\ast }$ and $p_{\text{dual,d}}^{\ast }\left( t\right) <p_{\text{%
dual,s}}^{\ast }\left( t\right) $ otherwise.

An increase in any parameter from $\left\{ r,g,\alpha ,\Lambda ,\gamma
,\kappa \right\} $ either increases $p_{\text{single,i}}^{\ast }\left(
t\right) $ for all $t$ (as with $g$, $\alpha $, $\Lambda $, and $\gamma $)
or decreases $p_{\text{single,i}}^{\ast }\left( t\right) $ for all $t$ (as
with $r$ and $\kappa $). Therefore, the selection effect cannot stem from
one-dimensional heterogeneity in these parameters.

\begin{itemize}
\item Consider the effect of parameter $\phi $. Recall firm $i$ chooses
dual-class shares if and only if $\phi _{i}<\phi _{i}^{\ast }$. Therefore,
it must be $\phi _{d}<\phi _{s}$. Therefore, $p_{\text{dual,d}}^{\ast }\left(
t\right) >p_{\text{dual,s}}^{\ast }\left( t\right) $ if and only if%
\begin{eqnarray*}
&&\frac{1}{\kappa }\frac{\alpha \Lambda }{r-g}\left( 1+\phi _{d}\frac{\gamma
-\alpha \Lambda }{\alpha \Lambda }\frac{\tau }{r-g+\tau }\right) \left( 
\frac{1-\phi _{d}}{r-g}+\frac{\phi _{d}e^{-\tau t}}{r-g+\tau }\right) e^{gt}
\\
&>&\frac{1}{\kappa }\frac{\alpha \Lambda }{r-g}\left( 1+\phi _{s}\frac{%
\gamma -\alpha \Lambda }{\alpha \Lambda }\frac{\tau }{r-g+\tau }\right)
\left( \frac{1-\phi _{s}}{r-g}+\frac{\phi _{s}e^{-\tau t}}{r-g+\tau }\right)
e^{gt}\Leftrightarrow \\
\frac{\phi _{s}-\phi _{d}}{r-g}-\frac{\phi _{s}-\phi _{d}}{r-g+\tau }%
e^{-\tau t} &>&\frac{\gamma -\alpha \Lambda }{\alpha \Lambda }\frac{\tau }{%
r-g+\tau }\left[ \frac{\phi _{s}-\phi _{d}-\left( \phi _{s}^{2}-\phi
_{d}^{2}\right) }{r-g}+\left( \phi _{s}^{2}-\phi _{d}^{2}\right) \frac{%
e^{-\tau t}}{r-g+\tau }\right] \Leftrightarrow \\
e^{-\tau t} &<&\frac{r-g+\tau }{r-g}\frac{1+\frac{\gamma -\alpha \Lambda }{%
\alpha \Lambda }\frac{\tau }{r-g+\tau }\left( \phi _{s}+\phi _{d}-1\right) }{%
1+\frac{\gamma -\alpha \Lambda }{\alpha \Lambda }\frac{\tau }{r-g+\tau }%
\left( \phi _{s}+\phi _{d}\right) }.
\end{eqnarray*}%
This means that either the inequality has the same sign for all $t$, or that
it holds only for large $t$. Either way, it does not replicate the dynamic
effect Proposition \ref{price_dynamic} describes.

\item Consider the effect of parameter $\tau $. Recall firm $i$ chooses
dual-class shares if and only if $\phi _{i}<\phi _{i}^{\ast }$, where $\phi _{i}^{\ast }$ decreases in $\tau$.
Therefore, it must be $\tau _{d}<\tau _{s}$. Notice%
\begin{equation*}
x_{\text{dual,i}}^{\ast }=\frac{1}{\kappa }\frac{\alpha \Lambda }{r-g}\left(
1+\phi \frac{\gamma -\alpha \Lambda }{\alpha \Lambda }\frac{\tau _{i}}{%
r-g+\tau _{i}}\right)
\end{equation*}%
increases in $\tau _{i}$, and thus, $x_{\text{dual,d}}^{\ast }<x_{\text{%
dual,s}}^{\ast }$. Moreover, notice%
\begin{equation*}
\frac{\partial B_{\infty }\left( t\right) }{\partial \tau }=-\phi e^{gt}%
\frac{\partial }{\partial \tau }\frac{e^{-\tau t}}{r-g+\tau }>0,
\end{equation*}%
which implies $B_{\infty ,d}\left( t\right) <B_{\infty ,s}\left( t\right) $
for all $t$.

Finally, we show by examples that it is possible to generate selection ($\phi ^{\ast }(\tau_s)<\phi<\phi^{\ast }(\tau_d)$) and the dynamic pattern in
Proposition \ref{price_dynamic} ($\underline{\phi }<\phi <\overline{\phi }%
\left( \tau _{s}\right) <\overline{\phi }\left( \tau _{d}\right) $). Suppose%
\begin{eqnarray*}
\left\{ \tau _{s},\tau _{d}\right\}  &=&\left\{ 0.3,0.03\right\}  \\
\left\{ g,r,\phi ,\alpha ,\Lambda ,\gamma \right\}  &=&\left\{ 0.02,0.03,%
0.85,0.5,0.05,0.1\right\}. 
\end{eqnarray*}%
Then $\frac{\gamma }{\alpha \Lambda }=\frac{0.1}{0.5\cdot 0.05}=4$, $%
\underline{\phi }=\frac{4-2}{4-1}-\frac{0.01}{0.03}\frac{1}{4-1}=\frac{5}{9}%
\approx 0.556$, $\overline{\phi }\left( \tau _{s}\right) =\frac{2}{3}\frac{%
0.01+0.3}{0.3}\approx 0.689$, $\overline{\phi }\left( \tau _{d}\right) =%
\frac{2}{3}\frac{0.01+0.03}{0.03}=\frac{8}{9}\approx 0.889$, $\phi ^{\ast }(\tau_s)=\frac{1271}{1575}\approx0.807$, and $\phi ^{\ast }(\tau_d)=\frac{328}{315}\approx1.041$. Under these parameters, firm $d$ (low $\tau _{d}$)
adopts dual-class and firm $s$ (high $\tau _{s}$) single-class shares, while the
forced-dual class price paths $p_{\text{dual},d}^{\ast }$ and $p_{\text{dual},s%
}^{\ast }$ cross once (at $t^{\ast }\approx 66.9$), replicating Proposition %
\ref{price_dynamic}.
\end{itemize}
\end{enumerate}
\end{proof}

    \clearpage

    \section{Variable Definitions} \label{appendix:variables}
This appendix provides definitions of firm-level variables used in the analysis. 

\begin{center}
\renewcommand{\arraystretch}{1.15}
\begin{tabular}{p{3.5cm} p{10.5cm}}
\toprule
\textbf{Variable} & \textbf{Definition} \\
\midrule

\textit{Log book assets} & Log of book assets. \\

\textit{Age} & Number of years since IPO, proxied by the first appearance in Compustat with a valid stock price. \\

\textit{Tobin's $q$} & Ratio of the market value of capital to the book value of capital. The market value is market equity plus book debt (proxy for market debt), and the book value is book equity plus deferred taxes plus book debt. Computing market equity for dual-class firms uses the observed prices for both classes when available and otherwise assumes a 5\% premium for superior-voting shares relative to inferior-voting shares.\\

\textit{Sales growth} & First difference of the log of sales. \\

\textit{ROA} & Operating income before depreciation divided by lagged book assets. \\

\textit{Market leverage} & Total debt divided by the sum of total debt and market equity. \\

\textit{Capex} & Capital expenditures divided by lagged book assets. \\

\textit{R\&D} & Research and development expenses divided by lagged book assets. \\

\textit{Tangibility} & Net property, plant, and equipment divided by book assets. \\

\textit{Payout ratio} & Total payout (dividends plus repurchases) divided by market equity. \\

\textit{IHS(citations)} & Inverse hyperbolic sine-transformed number of citations the firm's patents receive one, two, or three years ahead. \\

\textit{Cash flow} & Net income plus depreciation and amortization divided by lagged book assets. \\

\textit{Voting premium} & $(P_A - P_B)/(P_B - rP_A)$, where $P_A$ ($P_B$) is the price of the superior- (inferior-) voting shares and $r$ < 1 is the relative number of votes between the inferior- and superior-voting shares. \\

\textit{Log market equity} & Log of market equity. \\

\textit{Shapley value} & Shapley value of the firm ownership structure. \\

\textit{Log volume (sup. to inf.)} & Log of the ratio of trading volume between the superior- and inferior-voting shares. \\

\bottomrule
\end{tabular}
\end{center}
    \clearpage

\setcounter{table}{0} 
\renewcommand{\thetable}{\thesection\arabic{table}}
\renewcommand{\thefigure}{\thesection\arabic{figure}}
\section{Cross-Sectional Relation between Dual-Class Shares and Valuation} \label{appendix:average}

This appendix describes how we estimate the cross-sectional relation between dual-class shares and firm valuation, as measured by Tobin's $q$, building upon the literature that examines the relation between deviations from one share--one vote and firm valuation (e.g., \citet{cronqvist2003agency}; \citet{gompers2010extreme}). We estimate the following equation:
\begin{align} \label{eq:q}
    \textit{q}_{ijt} = \alpha_{jt} + \beta\textit{Dual}_{it} + \gamma'\textit{X}_{it} + \epsilon_{ijt},
\end{align}
where $\textit{q}_{ijt}$ is Tobin's \textit{q} for firm $i$ in three-digit SIC industry $j$ in year $t$; $\alpha_{jt}$ represents industry $j$ by year $t$ fixed effects; $\textit{Dual}_{it}$ is an indicator variable equal to one if firm $i$ has multiple classes of shares with different voting rights in year $t$, and zero otherwise; $\textit{X}_{it}$ is a vector of control variables including log book assets, market leverage, R\&D expenses scaled by lagged book assets, tangibility, sales growth, payout ratio, and ROA; and $\epsilon_{ijt}$ represents random errors clustered at the firm level. \looseness=-1
Table \ref{tab:appendix_c1} presents the results of estimating Eq. (\ref{eq:q}). Figure
\ref{fig:appendix_c1_q_dynamics} presents the results of estimating a version of Eq. (\ref{eq:q}) that includes \text{$\sum_{k=0}^{50} \mathbb{1}\!\left[\text{age}=k\right]$}, where $\mathbb{1}\!\left[\text{age}=k\right]$ is an indicator equal to one if firm age = $k$ (0 $\leq k$ $\leq$ 50), and zero otherwise, and their interactions with $\textit{Dual}_{it}$. We restrict the estimation sample to firms with ages less than or equal to 50 (sample size is too small to allow for precise estimation for ages over 50).

\begin{table}[]
    \centering
    \small 
    \captionsetup{font = large, justification=centering} 
    \caption{\textbf{Average Cross-Sectional Relation between Dual-Class Shares and Valuation}}
\begin{tabular}{lcccc}
    \toprule \toprule
        & (1) & (2)  \\ 
        \multicolumn{1}{c}{Dependent variable:} & \multicolumn{2}{c}{Tobin's $q$}  \\  
\midrule

Dual                &       0.309\sym{***}&       0.327\sym{***}\\
                    &      (0.080)         &      (0.074)         \\
\addlinespace
Log book assets     &          -           &       0.015\sym{**} \\
                    &          -           &      (0.007)         \\
\addlinespace
Market leverage     &          -          &      -1.968\sym{***}\\
                    &          -           &    (0.045)         \\
\addlinespace
R\&D                &          -           &       6.637\sym{***}\\
                    &          -           &     (0.230)         \\
\addlinespace
Tangibility         &          -           &      -0.323\sym{***}\\
                    &          -           &     (0.072)         \\
\addlinespace
Sales growth        &          -           &       0.141\sym{***}\\
                    &          -           &     (0.013)         \\
\addlinespace
Payout ratio        &          -           &      -2.452\sym{***}\\
                    &          -           &    (0.161)         \\
\addlinespace
ROA                 &          -           &       1.089\sym{***}\\
                    &          -           &     (0.088)         \\
\midrule
SIC3 $\times$ year fixed effects&           Y         &           Y         \\
$R^2$               &       0.218         &       0.310         \\
Observations        &     160,702         &     160,702         \\
\bottomrule \bottomrule
    \end{tabular}
\label{tab:appendix_c1}
\begin{tablenotes}
\small
\item \textit{Notes:} This table examines the average cross-sectional relation between dual-class shares and firm valuation as measured by Tobin’s $q$. ``Dual'' is an indicator variable equal to one if a firm-year has multiple classes of shares with different voting rights, and zero otherwise. Definitions of the other variables are as in Table \ref{tab:descriptive_stats}. Standard errors clustered at the firm level are reported in parentheses. *, **, and *** represent significance at the 10\%, 5\%, and 1\% levels, respectively.
\end{tablenotes}
\end{table}

\clearpage

\begin{figure}[t] 
    \centering
    \captionsetup{font = large, justification=centering}
    \caption{\\\textbf{Cross-Sectional Relation between Dual-Class Shares and Valuation over Maturity}}
    \includegraphics[width = \textwidth]{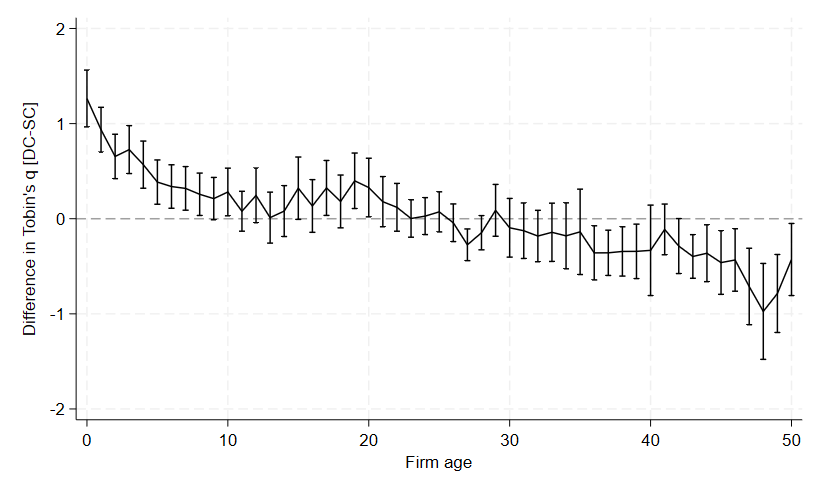}
    \captionsetup{font = small, justification=justified}
    \caption*{\textit{Notes}: This figure plots the differences in Tobin's \textit{q} between dual-class and single-class firms over their age. To construct the graph, we first estimate a version of Eq. (\ref{eq:q}) that includes \text{$\sum_{k=0}^{50} \mathbb{1}\!\left[\text{age}=k\right]$}, where $\mathbb{1}\!\left[\text{age}=k\right]$ is an indicator equal to one if firm age = $k$ (0 $\leq k$ $\leq$ 50), and zero otherwise, and their interactions with $\textit{Dual}$. The line plots the coefficients on $\textit{Dual} \times \mathbb{1}\!\left[\text{age}=k\right]$ along with 95\% confidence intervals (in vertical lines).}
    \label{fig:appendix_c1_q_dynamics}
\end{figure}{}
    \clearpage

\end{appendices}

\clearpage

\numberwithin{equation}{section} 
\renewcommand\thefigure{A.\arabic{figure}} 
\renewcommand\thetable{A.\arabic{table}} 
\setcounter{figure}{0} 
\setcounter{table}{0} 

\begin{appendices}
    \renewcommand{\appendixname}{Online Appendix}
    \renewcommand{\thesection}{\Alph{section}}
    \renewcommand{\thesubsection}{\thesection.\arabic{subsection}}
    \setcounter{section}{0}
    
    \setcounter{table}{0}
    \setcounter{figure}{0}
    \renewcommand{\thetable}{\Alph{section}\arabic{table}}
    \renewcommand{\thefigure}{\Alph{section}\arabic{figure}}

\clearpage
\pagenumbering{arabic}
\setcounter{page}{1}

\thispagestyle{empty}

\begin{center}

\vspace*{1.5cm}

{\LARGE\bfseries
Online Appendix
\par}

\vspace{0.5cm}

{\large for}

\vspace{0.4cm}

{\LARGE\bfseries
The Dynamic Trade-Off of Dual-Class Shares
\par}

\vspace{1.2cm}

{\large
\href{https://sites.google.com/view/hyunseobkim/}{Hyunseob Kim}\textsuperscript{1}
\qquad
\href{https://sites.google.com/view/doronlevit}{Doron Levit}\textsuperscript{2}
\qquad
\href{https://www.hkubs.hku.hk/people/roni-michaely/}{Roni Michaely}\textsuperscript{3}
}

\vspace{0.6cm}

\begin{tabular}{c}
\textsuperscript{1}Federal Reserve Bank of Chicago, Economic Research Department\\
\textsuperscript{2}University of Washington, Foster School of Business\\
\textsuperscript{3}The University of Hong Kong, Faculty of Business and Economics
\end{tabular}

\vspace{1.5cm}

\begin{minipage}{0.80\textwidth}
\centering
\small
This online appendix provides supplementary empirical analyses, additional figures and tables and theoretical results, and details of data collection and the matched-sample approach that complement the main paper.
\end{minipage}

\end{center}

\vfill
\setcounter{footnote}{0}
\clearpage
\setstretch{1}
\setcounter{page}{1}
\newpage
\newgeometry{left=1cm,bottom=2cm, right=1cm, top=1cm}
\setcounter{figure}{0} 
\renewcommand{\thefigure}{OA\arabic{figure}}
\renewcommand{\figurename}{Online Appendix Figure}

\clearpage

\begin{figure}
    \centering
    \captionsetup{font = large, justification=centering}
    \caption{\\\textbf{Model Prediction: Dynamic Treatment Effects of Dual-Class vs. Single-Class Shares on Firm Value}}
    \includegraphics[width=0.75\linewidth]{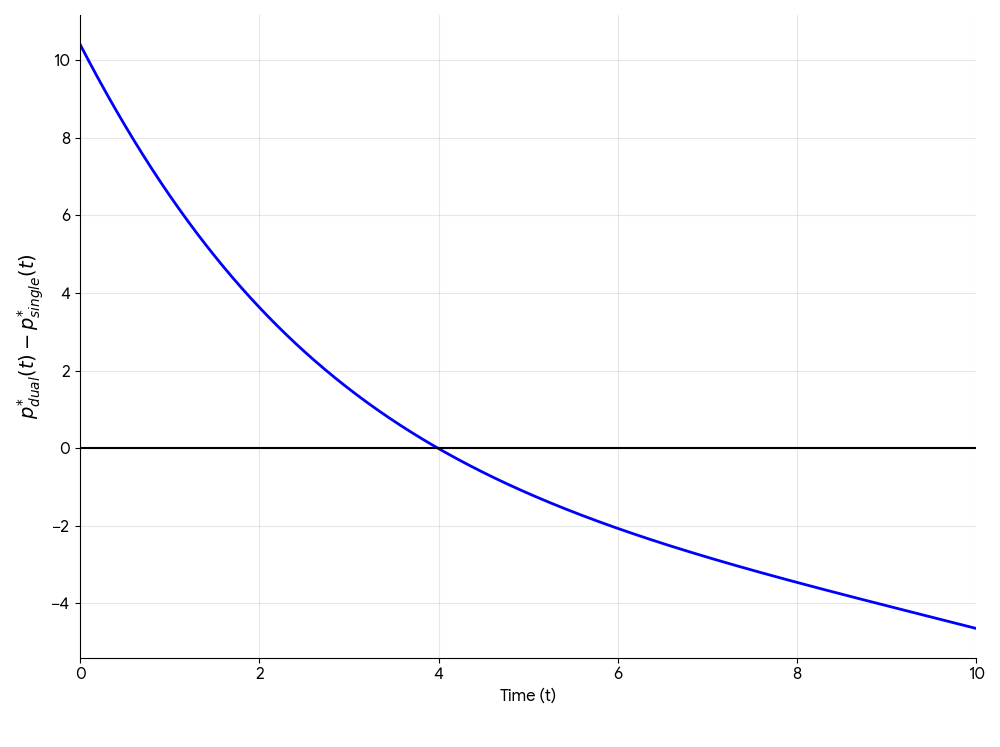}
    \captionsetup{font = small, justification=justified}    
    \caption*{\textit{Notes}: The figure plots how the value difference between two otherwise identical firms, one with dual-class shares and another with single-class shares, evolves over time, when the governance structures are randomly assigned (Proposition \ref{price_dynamic}). We use the following parameter values: $r = 0.1, g = 0.05, \tau = 0.5, \alpha = 0.5, \Lambda = 0.1, \kappa = 1, \gamma = 0.5, \phi = 0.9$.}
    \label{fig:prop1}
\end{figure}

\clearpage
\setstretch{1}

\newpage
\newgeometry{left=1cm,bottom=2cm, right=1cm, top=1cm}
\setcounter{table}{0} 
\renewcommand{\thetable}{OA\arabic{table}} 
\renewcommand{\tablename}{Online Appendix Table}


\begin{table}[!htbp] \centering
\captionsetup{font = large, justification=centering}
  \caption{\textbf{Descriptive Statistics on Matched Samples for Difference-in-Differences Analysis of Recapitalization: AMEX and Nasdaq vs. NYSE Firms Pre-1986}}
  \label{tab:nyse_balance}
\scalebox{1}{
\begin{tabular}{@{\extracolsep{5pt}} lccc} 
\\[-1.8ex]
\toprule 
\toprule \\[-1.8ex] 
\multicolumn{1}{c}{Variable} & Treated Mean & Control Mean & Diff. \\ 
\midrule \\[-1.8ex] 
Log book assets & 4.280 & 4.419 & -0.140  \\ 
Age & 9.500 & 9.591 & -0.091 \\ 
Tobin's $q$ & 1.742 & 1.920 & -0.178 \\ 
Market leverage & 0.260 & 0.187 & 0.072 \\ 
R\&D & 0.001 & 0.019 & -0.012 \\
Tangibility & 0.399 & 0.298 & 0.101 \\ 
Sales growth & 0.072 & 0.055 & 0.016 \\
Payout Ratio & 0.021 & 0.038 & -0.017 \\
ROA & 0.051 & 0.040 & 0.011 \\ 
Observations & 22 & 22 & - \\ 
 \\[-1.8ex]
\bottomrule 
\bottomrule \\[-1.8ex] 
\end{tabular} 
}

\begin{tablenotes}
\small
\item \textit{Notes:} This table presents descriptive statistics for treated and matched control groups for the subsample of recapitalization events used for the difference-in-differences analysis in Section \ref{sec:event_study_add}, measured one quarter prior to the recapitalization. The treated group consists of firms listed on the AMEX and Nasdaq that recapitalized between 1980 and September 1986, while the control group consists of matched firms listed on the NYSE that remain single-class during the same period. We construct the matched control group using firm age, (quarterly) log book assets, and sales growth. We additionally trim the matched pairs with the largest absolute differences in (quarterly) ROA and log book assets, whose exclusion would make the mean difference in ROA and log book assets statistically insignificant. Variables are defined as in Table \ref{tab:descriptive_stats}. The column ``Diff.'' shows differences between the means for treated and matched control groups of firms. *, **, and *** represents significance at the 10\%, 5\%, and 1\% levels, respectively, based on standard errors clustered at the match pair level.
\end{tablenotes}
\end{table} 

\clearpage

\begin{table}[]
    \small
    \captionsetup{font = large, justification=centering}
    \caption{\textbf{The Effects of Dual-Class Recapitalization on Tobin's $q$: \\ AMEX and Nasdaq vs. NYSE Firms Pre-1986 - Standard Clustered vs. Wild Cluster Bootstrapped $p$-Values}}
\begin{center}
    
\begin{tabular}{l*{2}{c}}
\toprule \toprule
                    
                    &\multicolumn{1}{c}{(1)}&\multicolumn{1}{c}{(2)}\\
                    \multicolumn{1}{c}{Dependent variable:} 
                    & \multicolumn{2}{c}{Tobin's $q$} \\
\midrule
Pre-event    &      0.031 &     0.027   \\
                    &     [0.799]   &    [0.825]   \\
                    &     \{0.807\}    &    \{0.831\}   \\
\addlinespace
$\tau =$ 0 to 4       &      0.073 &       0.077  \\
                    &     [0.394]   &     [0.374] \\
                    &     \{0.351\}   &     \{0.334\} 
                    \\
\addlinespace
$\tau =$ 5 to 11        &      0.271  &       0.285 \\
                    &     [0.059]   &     [0.050] \\
                    &     \{0.055\}   &     \{0.047\} 
                    \\
\midrule
Firm fixed effects &    Y        &   Y \\
Matched pair fixed effects &    N        &   Y \\
Year-quarter fixed effects &    Y        &   Y \\
Firm-level controls &    Y        &   Y \\
$R^2$                & 0.908 & 0.909  \\
Observations        &      470   &      470   \\
\bottomrule \bottomrule
\end{tabular}

\begin{tablenotes}
\small
\item \textit{Notes:} This table presents coefficient estimates from Table \ref{tab:event_study_nyse} as well as $p$-values from the standard clustering adjustments [in braces] and $p$-values from the wild cluster bootstrap \{in brackets\} (\cite{cameron2008bootstrap}), both  at the matched pair level.

\end{tablenotes}
\label{tab:event_study_nyse_bootstrap}
\end{center}
\end{table}

\clearpage

\begin{table}[]
    \centering
    \small
    \captionsetup{font = large, justification=centering}
    \caption{\textbf{Market Reactions to Dual-Class Recapitalizations and Stock Unifications Conditional on Maturity with Alternate Event Windows}}
\begin{tabular}{lcc}
                                        & \multicolumn{1}{l}{}    & \multicolumn{1}{l}{}   \\ 
\toprule \toprule
& (1) & (2) \\
\multicolumn{3}{c}{\textbf{Panel A: Cumulative abnormal return [--3, +2]}} \\
\midrule
Event: & \multicolumn{1}{c}{Dual-class recapitalization} & \multicolumn{1}{c}{Dual-class stock unification} \\ 
\midrule
$\mathbb{1}[\text{age $>$ 12}]$ & -0.034** & 0.047** \\
& (0.014) & (0.022) \\
Constant & 0.029*** & -0.000 \\
& (0.011) & (0.017) \\
\midrule
$R^2$ & 0.050 & 0.072 \\
Observations & 116 & 61 \\
\toprule
\multicolumn{3}{c}{\textbf{Panel B: Cumulative abnormal return [--2, +2]}} \\
\midrule
$\mathbb{1}[\text{age $>$ 12}]$ & -0.023** & 0.050** \\
& (0.012) & (0.021) \\
Constant & 0.017** & -0.004 \\
& (0.009) & (0.016) \\
\midrule
$R^2$ & 0.034 & 0.087 \\
Observations & 116 & 61 \\
\bottomrule\bottomrule
                                        & \multicolumn{1}{l}{}    & \multicolumn{1}{l}{}  
\end{tabular}
\begin{tablenotes}
\small
\item \textit{Notes:} This table examines robustness of results in Table \ref{tab:recaps_and_unifications} using alternative event windows. The dependent variable is the cumulative abnormal return for three days before to two days after (two days before to two days after) the announcement of a dual-class recapitalization and stock unification in Panel A (Panel B). ``$\mathbb{1}[\text{age $>$ 12}]$ '' is an indicator variable equal to one if firm age is greater than 12 years, the sample median for dual-class firms, and zero otherwise. White robust standard errors are reported in parentheses. *, **, and *** represent significance at the 10\%, 5\%, and 1\% levels, respectively.
\end{tablenotes}
\label{tab:recaps_and_unifications_robust}
\end{table}

\clearpage

\begin{table}[]
    \centering
    \small
    \captionsetup{font = large, justification=centering}
    \caption{\textbf{Robustness of Investment-\emph{q} Sensitivity to Different Sales Growth Cutoffs}}
\begin{tabular}{l*{5}{c}}
\toprule \toprule
                    &\multicolumn{1}{c}{(1)}&\multicolumn{1}{c}{(2)}&\multicolumn{1}{c}{(3)}&\multicolumn{1}{c}{(4)}\\
                    \multicolumn{1}{c}{Dependent variable:}&\multicolumn{4}{c}{Capex}\\
                    \multicolumn{1}{c}{Sales growth percentile:}
                        & \multicolumn{2}{c}{15\%}
                        & \multicolumn{2}{c}{25\%} \\
                    \multicolumn{1}{c}{Sample:} &\multicolumn{1}{c}{Young}&\multicolumn{1}{c}{Mature}&\multicolumn{1}{c}{Young}&\multicolumn{1}{c}{Mature}\\
\midrule
$ q$                &       0.003\sym{***}&       0.003\sym{***}&       0.004\sym{***}&       0.004\sym{***}\\
                    &     (0.001)         &     (0.001)         &     (0.000)         &     (0.000)         \\
\addlinespace
$q$ $\times$ Dual   &       0.004  &      -0.003\sym{**} &       0.000         &      -0.003\sym{***} \\
                    &     (0.003)         &     (0.001)         &     (0.001)         &     (0.001)         \\

\midrule
Firm $\times$ Dual fixed effects&           Y         &           Y         &           Y         &           Y         \\
Year $\times$ Dual fixed effects&           Y         &           Y         &           Y         &           Y         \\
Cash flow $\times$ Dual controls  &           Y         &           Y         &           Y         &           Y         \\
$R^2$               &       0.778         &       0.689         &       0.727         &       0.645         \\
Observations        &      13,657         &      9,862         &      21,035         &      18,922         \\
\midrule 
$q$ $\times$ Dual $\times$ (Mature -- Young)   &  \multicolumn{2}{c}{-0.006**} &  \multicolumn{2}{c}{-0.003*} \\
$p$-value   &  \multicolumn{2}{c}{0.040} &  \multicolumn{2}{c}{0.065} \\
\bottomrule \bottomrule
\end{tabular}
\begin{tablenotes}
\small
\item \textit{Notes:} This table examines robustness of results in Table \ref{tab:invest_q_sensitivity} using different sales growth cutoffs to define an estimation sample of firms that are likely near the disinvestment threshold, the $15^{th}$ and $25^{th}$ percentiles. All variables are defined in Table \ref{tab:invest_q_sensitivity}. Standard errors clustered at the firm level are in parentheses. *, **, and *** represent significance at the 10\%, 5\%, and 1\% levels, respectively.
\end{tablenotes}
    \label{tab:invest_q_sensitivity_robust}
    \end{table}
    
\clearpage


\begin{table}[htbp]
    \centering
    \small
    \captionsetup{font = large, justification=centering}
    \caption{\textbf{Voting Premium over Firm Maturity Excluding Market Equity Control}}
    \label{tab:voting_premium_robust}

    \centering
    {\def\sym#1{\ifmmode^{#1}\else\(^{#1}\)\fi}
    \begin{tabular}{lccc}
    \toprule\toprule

                        & (1) & (2) & (3) \\
    \multicolumn{1}{c}{Dependent variable:} 
                        & \multicolumn{3}{c}{Voting premium} \\ \midrule
    ${\Large \mathbb{1}}[\text{age in quartile 2}]$
                        & 0.038\sym{**}  & 0.024        & 0.028\sym{*}        \\
                        & (0.018)      & (0.015)      & (0.014)      \\
    ${\Large \mathbb{1}}[\text{age in quartile 3}]$
                        & 0.043\sym{*} & 0.039        & 0.048\sym{*}\\
                        & (0.022)      & (0.024)      & (0.027)      \\
    ${\Large \mathbb{1}}[\text{age in quartile 4}]$
                        & 0.057\sym{**}& 0.069\sym{*} & 0.088\sym{*}\\
                        & (0.025)      & (0.038)      & (0.049)      \\
    Shapley value
                        & 0.051\sym{**} & -0.001        & -0.003        \\
                        & (0.023)      & (0.021)      & (0.022)      \\
    Log volume (sup.to inf.)
                        & 0.007 & 0.012\sym{**}        & 0.010\sym{**}        \\
                        & (0.006)      & (0.005)      & (0.005)      \\
    \midrule
    Firm fixed effects   & N & Y & Y \\
    Year fixed effects   & Y & Y & Y \\
    $R^2$                & 0.062 & 0.521 & 0.522 \\
    Observations         & 1,182 & 1,182 & 1,182 \\
    \bottomrule\bottomrule
    \end{tabular}
    }

    \vspace{0.4cm}

    \begin{tablenotes}
    \small
    \item \textit{Notes:} This table presents the robustness of Table \ref{tab:voting_premium0} to the exclusion of market equity as a control. All variables are defined in Table \ref{tab:voting_premium0}. In column (2), we include year fixed effects using \cite{deaton1997}'s normalization. In column (3), we use an alternative restriction that normalizes one year fixed effect to zero. Standard errors clustered at the firm level are in parentheses. *, **, and *** represent significance at the 10\%, 5\%, and 1\% levels, respectively.
    \end{tablenotes}

\end{table}

\clearpage
\section{Additional predictions of the model}\label{appendix:theory_online}

\subsection{Unexpected recapitalization and unification}

Consider the following variant of the baseline model. There exists a
threshold $\overline{t}>0$ such that, for all $t<\overline{t},$ the firm
maintains its original governance structure $n$ and generates cash-flows $\overline{x}_{n}e^{gt}$, where $%
\overline{x}_{n}\geq 0$ denotes the founder's initial effort choice. At $t=%
\overline{t}$, the firm is randomly assigned a new governance
structure: $n=\overline{n}>\overline{t}$ (unexpected recapitalization) with probability $%
\xi $, or $n=0$ (unexpected unification) with probability $1-\xi $. Given
new governance structure, the founder then chooses a new effort level $x$,
incurring a cost of $\frac{1}{2}\kappa x^{2}$. From this point onward, the
firm's cash flow is $xe^{g(t-\overline{t})}$. The rate of private benefit
extraction, $b_{t}=\phi \left( 1-e^{-\tau t}\right) $, is independent of $%
\overline{t}$. Thus, the later the recapitalization or unification occurs
(i.e., the larger $\overline{t}$ is), the more mature the firm and the
greater the relative importance of agency costs at the time of the event.
The baseline model obtains as the special case $\overline{t}=0$. 

\begin{lemma}
\label{res:unexpected}Suppose $\gamma -\alpha \Lambda >\alpha \Lambda $ and $%
\phi \in (\underline{\phi },$\ $\overline{\phi })$, as in the conditions of
Proposition \ref{price_dynamic}. Then, there exist $0\leq \hat{t}_{L}<\hat{t}%
_{H}<\infty $ such that:

\begin{itemize}
\item[$\left( i\right) $] Consider dual-class firms without sunset
provisions:

\begin{itemize}
\item[$\left( a\right) $] The abnormal return following an unexpected
recapitalization (unification) is positive if and only if $\overline{t}<\hat{%
t}_{H}$ ($\overline{t}>\hat{t}_{H}$).

\item[$\left( b\right) $] The abnormal return following an unexpected
recapitalization (unification) is increasing in $\overline{t}$ if and only
if $\overline{t}<\hat{t}_{L}$ ($\overline{t}>\hat{t}_{L}$).
\end{itemize}

\item[$\left( ii\right) $] If $\overline{t}<\hat{t}_{L}$, the abnormal
return following an unexpected unification is larger for dual-class firms
with sunset provisions than for those without such provisions.
\end{itemize}
\end{lemma}

\begin{proof}
The share price at time $t<$ $\overline{t}$ under governance structure $n\in
\left\{ 0,\overline{n}\right\} $ is%
\begin{eqnarray*}
p_{n}\left( t,\overline{t}\right)  &=&\overline{x}_{n}\int_{t}^{\overline{t}%
}e^{-r\left( s-t\right) }e^{gs}\left( 1-b_{s}\cdot 1_{t<n}\right) ds \\
&&+e^{-r\left( \overline{t}-t\right) }\left[ \xi x_{\overline{n}}^{\ast
}\,\left( \overline{t}\right) \left( D\left( \overline{t},\overline{t}%
\right) -B_{\overline{n}}(\overline{t},\overline{t})\right) +\left( 1-\xi
\right) x_{0}^{\ast }\,\left( \overline{t}\right) D\left( \overline{t},%
\overline{t}\right) \right] 
\end{eqnarray*}%
where for $t\geq $ $\overline{t}\,$,%
\begin{eqnarray*}
D\left( t,\overline{t}\right)  &=&\frac{e^{g\left( t-\overline{t}\right) }}{%
r-g} \\
B_{\overline{n}}\left( t,\overline{t}\right)  &=&\phi e^{g\left( t-\overline{%
t}\right) }\left[ \frac{1-e^{-\left( r-g\right) \left( \overline{n}-t\right)
}}{r-g}-e^{-\tau t}\frac{1-e^{-\left( r-g+\tau \right) \left( \overline{n}%
-t\right) }}{r-g+\tau }\right] \cdot 1_{t<\overline{n}}
\end{eqnarray*}%
and for $n\in \left\{ 0,\overline{n}\right\} $,%
\begin{equation*}
x_{n}^{\ast }\left( \overline{t}\right) =\frac{\alpha \Lambda D\left( 
\overline{t},\overline{t}\right) +\left( \gamma -\alpha \Lambda \right)
B_{n}\left( \overline{t},\overline{t}\right) }{\kappa }.
\end{equation*}%
Notice, $x_{\overline{n}}^{\ast }\left( \overline{t}\right) >x_{0}^{\ast
}\left( \overline{t}\right) $, that is, the founder exerts more effort
following recapitalization.

If an unexpected recapitalization (unification) happens at time $\overline{t}
$, then the share price changes from $p_{n}\left( \overline{t},\overline{t}%
\right) $ to $x_{\overline{n}}^{\ast }\,\left( \overline{t}\right) \left(
D\left( \overline{t},\overline{t}\right) -B_{\overline{n}}(\overline{t},%
\overline{t})\right) $ ($x_{0}^{\ast }\,\left( \overline{t}\right) D\left( 
\overline{t},\overline{t}\right) $), and the price difference is given by $%
-\left( 1-\xi \right) \Delta _{\overline{n}}$ ($\xi \Delta _{\overline{n}}$%
), where%
\begin{eqnarray*}
\Delta _{\overline{n}} &\equiv &x_{0}^{\ast }\,\left( \overline{t}\right)
D\left( \overline{t},\overline{t}\right) -x_{\overline{n}}^{\ast }\,\left( 
\overline{t}\right) \left[ D\left( \overline{t},\overline{t}\right) -B_{%
\overline{n}}(\overline{t},\overline{t})\right]  \\
&=&\frac{\alpha \Lambda \frac{1}{r-g}}{\kappa }\frac{1}{r-g}-\frac{\alpha
\Lambda \frac{1}{r-g}+\left( \gamma -\alpha \Lambda \right) B_{\overline{n}}(%
\overline{t},\overline{t})}{\kappa }\left[ \frac{1}{r-g}-B_{\overline{n}}(%
\overline{t},\overline{t})\right]  \\
&=&\frac{\gamma -\alpha \Lambda }{\kappa }B_{\overline{n}}(\overline{t},%
\overline{t})\left[ B_{\overline{n}}(\overline{t},\overline{t})-\frac{1}{r-g}%
\frac{\frac{\gamma }{\alpha \Lambda }-2}{\frac{\gamma }{\alpha \Lambda }-1}%
\right] .
\end{eqnarray*}%
Notice that $\Delta _{\overline{n}}>0$ if and only if $B_{\overline{n}}(%
\overline{t},\overline{t})>\frac{1}{r-g}\frac{\frac{\gamma }{\alpha \Lambda }%
-2}{\frac{\gamma }{\alpha \Lambda }-1}$.  

Suppose $\overline{n}=\infty $. Notice $B_{\infty }(\overline{t},\overline{t}%
)=\phi \left[ \frac{1}{r-g}-e^{-\tau \overline{t}}\frac{1}{r-g+\tau }\right] 
$ is increasing in $\overline{t}$ from $\phi \frac{1}{r-g}\frac{\tau }{%
r-g+\tau }$ to $\phi \frac{1}{r-g}$. Therefore, $\Delta _{\infty }$
increases in $\overline{t}$ if and only if $2B_{\infty }(\overline{t},%
\overline{t})-\frac{1}{r-g}\frac{\frac{\gamma }{\alpha \Lambda }-2}{\frac{%
\gamma }{\alpha \Lambda }-1}>0$. Notice%
\begin{eqnarray*}
B_{\infty }\left( 0,0\right)  &>&\frac{1}{r-g}\frac{\frac{\gamma }{\alpha
\Lambda }-2}{\frac{\gamma }{\alpha \Lambda }-1}>0\Leftrightarrow \phi >\text{%
\ }\overline{\phi } \\
\lim_{\overline{t}\rightarrow \infty }B_{\infty }\left( \overline{t},%
\overline{t}\right)  &<&\frac{1}{r-g}\frac{\frac{\gamma }{\alpha \Lambda }-2%
}{\frac{\gamma }{\alpha \Lambda }-1}>0\Leftrightarrow \phi <\underline{\phi }%
.
\end{eqnarray*}%
Thus, under the conditions of Proposition \ref{price_dynamic}, there exist $%
0\leq \hat{t}_{L}<\hat{t}_{H}<\infty $ such that%
\begin{eqnarray*}
B_{\infty }(\overline{t},\overline{t})-\frac{1}{r-g}\frac{\frac{\gamma }{%
\alpha \Lambda }-2}{\frac{\gamma }{\alpha \Lambda }-1} &>&0\Leftrightarrow 
\overline{t}>\hat{t}_{H} \\
2B_{\infty }(\overline{t},\overline{t})-\frac{1}{r-g}\frac{\frac{\gamma }{%
\alpha \Lambda }-2}{\frac{\gamma }{\alpha \Lambda }-1} &>&0\Leftrightarrow 
\overline{t}>\hat{t}_{L},
\end{eqnarray*}%
as required. 

Finally, notice 
\begin{equation*}
B_{\overline{n}}(\overline{t},\overline{t})=\phi \left[ \frac{1-e^{-\left(
r-g\right) \left( \overline{n}-\overline{t}\right) }}{r-g}-e^{-\tau 
\overline{t}}\frac{1-e^{-\left( r-g+\tau \right) \left( \overline{n}-%
\overline{t}\right) }}{r-g+\tau }\right] 
\end{equation*}%
and 
\begin{equation*}
\frac{\partial B_{\overline{n}}(\overline{t},\overline{t})}{\partial 
\overline{n}}=\phi \left( 1-e^{-\tau \overline{n}}\right) e^{-\left(
r-g\right) \left( \overline{n}-\overline{t}\right) }>0.
\end{equation*}%
Therefore, $B_{\infty }(\overline{t},\overline{t})>B_{\overline{n}}(%
\overline{t},\overline{t})>0$ and%
\begin{equation*}
\Delta _{\infty }>\Delta _{\overline{n}}\Leftrightarrow B_{\infty }(%
\overline{t},\overline{t})+B_{\overline{n}}(\overline{t},\overline{t})>\frac{%
1}{r-g}\frac{\frac{\gamma }{\alpha \Lambda }-2}{\frac{\gamma }{\alpha
\Lambda }-1}.
\end{equation*}%
Thus, if $\overline{t}<\hat{t}_{L}$ then%
\begin{equation*}
B_{\overline{n}}(\overline{t},\overline{t})<B_{\infty }(\overline{t},%
\overline{t})<B_{\infty }(\overline{t},\overline{t})+B_{\overline{n}}(%
\overline{t},\overline{t})<2B_{\infty }(\overline{t},\overline{t})<\frac{1}{%
r-g}\frac{\frac{\gamma }{\alpha \Lambda }-2}{\frac{\gamma }{\alpha \Lambda }%
-1},
\end{equation*}%
and hence, $\Delta _{\infty }<\Delta _{\overline{n}}<0$
\end{proof}

\subsection{Voting premium\label{app:vp}}

We define the voting premium under dual-class shares as the difference
between the founder's and shareholders' intrinsic valuations of the firm's
shares.

\begin{lemma}
\label{res:vp}The dual-class voting premium is given by%
\begin{equation}
\frac{\gamma }{\alpha \Lambda }x_{\infty }^{\ast }B_{\infty }(t),
\end{equation}
which is positive and increasing in $t$.
\end{lemma}

\begin{proof}
Shareholder value at time $t$ is given by $p_{\infty }^{\ast }(t)=x_{\infty
}^{\ast }\left[ D(t)-B_{\infty }(t)\right] $. The expected value to the
founder at period $t\geq 0$ gross of the proceeds from selling shares and
the cost of effort is%
\begin{align*}
x_{\infty }^{\ast }\left[ \alpha \Lambda D(t)+(\gamma -\alpha \Lambda
)B_{\infty }(t)\right] & =\alpha \Lambda x_{\infty }^{\ast }\left[
D(t)-B_{\infty }(t)\right] +\gamma x_{\infty }^{\ast }B_{\infty }(t) \\
& =\alpha \Lambda \left[ p_{\infty }^{\ast }(t)+\frac{\gamma }{\alpha
\Lambda }x_{\infty }^{\ast }B_{\infty }(t)\right]. 
\end{align*}%
Thus, the voting premium is $\frac{\gamma }{\alpha \Lambda }x_{\infty
}^{\ast }B_{\infty }(t)$. Since $B_{\infty }(t)$ increases in $t$, so does
the voting premium.
\end{proof}

\subsection{Effect of $n$ on the founder's utility}\label{app:effect_n}

\begin{lemma}
\label{res:effect_n} $u_{n}^{\ast }$ increases in $n$ if $\frac{\gamma }{\alpha }>2-\Lambda $.
\end{lemma}

\begin{proof}
Recall that $B_{n}\left( 0\right) $ increases in $n$ and
that $n$ affects $u_{n}^{\ast }$ only through $B_{n}\left( 0\right) $;
it therefore suffices to show that $u_{n}^{\ast }$ increases in $B_{n}\left(
0\right) $. Note that
\begin{eqnarray*}
u_{n}^{\ast } &=&\alpha \left( 1-\Lambda \right) p_{n}^{\ast }\left(
0\right) +\frac{1}{2}\kappa x_{n}^{\ast 2} \\
&=&x_{n}^{\ast }\left[ \alpha \left( 1-\Lambda \right) \left( D\left(
0\right) -B_{n}\left( 0\right) \right) +\frac{\alpha \Lambda D\left(
0\right) +\left( \gamma -\alpha \Lambda \right) B_{n}\left( 0\right) }{2}%
\right] .
\end{eqnarray*}%
The first factor, $x_{n}^{\ast }$, increases in $B_{n}\left( 0\right) $
since $\gamma >\alpha \Lambda $. The bracketed term increases in $%
B_{n}\left( 0\right) $ if and only if $-\alpha \left( 1-\Lambda \right) +%
\frac{\gamma -\alpha \Lambda }{2}>0\Leftrightarrow \frac{\gamma }{\alpha }%
>2-\Lambda $. Hence, if $\frac{\gamma }{\alpha }>2-\Lambda $, both
factors increase in $B_{n}\left( 0\right) $ and $u_{n}^{\ast }$
increases in $n$, as required.
\end{proof}

\clearpage
\section{Identifying Dual-Class Firms, 1971--2022}\label{appendix:data}

We start by identifying candidate dual-class firms. We compare a given firm's number of shares outstanding obtained from CRSP and Compustat. CRSP provides the number of shares outstanding at the security level (i.e., for each class of shares), whereas Compustat provides the corresponding number at the firm level (i.e., the sum across all classes of shares). Thus, a significant difference between the two numbers indicates that the firm might have multiple classes of shares. If the numbers of shares from CRSP and Compustat differ by more than 2\%, we place those firm-years into a candidate set (\cite*{gompers2010extreme}). We supplement this set from a dataset of dual-class IPOs from Jay Ritter's website. We then hand-check whether firms in this candidate sample have multiple share classes and for confirmed dual-class firms, we collect information on the number of votes and shares outstanding with two data sources. For 1994--2022, we use the annual report (Form 10-K) and proxy statement (DEF 14A) taken from the SEC's EDGAR database for each firm-year, except for those covered by the \cite{gompers2010extreme} dataset. The SEC EDGAR does not provide information in electronic form for 1971--93, and so we use \emph{Moody's Manuals} (the Capital Stock section) during the period.

A notable fraction of firms with more than one class of common shares has the same number of votes across classes (e.g., one vote per share). We therefore determine whether these firms have different voting rights between share classes by manually examining security filings from the EDGAR (such as the DEF 14A, 424B, S-1, 10-K/Q) and \emph{Moody's}. We find that there are three possible reasons that these firms have multiple share classes, the first two of which represent differential voting rights: 

(1)	Different classes have differing voting rights for director election. In a typical case in this category, one class has the right to elect two-thirds of directors, and the other class, one-third. For these cases, we define the class with greater director election right as “superior.” There are 435 dual-class firm-years in this category.

(2)	Some firms use specific formulas to calculate the numbers of votes for different classes. A typical example involves a ``superior'' class of common stocks with the number of votes per share equal to the number of ``holdings units'' in a limited liability company that a small group of shareholders own. An ``inferior'' class typically carries one vote per share. These cases are rare, with only 24 firm-years involved.

(3)	The final category includes cases for which a dual-class structure appears to be set up for reasons other than giving different voting rights. For example, Triple-S Management has issued Class B common stocks as a capital asset for tax purposes. In other cases, non-US firms restrict ownership of one class of common stocks to citizens of specific countries.\footnote{For example, Grupo Iusacell, S.A. de C.V. restricts its Class B common stock ownership to Mexican citizens only, whereas their Class A common stock has no ownership restrictions.} Given that these cases do not represent a deviation from one share--one vote, we define them as ``non-dual-class'' and drop them from the analysis. This category includes 202 firm-years. \looseness=-1

\clearpage
\section{Constructing Matched Samples for Difference-in-Diffe
-rences Event Studies}\label{appendix:matching}
\setcounter{figure}{0}

In Section \ref{sec:q_robustness}, we perform the event studies with matched samples. We match each treated firm with a control firm with the smallest Mahalanobis distance over a parsimonious set of variables that are not balanced between treated and control groups (see Panel A of Online Appendix Tables \ref{tab:recap_balance_q}, \ref{tab:unif_balance_q}, \ref{tab:recap_balance_cites}, and \ref{tab:unif_balance_cites}). For recapitalizations, we match on log book assets and an indicator variable for positive R\&D expenses, while for unifications we match on age and sales growth. We require that the distance between the treated and matched firms is within the $60^{th}$ percentile of distance among all potential control firms. We also require that the matched firm be in the same industry as the treated firm---either the same three-digit SIC code, or the same two-digit SIC code if we cannot find a match with the three-digit. Finally, we trim the matched pairs with the largest absolute differences in selected variables to achieve full balance across covariates: for recapitalizations, we trim based on ROA and age, while for unifications based on leverage. Panel B of Online Appendix Tables \ref{tab:recap_balance_q}-\ref{tab:unif_balance_cites} shows the resulting balance for the matched samples. Using the matched samples, we estimate a version of Eq. (\ref{eq:dd}) in which we further include matched pair fixed effects and cluster standard errors at the matched pair level (\cite{Abadie_Spiess2022}). 

\newpage{}

\setcounter{table}{0}
\renewcommand{\thetable}{O\thesection\arabic{table}}
\begin{table}[!htbp] \centering
  \caption{\textbf{Descriptive Statistics on Baseline and Matched Samples for Difference-in-Differences Analysis of Recapitalization: Tobin's \emph{q}}} 
  \label{tab:recap_balance_q} 
\scalebox{1}{
\begin{tabular}{@{\extracolsep{5pt}} lccc} 
\\[-1.8ex]
\toprule 
\toprule \\[-1.8ex] 
\multicolumn{1}{c}{Variable} &
\multicolumn{1}{c}{Treated Mean} &
\multicolumn{1}{c}{Control Mean} &
\multicolumn{1}{c}{Diff.} \\
\midrule \\[-1.8ex] 
  \multicolumn{4}{c}{\textbf{Panel A: Baseline Sample}}   \\ 
\midrule \\[-1.8ex] 
Log book assets & 5.213 & 5.774 & -0.561\textasteriskcentered \textasteriskcentered \textasteriskcentered  \\ 
Age & 12.901 & 13.647 & -0.747 \\ 
Tobin's \emph{q} & 1.825 & 1.975 & -0.150 \\ 
Market leverage & 0.287 & 0.258 & 0.029\textasteriskcentered \textasteriskcentered  \\ 
R\&D & 0.017 & 0.034 & -0.017\textasteriskcentered \textasteriskcentered \textasteriskcentered  \\ 
Tangibility & 0.330 & 0.332 & -0.002 \\ 
Sales growth & 0.221 & 0.158 & 0.063 \\ 
Payout ratio & 0.022 & 0.028 & -0.006\textasteriskcentered \textasteriskcentered  \\ 
ROA & 0.181 & 0.141 & 0.041\textasteriskcentered \textasteriskcentered \textasteriskcentered  \\ 
Observations & 171 & 84,152 & - \\ 
\midrule \\[-1.8ex] 
   \multicolumn{4}{c}{\textbf{Panel B: Matched Sample}}   \\ 
\midrule \\[-1.8ex] 
Log book assets & 5.182 & 5.208 & -0.026 \\ 
Age & 12.529 & 11.471 & 1.057 \\ 
Tobin's \emph{q} & 1.795 & 1.801 & -0.006 \\ 
Market leverage & 0.293 & 0.260 & 0.033 \\ 
R\&D & 0.017 & 0.016 & 0.002 \\ 
Tangibility & 0.330 & 0.338 & -0.008 \\ 
Sales growth & 0.168 & 0.153 & 0.015 \\ 
Payout ratio & 0.023 & 0.027 & -0.004 \\ 
ROA & 0.177 & 0.162 & 0.015 \\ 
Observations & 157 & 157 & - \\ 
 \\[-1.8ex]
\bottomrule 
\bottomrule \\[-1.8ex] 
\end{tabular} 
}
\begin{tablenotes}
\small
\item \textit{Notes:} This table presents descriptive statistics for treated and control groups for recapitalization events used in the difference-in-differences analysis for Tobin's $q$, measured one year prior to the recapitalization. Panels A and B show for the baseline sample and the matched sample, respectively. In both samples, the treated group consists of firms that recapitalize with dual-class shares. The control group consists of firms that remain single-class (Panel A) and firms that remain single-class with matched characteristics (Panel B). Variables are defined as in Table \ref{tab:descriptive_stats}. The column ``Diff.'' shows differences between the means for treated and control groups of firms. *, **, and *** represent significance at the 10\%, 5\%, and 1\% levels, respectively, based on standard errors clustered at the firm (Panel A) and matched pair (Panel B) levels.
\end{tablenotes}
\end{table} 


\begin{table}[!htbp] \centering 
  \caption{\textbf{Descriptive Statistics on Baseline and Matched Samples for Difference-in-Differences Analysis of Unification: Tobin's \emph{q}}} 
  \label{tab:unif_balance_q} 
\scalebox{1}{
\begin{tabular}{@{\extracolsep{5pt}} lccc} 
\\[-1.8ex]
\toprule 
\toprule \\[-1.8ex] 
\multicolumn{1}{c}{Variable} &
\multicolumn{1}{c}{Treated Mean} &
\multicolumn{1}{c}{Control Mean} &
\multicolumn{1}{c}{Diff.} \\
\midrule \\[-1.8ex] 
   \multicolumn{4}{c}{\textbf{Panel A: Baseline Sample}}   \\ 
\midrule \\[-1.8ex]
Log book assets & 6.197 & 6.432 & -0.235 \\ 
Age & 12.514 & 16.761 & -4.248\textasteriskcentered \textasteriskcentered \textasteriskcentered  \\ 
Tobin's \emph{q} & 2.588 & 2.171 & 0.416* \\ 
Market leverage & 0.291 & 0.283 & 0.008 \\ 
R\&D & 0.023 & 0.019 & 0.004 \\ 
Tangibility & 0.328 & 0.294 & 0.034\textasteriskcentered  \\ 
Sales growth & 0.187 & 0.125 & 0.062\textasteriskcentered \textasteriskcentered  \\ 
Payout ratio & 0.028 & 0.029 & -0.001 \\ 
ROA & 0.139 & 0.141 & -0.002 \\ 
Observations & 185 & 6,539 & - \\ 
\midrule \\[-1.8ex] 
   \multicolumn{4}{c}{\textbf{Panel B: Matched Sample}}   \\ 
\midrule \\[-1.8ex] 
Log book assets & 6.222 & 6.317 & -0.095 \\ 
Age & 12.650 & 12.825 & -0.175 \\ 
Tobin's \emph{q} & 2.740 & 2.386 & 0.354 \\ 
Market leverage & 0.293 & 0.254 & 0.039 \\ 
R\&D & 0.026 & 0.034 & -0.008 \\ 
Tangibility & 0.315 & 0.290 & 0.025 \\ 
Sales growth & 0.182 & 0.136 & 0.046 \\ 
Payout ratio & 0.028 & 0.025 & 0.003 \\ 
ROA & 0.132 & 0.132 & -0.001 \\ 
Observations & 160 & 160 & - \\ 
 \\[-1.8ex]
\bottomrule 
\bottomrule \\[-1.8ex] 
\end{tabular} 
}
\begin{tablenotes}
\small
\item \textit{Notes:} This table presents descriptive statistics for treated and control groups for unification events used in the difference-in-differences analysis for Tobin's $q$, measured one year prior to the unification. Panels A and B show for the baseline sample and the matched sample, respectively. In both samples, the treated group consists of firms that unify share classes. The control group consists of firms that remain dual-class (Panel A) and firms that remain dual-class with matched characteristics (Panel B). Variables are defined as in Table \ref{tab:descriptive_stats}. The column ``Diff.'' shows differences between the means for treated and control groups of firms. *, **, and *** represent significance at the 10\%, 5\%, and 1\% levels, respectively, based on standard errors clustered at the firm (Panel A) and matched pair (Panel B) levels.
\end{tablenotes}
\end{table} 

\newpage{}

\begin{table}[!htbp] \centering
  \caption{\textbf{Descriptive Statistics on Baseline and Matched Samples for Difference-in-Differences Analysis of Recapitalization: Innovative Output}} 
  \label{tab:recap_balance_cites} 
\scalebox{1}{
\begin{tabular}{@{\extracolsep{5pt}} lccc} 
\\[-1.8ex]
\toprule 
\toprule \\[-1.8ex] 
\multicolumn{1}{c}{Variable} &
\multicolumn{1}{c}{Treated Mean} &
\multicolumn{1}{c}{Control Mean} &
\multicolumn{1}{c}{Diff.} \\
\midrule \\[-1.8ex] 
  \multicolumn{4}{c}{\textbf{Panel A: Baseline Sample}}   \\ 
\midrule \\[-1.8ex] 
Log book assets & 5.264 & 5.836 & -0.572\textasteriskcentered \textasteriskcentered \textasteriskcentered  \\ 
Age & 12.823 & 13.854 & -1.032 \\ 
Tobin's \emph{q} & 1.906 & 2.004 & -0.098 \\ 
Market leverage & 0.295 & 0.261 & 0.034\textasteriskcentered \textasteriskcentered  \\ 
R\&D & 0.021 & 0.035 & -0.013\textasteriskcentered \textasteriskcentered \textasteriskcentered  \\ 
Tangibility & 0.337 & 0.328 & 0.009 \\ 
Sales growth & 0.231 & 0.158 & 0.074 \\ 
Payout ratio & 0.023 & 0.028 & -0.005\textasteriskcentered  \\ 
ROA & 0.178 & 0.138 & 0.040\textasteriskcentered \textasteriskcentered  \\ 
IHS(citations) & 1.763 & 2.362 & -0.599 \\ 
Observations & 186 & 91,985 & - \\ 
\midrule \\[-1.8ex] 
   \multicolumn{4}{c}{\textbf{Panel B: Matched Sample}}   \\ 
\midrule \\[-1.8ex] 
Log book assets& 5.240 & 5.256 & -0.016 \\ 
Age & 12.676 & 11.585 & 1.091 \\ 
Tobin's \emph{q} & 1.806 & 1.759 & 0.046 \\ 
Market leverage & 0.294 & 0.261 & 0.033 \\ 
R\&D & 0.020 & 0.021 & -0.000 \\ 
Tangibility & 0.338 & 0.341 & -0.003 \\ 
Sales growth & 0.165 & 0.147 & 0.018 \\ 
Payout ratio & 0.024 & 0.026 & -0.003 \\ 
ROA & 0.177 & 0.158 & 0.020 \\ 
IHS(citations) & 1.769 & 2.058 & -0.289 \\ 
Observations & 176 & 176 & - \\ 
 \\[-1.8ex]
\bottomrule 
\bottomrule \\[-1.8ex] 
\end{tabular} 
}
\begin{tablenotes}
\small
\item \textit{Notes:} This table presents descriptive statistics for treated and control groups for recapitalization events used in the difference-in-differences analysis for patent citations, measured one year prior to the recapitalization. Panels A and B show for the baseline sample and the matched sample, respectively. In both samples, the treated group consists of firms that recapitalize with dual-class shares. The control group consists of firms that remain single-class (Panel A) and firms that remain single-class with matched characteristics (Panel B). Variables are defined as in Table \ref{tab:descriptive_stats}. The column ``Diff.'' shows differences between the means for treated and control groups of firms. *, **, and *** represent significance at the 10\%, 5\%, and 1\% levels, respectively, based on standard errors clustered at the firm (Panel A) and matched pair (Panel B) levels.
\end{tablenotes}
\end{table} 


\begin{table}[!htbp] \centering 
  \caption{\textbf{Descriptive Statistics on Baseline and Matched Samples for Difference-in-Differences Analysis of Unification: Innovative Output}} 
  \label{tab:unif_balance_cites}
\scalebox{1}{
\begin{tabular}{@{\extracolsep{5pt}} lccc} 
\\[-1.8ex]
\toprule 
\toprule \\[-1.8ex] 
\multicolumn{1}{c}{Variable} &
\multicolumn{1}{c}{Treated Mean} &
\multicolumn{1}{c}{Control Mean} &
\multicolumn{1}{c}{Diff.} \\
\midrule \\[-1.8ex] 
   \multicolumn{4}{c}{\textbf{Panel A: Baseline Sample}}   \\ 
\midrule \\[-1.8ex]
Log book assets & 6.263 & 6.472 & -0.209 \\ 
Age & 11.746 & 16.720 & -4.974\textasteriskcentered \textasteriskcentered \textasteriskcentered  \\ 
Tobin's \emph{q} & 2.626 & 2.190 & 0.436\textasteriskcentered \textasteriskcentered  \\ 
Market leverage & 0.321 & 0.285 & 0.036\textasteriskcentered  \\ 
R\&D & 0.026 & 0.021 & 0.005 \\ 
Tangibility & 0.323 & 0.288 & 0.035\textasteriskcentered \textasteriskcentered  \\ 
Sales growth & 0.187 & 0.129 & 0.058\textasteriskcentered \textasteriskcentered  \\ 
Payout ratio & 0.026 & 0.029 & -0.003 \\ 
ROA & 0.126 & 0.140 & -0.014 \\ 
IHS(citations) & 2.027 & 1.875 & 0.152 \\ 
Observations & 224 & 7,243 & - \\ 
\midrule \\[-1.8ex] 
   \multicolumn{4}{c}{\textbf{Panel B: Matched Sample}}   \\ 
\midrule \\[-1.8ex] 
Log book assets & 6.260 & 6.377 & -0.117 \\ 
Age & 11.968 & 12.126 & -0.158 \\ 
Tobin's \emph{q} & 2.751 & 2.620 & 0.131 \\ 
Market leverage & 0.302 & 0.269 & 0.033 \\ 
R\&D & 0.030 & 0.033 & -0.003 \\ 
Tangibility & 0.309 & 0.287 & 0.021 \\ 
Sales growth & 0.185 & 0.146 & 0.040 \\ 
Payout ratio & 0.026 & 0.021 & 0.006 \\ 
ROA & 0.120 & 0.136 & -0.016 \\ 
IHS(citations) & 1.769 & 2.058 & -0.289 \\ 
Observations & 190 & 190 & - \\ 
 \\[-1.8ex]
\bottomrule 
\bottomrule \\[-1.8ex] 
\end{tabular} 
}
\begin{tablenotes}
\small
\item \textit{Notes:} This table presents descriptive statistics for treated and control groups for unification events used in the difference-in-differences analysis for patent citations, measured one year prior to the unification. Panels A and B show for the baseline sample and the matched sample, respectively. In both samples, the treated group consists of firms that unify share classes. The control group consists of firms that remain dual-class (Panel A) and firms that remain dual-class with matched characteristics (Panel B). Variables are defined as in Table \ref{tab:descriptive_stats}. The column ``Diff.'' shows differences between the means for treated and control groups of firms. *, **, and *** represent significance at the 10\%, 5\%, and 1\% levels, respectively, based on standard errors clustered at the firm (Panel A) and matched pair (Panel B) levels.
\end{tablenotes}
\end{table}

\end{appendices}

\end{document}